\documentclass[journal, twoside]{IEEEtran}
\usepackage{cite}
\usepackage{amsmath,amssymb,amsfonts}
\usepackage{mathtools}
\usepackage{algorithmic}
\usepackage{graphicx}
\usepackage{textcomp}
\usepackage{url}
\def\BibTeX{{\rm B\kern-.05em{\sc i\kern-.025em b}\kern-.08em
    T\kern-.1667em\lower.7ex\hbox{E}\kern-.125emX}}

\usepackage[utf8]{inputenc}
\usepackage{tikz}
\usepackage{pgfplots}
\usepackage{siunitx}
\DeclareUnicodeCharacter{2212}{−}
\usepgfplotslibrary{groupplots, dateplot}
\usetikzlibrary{patterns,shapes.arrows,arrows.meta}
\pgfplotsset{compat=newest}
\newlength\axisheight
\newlength\axiswidth

\definecolor{myorange}{RGB}{255,104,25}
\definecolor{mypurple}{RGB}{245,0,228}
\definecolor{myblue}{RGB}{19,15,255}

\newtheorem{thm}{Theorem}
\newtheorem{assum}{Assumption}
\newtheorem{defn}{Definition}
\newtheorem{rem}{Remark}
\newtheorem{lem}{Lemma}

\newtheorem{prop}{Proposition}
\newtheorem{proof}{Proof}  

\DeclareMathOperator*{\argmin}{argmin}
\DeclareMathOperator*{\interior}{int}
\DeclareMathOperator*{\col}{\mathrm{col}}

\DeclarePairedDelimiter\floor{\lfloor}{\rfloor}

\newcommand{\change}[1]{{\color{black}#1}}

\begin{document}
\title{Distributed Model Predictive Control for Dynamic Cooperation of Multi-Agent Systems}
\author{Matthias K\"ohler, Matthias A. M\"uller, and Frank Allg\"ower
\thanks{%
F. Allg{\"o}wer is thankful that this work was funded by the Deutsche Forschungsgemeinschaft (DFG, German Research Foundation) -- AL 316/11-2 - 244600449; and under Germany's Excellence Strategy -- EXC 2075 -- 390740016.}
\thanks{M. K\"ohler and F. Allg\"ower are with the University of Stuttgart, Institute for Systems Theory and Automatic Control, 70550 Stuttgart, Germany (e-mails: matthias.koehler@ist.uni-stuttgart.de, frank.allgower@ist.uni-stuttgart.de).}
\thanks{M. A. M\"uller is with the Leibniz University Hannover, Institute of Automatic Control, 30167 Hanover, Germany (e-mail: mueller@irt.uni-hannover.de).}
}

\maketitle
\thispagestyle{empty}  

\begin{abstract}
We propose a distributed model predictive control (MPC) framework for coordinating heterogeneous, nonlinear multi-agent systems under individual and coupling constraints.
The cooperative task is encoded as a shared objective function minimized collectively by the agents.
Each agent optimizes an artificial reference as an intermediate step towards the cooperative objective, along with a control input to track it.
We establish recursive feasibility, asymptotic stability, and transient performance bounds under suitable assumptions.
The solution to the cooperative task is not predetermined but emerges from the optimized interactions of the agents.
We demonstrate the framework on numerical examples inspired by satellite constellation control, collision-free narrow-passage traversal, and coordinated quadrotor flight.
\end{abstract}


\section{Introduction}
Multiple systems coordinating as dynamically decoupled agents have various applications due to a high degree of flexibility, modularity, and avoidance of single points of failure.
For example, these multi-agent systems arise in control and coordination of vehicles~\cite{Dunbar2006_Distributedrecedinghorizon, Keviczky2008_DecentralizedRecedingHorizon, Lyu2021_MultivehicleFlockingCollision}, and can be used for exploration and establishing communication networks, e.g.~\cite{Oegren2004_CooperativeControlMobile, Sin2020_PassivityBasedDistributedAcquisition, Pippia2022_Reconfigurationsatelliteconstellation, Rickenbach2024_ActiveLearningBasedModel}.
In general, control design for these agents needs to account for coupling through common objectives, e.g. covering an area or moving in formation, or constraints, e.g. collision avoidance or remaining in communication range, while also dealing with agent-specific dynamics and constraints.
\change{Furthermore, some cooperative tasks naturally exhibit periodic trajectories \cite{Sin2020_PassivityBasedDistributedAcquisition, Pippia2022_Reconfigurationsatelliteconstellation}.}
Due to its applicability to nonlinear systems that are subject to constraints, and the possibility to take performance criteria into account, model predictive control (MPC) remains a powerful control method to handle these complex interdependencies.
To ensure scalability as well as modularity, and to avoid a single point of failure, distributed MPC is especially suitable for control of multi-agent systems.
In distributed MPC, the optimization problem is distributed across agents, enabling local computation and communication to generate the control input.
See~\cite{Gruene2017,Rawlings2020,Mueller2017b} for an introduction to (distributed) MPC. 

A key component of the proposed scheme in this paper is the use of artificial references as introduced in~\cite{Limon2008} under the name MPC for tracking.
To track an externally given reference independent of its feasibility or variation, an artificial reference is included in the optimization problem as an additional decision variable.
This allows the system to simultaneously optimize the (artificial) reference it tracks and the control input steering the system towards it. 
By penalizing the distance from the artificial reference to the externally given one, it can be ensured that the best-reachable reference is stabilized.
This method has been extended to various settings, including nonlinear systems~\cite{Limon2018} and periodic references~\cite{jkoehler_nonlinear_dynamic_2020, Krupa2022_HarmonicBasedModel}.
The performance of MPC for tracking has been analysed in~\cite{ferramoscaOptimalMPCTracking2011,Limon2018, MatthiasKohler2023_TransientPerformanceMPC}.
\change{See \cite{Krupa2024_Modelpredictivecontrol} for an overview of MPC for tracking.}

\change{%
Task-specific distributed MPC strategies have been proposed for a range of cooperative objectives. 
For instance, \cite{Dunbar2012_DistributedRecedingHorizon} establishes string stability for vehicle platoons with nonlinear dynamics by penalizing deviations from previously communicated trajectories. 
In \cite{Zheng2017}, a non-iterative distributed MPC approach for heterogeneous vehicle platoons is introduced. 
There, a terminal equality constraint based on the average of neighbouring outputs is imposed to ensure asymptotic stability, without prescribing a common desired goal for all agents.
Artificial references are used in \cite{Carron2020a} to address a coverage problem using a nonlinear multi-agent system. 
Specifically, each agent is assigned an external reference corresponding to the centroid of a Voronoi cell. 
The work of \cite{Rickenbach2024_ActiveLearningBasedModel} extends this to collision avoidance constraints.
In addition, a new scheme is proposed that integrates the computation of the coverage configuration directly into the MPC's optimization problem.
Overall, these approaches address particular system dynamics or specific cooperative tasks, but they do not offer a general framework for cooperative tasks.

A scheme for general cooperative tasks is proposed in~\cite{Muller2012_Cooperativecontroldynamically} for agents with nonlinear dynamics. 
It relies on terminal regions that must be centrally designed to match the cooperative tasks.
This limits the flexibility of the scheme in the face of changes in the cooperative task or multi-agent topology.
Moreover, it remains an open question how these terminal components can be designed for cooperative tasks beyond consensus and synchronization.
In \cite{Ferranti2022_DistributedNonlinearTrajectory}, a method for multi-agent trajectory coordination based on distributed MPC is presented.
While the method is suitable for a general class of agents and cooperative tasks with pre-specified solutions, it focuses on solving the optimization problem and does not discuss recursive feasibility or the closed-loop dynamics.

In summary, there is only limited research on distributed MPC approaches for nonlinear multi-agent systems engaged in general cooperative tasks, especially for periodic cooperative tasks. 
This paper proposes a distributed MPC methodology to close this gap.
A key mechanism is the use of artificial references to decouple local agent behaviour from the global cooperative task.
This enables a simplified design that is scalable and flexible.
In particular, it allows for a decentralized design of terminal costs and constraints irrespective of the cooperative task.
Hence, terminal components do not have to be redesigned when changing the cooperative task or the multi-agent topology.
}%
In previous work, we introduced a framework for cooperative control of nonlinear multi-agent systems using sequential distributed MPC and artificial references~\cite{MatthiasKohler2024_DistributedMPCSelfOrganized, MKoehler2023}.
\change{%
However, in these schemes, coupling constraints are not included and their inclusion is not straightforward.
In addition, fewer cooperative tasks are covered since a specific coupling structure is required in the objective function of the optimization problems.
Moreover, it is unclear how a performance bound can be achieved for these sequential schemes.
}
\change{The scheme presented in this paper does not have these limitations.}
The main contributions of this paper are as follows:
\begin{itemize}
    \item A general formulation of distributed MPC for cooperative tasks characterized by dynamic (periodic) trajectories, accommodating heterogeneous agents with nonlinear dynamics and constraint coupling.
    \item \change{Partial decoupling of handling} the agents' dynamics and constraints, and the design of terminal costs and terminal constraints, from the cooperative task. This facilitates a flexible design of components encoding the cooperative task.  
    \item Rigorous guarantees of recursive feasibility and asymptotic stability of a set containing solutions to the cooperative task.
    \item Transient performance bounds that show how closed-loop performance of the scheme improves with prediction horizon length, extending prior results on MPC for tracking.
    \item A design that does not require a solution \change{to} the cooperative task to be prescribed in advance, but ensures its emergence through decentralized optimization.
\end{itemize}

\subsection{Notation}
The interior of a set $\mathcal{A}$ is denoted by $\interior\mathcal{A}$.
The nonnegative reals are denoted by $\mathbb{R}_{\ge 0}$, and $\mathbb{N}_0$ denotes the natural numbers including 0.
The set of integers from $a$ to $b$, $a \le b$, is denoted by $\mathbb{I}_{a:b}$.
Let $\Vert \cdot \Vert$ be the Euclidean norm.
Given a set $\mathcal{A}$, the distance from a point $x$ to $\mathcal{A}$ is denoted by $\vert x \vert_{\mathcal{A}} = \inf_{z\in\mathcal{A}}\Vert x - z \Vert$.
If $\mathcal{A} = \{ z \}$, we simply write $\vert x \vert_{z}$.
Given a positive (semi-)definite matrix $A=A^\top$, the corresponding (semi-)norm is written as $\Vert x \Vert_A = \sqrt{x^\top A x}$.
For a collection of $m$ vectors $v_i \in \mathbb{R}^{n_i}$, $i \in \mathbb{I}_{1:m}$, we denote the stacked vector by $v = \mathrm{col}_{i=1}^m(v_i) = [v_1^\top \dots v_m^\top]^\top$.
We define $\mathcal{A} \ominus \mathcal{Q} = \{ a \mid \forall q\in\mathcal{Q},\; a + q \in \mathcal{A}\}$ for two sets $\mathcal{A}$ and $\mathcal{Q}$.
Given $m$ sets $\mathcal{A}_i$, $i\in\mathbb{I}_{1:m}$, we use $\prod_{i\in\mathbb{I}_{1:m}}\mathcal{A}_i = \mathcal{A}_1 \times \dots \times \mathcal{A}_m$.
Define $\mathcal{B}_{\eta}(\tilde{x}) = \{ x \mid \vert x \vert_{\tilde{x}} \le \eta \}$, and $\mathcal{B}_{\eta} = \mathcal{B}_{\eta}(0)$.
Let $X\subseteq\mathbb{R}^n$ be a closed, convex set and $x\in\mathbb{R}^n$. 
Then, $\mathcal{P}_{X}[x] = \argmin_{y\in X}\Vert y - x \Vert$ is the unique projection of $x$ onto $X$~\cite[Prop. 1.1.4]{Bertsekas2016}.
For a set $X$, $2^X$ denotes the power set of $X$.
Comparison functions, e.g. $\mathcal{K}_\infty$, are used; see~\cite{Kellett2014} for definitions and properties.
For periodic trajectories $y_T \in \mathcal{Y}^{T}$ starting at $y_T(0)$ with period length $T$, we write $y_T(k)$ to mean $y_T(k \bmod T)$ for $k\in\mathbb{N}_0$, e.g. $y_T(T) = y_T(0)$.
The floor function is denoted by $\floor*{\cdot}$.
We define the distance between two periodic trajectories $\hat{y}_{T},  y_{T}$ with period length $T$ as
$\vert \hat{y}_{T} \vert_{y_{T}} = \sum_{\tau=0}^{T-1} \Vert \hat{y}_{T}(\tau) - y_{T}(\tau) \Vert$.

\section{Multi-agent system}\label{sec:setup}
We consider a multi-agent system comprising $m \in \mathbb{N}$ heterogeneous agents with nonlinear discrete-time dynamics
\begin{equation}\label{eq:agent_dynamics}
    x_i(t + 1) = f_i(x_i(t), u_i(t))
\end{equation}
with state $x_i(t) \in {X}_i \subseteq \mathbb{R}^{n_i}$ and input $u_i(t) \in {U}_i \subseteq \mathbb{R}^{q_i}$ at time $t\in\mathbb{N}_0$, and where $f_i: {X}_i \times {U}_i \to {X}_i$ is continuous.
We assume that the agents are subject to individual constraints $(x_i(t), u_i(t)) \in {Z}_i \subseteq {X}_i \times {U}_i$ for $t \in \mathbb{N}_0$, where ${Z}_i$ is compact.
Here, $i\in\mathbb{I}_{1:m}$, which we omit in the following when it is clear that the statement should be interpreted for all $i\in\mathbb{I}_{1:m}$.
We write $x_{i,u_i}(\change{t}, x_i)$ to denote the solution of~\eqref{eq:agent_dynamics} with initial state $x_i$ generated by the input sequence $u_i$, and use $x_{i,u_i}(\change{t})$ \change{equivalently to shorten the notation} whenever the initial state is clear.
Define the set of (locally) admissible input sequences of length $K\in\mathbb{N}\cup\{\infty\}$ as $\mathbb{U}_i^K(x_i) = \{u_i \in U_i^K \mid (x_{i,u_i}(\change{t}, x_i), u_i(\change{t})) \in Z_i, \change{t}\in\mathbb{I}_{0:K-1}\}$.
If $K=1$, we simply write $\mathbb{U}_i(x_i)$.

It is assumed that agents can communicate according to an undirected graph $\mathcal{G} = (\mathcal{V}, \mathcal{E})$ with vertices $\mathcal{V}$ and edges $\mathcal{E}$.
Each agent is assigned a vertex $i\in\mathcal{V}$, and the vertices are connected through edges $e_{ij} = e_{ji} \in\mathcal{E}$.
The set of neighbours of agent $i$ is then $\mathcal{N}_i = \{j \in\mathcal{V} \mid e_{ji} \in \mathcal{E}\}$, i.e. it contains all agents with which agent $i$ may communicate.
In the following, we assume lossless and immediate communication.
\change{%
Moreover, $v_{\mathcal{N}_i}$ denotes the set of variables $\{v_j\}_{j \in \mathcal{N}_i}$. 
For example, $x_{\mathcal{N}_i}$ collects the state vectors of all neighbours of agent $i$.
}

Furthermore, agents may also be coupled through constraints described by an appropriate set $\mathcal{C}_i$, i.e. the constraint is $(x_i, x_{\mathcal{N}_i}) \in \mathcal{C}_i$.
For example, if agents need to remain in a certain communication range specified by $\delta_i \in \mathbb{R}^{\vert \mathcal{N}_i \vert}$ to neighbours, the constraint could be $\mathcal{C}_i = \{(x_i, x_{\mathcal{N}_i}) \mid \col_{j\in\mathcal{N}_i} (\Vert x_i - x_j \Vert^2) \le \delta_i\}$.
Implicitly, we assume that coupled agents are connected in the communication graph $\mathcal{G}$.
To avoid excessive notation, we do not consider coupling constraints including the inputs of agents, although these may be included.

The control goal is for the agents' outputs to converge to a periodic trajectory that satisfies a cooperative task, such as synchronization or flocking, which includes consensus as a special case if the trajectory is an equilibrium with an arbitrary period.
For this purpose, we define the \change{agents'} outputs
\begin{equation*}
    y_i(t) = h_i(x_i(t), u_i(t))
\end{equation*}
with $y_i(t) \in {Y}_i \subseteq \mathbb{R}^{p_i}$ and continuous $h_i: {X}_i \times {U}_i \to {Y}_i$. 
Notably, the final periodic output trajectory is not given \emph{a priori} by an external governor, but could be any that achieves the cooperative task.
We characterize this cooperative task through an \emph{output cooperation set} that contains acceptable output trajectories achieving this task.
\begin{defn}\label{def:output_cooperation_set}
    A nonempty set $\mathcal{Y}_T^{\mathrm{c}}$ is called an \emph{output cooperation set} if it is compact and the cooperative task is achieved whenever $[\col_{i=1}^m(y_i(t)), \dots, \col_{i=1}^m(y_i(t+T-1))] \in \mathcal{Y}_T^{\mathrm{c}}$.
\end{defn}

In the following, we will combine penalties on the distance to the output cooperation set and to an intermediary periodic trajectory in order to design a distributed MPC scheme that achieves the cooperative control goal.

\section{Tracking of cooperation outputs}
In the distributed MPC scheme proposed below, the strategy is to design energy-like functions that ensure asymptotic fulfilment of the cooperative task.
The agents decide in every time step the reference output trajectory they want to track by minimizing a tracking cost as well as a penalty function on the distance from the reference output trajectory to the cooperative task.
We will call these reference output trajectories \emph{cooperation outputs}.

Thus, it is not necessary to specify \emph{a priori} a particular solution to the cooperative task.
Instead, the agents coordinate a solution of the cooperative task depending on the optimization problem in the distributed MPC scheme using cooperation outputs as intermediary targets.
Furthermore, since any feasible cooperation output can be assumed, the region of attraction is greatly extended and smaller prediction horizons can be realized compared to targeting a solution of the cooperative task directly, similarly to standard MPC for tracking~\cite{Limon2018}. 
This can offset the added computational complexity caused by including cooperation outputs as decision variables.

Before further specifying the cooperation outputs, we define the set of feasible periodic trajectories on a state and input level that are strictly inside the constraints,
\begin{align*}
    Z_{T,i} = &\{r_{T,i} = (x_{T,i}, u_{T,i}) \in (\interior Z_i)^{T} \\
    &\mid x_{T,i}(\tau+1) = f_i(x_{T,i}(\tau), u_{T,i}(\tau)),\; \tau\in\mathbb{I}_{0:T-2}, \\
    &\hspace{0.87em} x_{T,i}(0) = f_i(x_{T,i}(T-1), u_{T,i}(T-1))\}.
\end{align*}
We choose a nonempty subset $\mathcal{Z}_{T,i} \subseteq Z_{T,i}$ of admissible periodic cooperation trajectories, since not all in $Z_{T,i}$ may be desirable.
Next, we identify all output trajectories that correspond to these references and which strictly satisfy the coupling constraints in the set 
\begin{align*}
    \mathbb{Y}_{T,i} = &\{ y_{T,j} \in Y_j^{T},\; j\in\mathcal{N}_i\cup\{i\} \\
    &\mid \exists r_{T,j} = (x_{T,j}, u_{T,j}) \in \mathcal{Z}_{T,j},\\
    &\hspace{0.87em} y_{T,j}(\tau) = h_j(x_{T,j}(\tau), u_{T,j}(\tau)),\\
    &\hspace{0.87em} (x_{T,i}(\tau), x_{T,\mathcal{N}_i}(\tau)) \in \mathcal{C}_i \ominus \mathcal{B}_{\eta_i},\; \forall \tau\in\mathbb{I}_{0:T-1} \}
\end{align*}
with some (small) constant $\eta_i > 0$.
Again, we choose a nonempty subset $\mathcal{Y}_{T,i} \subseteq \mathbb{Y}_{T,i}$ of admissible cooperation outputs.
These choices should satisfy the following condition.
\begin{assum}\label{assm:compact_cooperation_sets}
    The sets $\mathcal{Z}_{T,i}$ and $\mathcal{Y}_{T,i}$ are compact.
\end{assum}
\change{Note that the coupling constraints $\mathcal{C}_i$ need not be bounded for $\mathcal{Y}_{T,i}$ to be compact.}

The following definition provides a penalty function that penalizes the distance from the cooperation outputs to the output cooperation set. 
Define $\mathcal{Y}_{T} = \{y_{T,i} \in Y_i^{T},\; i\in\mathbb{I}_{1:m} \mid (y_{T,i}, y_{T,\mathcal{N}_i}) \in \mathcal{Y}_{T,i} \}$.
\begin{defn}\label{def:COF}
    A \change{continuous} function $W^{\mathrm{c}}: \mathcal{Y}_T \to \mathbb{R}_{\ge 0}$ is a \emph{cooperation objective function} if it has the following properties:
    \begin{enumerate}
        \item There exist $\alpha_{\mathrm{lb}}^\mathrm{c}, \alpha_{\mathrm{ub}}^\mathrm{c} \in \mathcal{K}_\infty$ such that 
        $\alpha_{\mathrm{lb}}^\mathrm{c}(\vert y_{T} \vert_{\mathcal{Y}_{T}^\mathrm{c}}) \le W^{\mathrm{c}}(y_T) \le \alpha_{\mathrm{ub}}^\mathrm{c}(\vert y_{T} \vert_{\mathcal{Y}_{T}^\mathrm{c}})$, where $y_T = \col_{i=1}^m y_{T,i}$.
        \item $W^{\mathrm{c}}$ is separable according to $\mathcal{G}$, i.e. $W^{\mathrm{c}}(y_T) = \sum_{i=1}^{m}  W_{i}^{\mathrm{c}}(y_{T,i}, y_{T,\mathcal{N}_i})$.
        \item For any $y_{T} \in \mathcal{Y}_{T}$ the cost is shift-invariant, i.e. $W_{i}^{\mathrm{c}}(y_{T,i}(\cdot), y_{T,\mathcal{N}_i}(\cdot)) = W_{i}^{\mathrm{c}}(y_{T,i}(\cdot+1), y_{T,\mathcal{N}_i}(\cdot+1))$.
    \end{enumerate}
\end{defn}

The cooperation objective function is designed to encode the cooperative task and indicate the distance from the multi-agent system to it.
The agents should choose their cooperation outputs over time such that the cooperation objective function is minimized.
Eventually this leads to closed-loop fulfilment of the cooperative task in the \change{agents'} outputs.
We will detail below additional conditions on $W^{\mathrm{c}}$ and $\mathcal{Y}_{T,i}$ that are sufficient for this.

In order to track cooperation outputs, we introduce a link between these and state and input trajectories.
For each cooperation output, a unique corresponding periodic reference trajectory should exist, and a change in the former should continuously result in a change in the latter.
We capture this in the following standard assumption in MPC for tracking (cf. \cite[Assm. 1]{Limon2018}, \cite[Assm. 6]{jkoehler_nonlinear_dynamic_2020}, {\cite[Assm. 3]{MatthiasKohler2024_DistributedMPCSelfOrganized}}).
\begin{assum}\label{assm:unique_corresponding_equilibrium}
    There exist Lipschitz \change{continuous}, injective functions $g_{x,i}:\mathcal{Y}_{T,i}^{T} \to X_i^{T}$ and $g_{u,i}:\mathcal{Y}_{T,i}^{T} \to U_i^{T}$ such that $r_{T,i} = (x_{T,i}, u_{T,i}) = (g_{x,i}(y_{T,i}), g_{u,i}(y_{T,i})) \in \mathcal{Z}_{T,i}$ is unique, and $y_{T,i}(\tau) = h_i(x_{T,i}(\tau), u_{T,i}(\tau))$ for $\tau \in \mathbb{I}_{0:T-1}$. The Lipschitz constants are $L_{x,i}$ and $L_{u,i}$, respectively.
\end{assum}

See for example \cite[Rmk. 1]{Limon2018} for a sufficient condition for Assumption~\ref{assm:unique_corresponding_equilibrium} based on the Jacobians of the agents' dynamics.
\begin{rem}\label{rem:unique_corresponding_equilibrium}
    For $r_{T,i},\hat{r}_{T,i} \in \mathcal{Z}_{T,i}$, we define $\vert \hat{r}_{T,i} \vert_{r_{T,i}} = \sqrt{\vert \hat{x}_{T,i} \vert_{x_{T,i}}^2 + \vert \hat{u}_{T,i} \vert_{u_{T,i}}^2}$. Then, a simple calculation yields $\vert \hat{r}_{T,i} \vert_{r_{T,i}} \le \max(L_{x,i}, L_{u,i}) \vert \hat{y}_{T,i} \vert_{y_{T,i}}$.
\end{rem}

We use a stage cost satisfying the following standard assumption such that agents track $x_{T,i}$ in order to realize $y_{T,i}$.
Define $\ell_i'(x_i, r_{T,i}(\tau)) = \min_{u_i \in \mathbb{U}_i(x_i)}\, \ell_i(x_i, u_i, r_{T,i}(\tau))$.
\begin{assum}\label{assm:stage_cost_lower_and_upper_bound}
    There exist $\alpha_{\mathrm{lb}}^{\ell_i}, \alpha_{\mathrm{ub}}^{\ell_i} \in \mathcal{K}_\infty$ such that for all $(x_i, u_i) \in Z_i$ and $r_{T,i} \in \mathcal{Z}_{T,i}$, with $\tau\in\mathbb{I}_{0:T-1}$,
    \begin{equation}\label{eq:stage_cost_lower_and_upper_bound}
        \alpha_{\mathrm{lb}}^{\ell_i}(\vert x_i \vert_{x_{T,i}(\tau)}) \le \ell_i'(x_i, r_{T,i}(\tau)) \le \alpha_{\mathrm{ub}}^{\ell_i}(\vert x_i \vert_{x_{T,i}(\tau)}).
    \end{equation}
\end{assum}

The final component of the tracking part are suitable terminal ingredients for any periodic reference trajectory that the agents may choose.
These are also standard in MPC for tracking, cf.~\cite[Assm. 3]{Limon2018},~\cite[Assm. 2]{jkoehler_nonlinear_dynamic_2020}.
\begin{assum}\label{assm:terminal_ingredients}
    \change{\emph{a)}}
    There exist terminal control laws
    $
    k_{i}^\mathrm{f}: {X}_i \times Z_i \to {U}_i
    $, continuous terminal costs
    $
    V_i^\mathrm{f}: {X}_i \times Z_i \to \mathbb{R}_{\ge 0}
    $, and compact terminal sets
    $
    \mathcal{X}_i^\mathrm{f}: Z_i \to 2^{X_i}$ such that for any $r_{T,i} \in \mathcal{Z}_{T,i}$ and $x_i \in \mathcal{X}_i^\mathrm{f}(r_{T,i}(\tau))$, for all $\tau \in \mathbb{I}_{0:T-1}$,
    \begin{subequations}
        \begin{align}
            &V_i^\mathrm{f}\Big(f_i\big(x_i, k^\mathrm{f}_i(x_i, r_{T,i}(\tau))\big), r_{T,i}(\tau+1)\Big) - V_i^\mathrm{f}\Big(x_i, r_{T,i}(\tau)\Big)
            \notag
            \\
            &\le  - \ell_i\Big(x_i, k^\mathrm{f}_i(x_i, r_{T,i}(\tau)), r_{T,i}(\tau)\Big), \label{eq:terminal_cost_decrease}
            \\
            &\Big(x_i, k^\mathrm{f}_i\big(x_i, r_{T,i}(\tau)\big)\Big) \in \change{Z_i},
            \\
            &f_i\Big(x_i, k^\mathrm{f}_i\big(x_i, r_{T,i}(\tau)\big)\Big) \in \mathcal{X}_i^\mathrm{f}\Big(r_{T,i}(\tau+1)\Big).
        \end{align}
    \end{subequations}

    \change{\emph{b)}}
    Moreover, there exist $c_i^{\mathrm{b}} \ge 0$ and $c_i^{\mathrm{f}} > 0$ such that for any $r_{T,i} \in \mathcal{Z}_{T,i}$ and $x_i \in \mathcal{X}_i^\mathrm{f}(r_{T,i}(\tau))$, for all $\tau \in \mathbb{I}_{0:T-1}$,
    \begin{subequations}
        \begin{align}
            \mathcal{B}_{c_i^{\mathrm{b}}}(x_{T,i}(\tau)) &\subseteq \mathcal{X}_i^\mathrm{f}(r_{T,i}(\tau)),\label{eq:terminal_non_empty_interior}
            \\
            V_i^\mathrm{f}(x_i, r_{T,i}(\tau)) &\le c_i^{\mathrm{f}} \ell_i'(x_i, r_{T,i}(\tau)).\label{eq:terminal_cost_upper_bound}
        \end{align}
    \end{subequations}
    
    \change{\emph{c)}}
    \change{Furthermore}, if the only feasible choice is $c_i^{\mathrm{b}} = 0$ to satisfy~\eqref{eq:terminal_non_empty_interior}, there also exists $\varepsilon_i^{\mathrm{f}} > 0$ and $N_i^{\mathrm{f}}\in\mathbb{N}$ such that for any $r_{T,i} \in \mathcal{Z}_{T,i}$ and $x_i$ with $\ell_i'(x_i, r_{T,i}(0)) \le \varepsilon_i^{\mathrm{f}}$, there exists $u_i^{\mathrm{f}} \in U_i^{N_i^{\mathrm{f}}}$ with $x_{i,u_i^{\mathrm{f}}}(N_i^{\mathrm{f}}, x_i) = x_{T,i}(N_i^{\mathrm{f}})$ and $\sum_{k=0}^{N_i^{\mathrm{f}}-1} \ell_i (x_{i,u_i^{\mathrm{f}}}(k),u_i^{\mathrm{f}}(k), r_{T,i}(k)) \le  c_i^{\mathrm{f}}\ell_i'(x_i, r_{T,i}(0))$.
\end{assum}

\change{Note that Assumption~\ref{assm:terminal_ingredients} above results in terminal equality constraints when $c_i^{\mathrm{b}} = 0$.}
\change{Then,} $\mathcal{X}_i^{\mathrm{f}}(r_{T,i}(\tau)) = \{ x_{T,i}(\tau) \}$ and $V_i^{\mathrm{f}}(x_i, r_{T,i}(\tau)) = 0$ for $x_i \in \mathcal{X}_i^{\mathrm{f}}(r_{T,i}(\tau))$, and a suitable upper bound on the summed stage cost \change{is required}.
See, e.g., \cite[Lem. 5]{jkoehler_nonlinear_dynamic_2020} for a sufficient condition for Assumption~\ref{assm:terminal_ingredients}, and~\cite{JKoehler2020a} for ways to compute these generalized terminal ingredients offline based on linear matrix inequalities.
Terminal equality constraints can be used if the agents are locally uniformly finite time controllable and a bound on the penalty function for tracking holds, see~\cite[Prop. 4]{jkoehler_nonlinear_dynamic_2020}.
\change{A key strength of our framework is that Assumption~\ref{assm:terminal_ingredients} is fairly standard, as pointed out above, and that the components admit a fully decentralized design, since coupling constraints do not appear in Assumption~\ref{assm:terminal_ingredients}}.

\change{To handle the coupling constraints}, the set of admissible cooperation outputs $\mathcal{Y}_{T,i}$ is designed \change{to ensure their strict satisfaction for all cooperation outputs} (see Assumption~\ref{assm:compact_cooperation_sets}).
This allows for sufficiently small terminal regions.
\begin{assum}\label{assm:tightened_coupling_constraints}
    There exists $\eta_i > 0$ such that $k_{j}^\mathrm{f}$, $V_j^\mathrm{f}$, and $\mathcal{X}_j^\mathrm{f}(r_{T,j}(\tau))$, with $j\in\mathcal{N}_i \cup \{i\}$, from 
    Assumption~\ref{assm:terminal_ingredients} entail
    \begin{subequations}\label{eq:terminal_coupling_constraints}
        \begin{align}
            (x_i, x_{\mathcal{N}_i}) &\in \mathcal{C}_i
            \\
            \big(x_{i,k^\mathrm{f}_i(x_i, r_{T,i}\change{(\tau)})}(1), x_{\mathcal{N}_i, k^\mathrm{f}_{\mathcal{N}_i}(x_{\mathcal{N}_i}, r_{T,\mathcal{N}_i}\change{(\tau)})}(1)\big) &\in \mathcal{C}_i
        \end{align}
    \end{subequations}
    for all $x_j \in \mathcal{X}_j^\mathrm{f}(r_{T,j}(\tau))$ and all $\tau\in\mathbb{I}_{0:\change{T-1}}$, if  
    $(x_{T, i}(\tau), x_{T, \mathcal{N}_i}(\tau)) \in \mathcal{C}_i \ominus \mathcal{B}_{\eta_i}$
    for all $\tau\in\mathbb{I}_{0:\change{T-1}}$.
\end{assum}
\begin{rem}
    \change{An often-used approach to designing terminal sets is to select sublevel sets of the terminal cost function, i.e. $\mathcal{X}_i^\mathrm{f}(r_{T,i}(\tau)) = \{x_i\in\mathbb{R}^{n_i} \mid V_i^\mathrm{f}(x_i, x_{T,i}(\tau)) \le \alpha_i\}$ with $\alpha_i>0$. This construction also results from the procedure outlined in~\cite{JKoehler2020a} for deriving terminal components that satisfy Assumption~\ref{assm:terminal_ingredients}.
    In this standard setting, one can always choose a smaller value of} $\alpha_i$ to satisfy Assumption~\ref{assm:tightened_coupling_constraints}.
    For terminal equality constraints Assumption~\ref{assm:tightened_coupling_constraints} holds trivially with $k^\mathrm{f}_i(x_i, r_{T,i}) = k^\mathrm{f}_i(x_{T,i}(\tau), r_{T,i}) = u_{T,i}(\tau)$.
\end{rem}
\begin{rem}
    Choosing $\mathcal{Z}_{T,i}$ and $\mathcal{Y}_{T,i}$ offline, and relying on Assumption~\ref{assm:tightened_coupling_constraints}, may lead to small terminal regions.
    By choosing $\mathcal{Z}_{T,i}$ in the interior of $Z_i^{T}$ and tightening the coupling constraints with $\eta_i>0$, we trade off faster convergence with operation closer to the constraints' boundary.
    This is standard in MPC for tracking~\cite{Limon2016,Limon2018,jkoehler_nonlinear_dynamic_2020}.
    A potential remedy is given by also optimizing the terminal set size online; see~\cite[Prop. 11]{jkoehler_nonlinear_dynamic_2020} for a modification in the case of polytopic constraints.
\end{rem}

We have described the necessary components of our distributed MPC scheme, which we combine in the next section.

\section{Distributed MPC for cooperation}
\subsection{Distributed MPC scheme}
In this section, we introduce an iterative distributed MPC scheme in which a single optimization problem is formulated that can be solved using a decentralized optimization algorithm.
The penalty functions for tracking over the prediction horizon $N\in\mathbb{N}_0$ are given by
\begin{align*}
    &J_i^{\mathrm{tr}}(x_i, u_i, r_{T,i}) = \smash[b]{\sum_{k=0}^{N-1}} \ell_i (x_{i,u_i}(k), u_i(k), r_{T,i}(k))
    \\
    &\hspace{10em}+ V_i^\mathrm{f}(x_{i,u_i}(N), r_{T,i}(N)).
\end{align*}
In addition, we introduce a previous cooperation output $y_{T}^{\mathrm{pr}}$ and a function $V_i^{\Delta}$ that penalizes the difference between the cooperation output chosen at the current time step and $y_{T}^{\mathrm{pr}}$.
This ensures that the closed-loop state eventually follows a (unique) periodic trajectory despite the cooperative task often allowing (by design) multiple solutions; further details on $V_i^{\Delta}$ are given below.
Next, the penalty functions for tracking, for cooperation, and for changing the cooperation output are combined to the objective functions:
\begin{align*}
    &J_i(x_i, u_i, y_{T,i}, y_{T,i}^{\mathrm{pr}}, y_{T,\mathcal{N}_i})  
    \\
    &= \hspace{-1.5pt} J_i^{\mathrm{tr}}(x_i, u_i, r_{T,i}) \hspace{-2pt} + \hspace{-2pt} \lambda(N)\big(V_i^{\Delta}(y_{T,i}, y_{T,i}^{\mathrm{pr}}) \hspace{-2pt} + \hspace{-2pt} W_{i}^{\mathrm{c}}(y_{T,i}, y_{T,\mathcal{N}_i})\big). 
\end{align*}
Furthermore, we introduce scaling by $\lambda(N)$ satisfying the following assumption.
\begin{assum}\label{assm:scaling}
    The scaling function $\lambda:\mathbb{N}_0 \to \change{\mathbb{R}_{\ge 0}}$ satisfies $\lambda(N) \ge N$ and $\lambda(0) \ge 1$.
\end{assum}

This scaling can be satisfied easily with $\lambda(N) = N+1$. 
It is essential for the asymptotic performance bound in Section~\ref{sec:performance_bounds} below, cf. also~\cite{MatthiasKohler2023_TransientPerformanceMPC}, but not for stability, where $N$ is fixed.

Finally, we introduce the optimization problem that is solved in each time step.
Given (collected) state measurements $x$ and (previous) periodic output trajectories $y_{T}^{\mathrm{pr}}$, it is as follows.%
\begin{subequations}\label{eq:central_OP}
    \begin{align}
        &\mathcal{J}(x, y_T^{\mathrm{pr}}) = \min_{\substack{u \in \mathbb{U}^N(x) \\ y_{T} \\ \change{r_{T}} }} \sum_{i=1}^{m} J_i(x_i, u_i, y_{T,i}, y_{T,i}^{\mathrm{pr}}, y_{T,\mathcal{N}_i}) &
        \label{eq:central_OP_objective}
        \\
        \intertext{subject to, for all $i \in \mathbb{I}_{1:m}$,}
        &x_{i, u_i}(N, x_i) \in \mathcal{X}_i^\mathrm{f}(r_{T,i}(N)),
        \label{eq:central_OP_terminal_constraint}
        \\
        &(x_{i, u_i}(k, x_i), x_{\mathcal{N}_i, u_{\mathcal{N}_i}}(k, x_{\mathcal{N}_i})) \in \mathcal{C}_i, \; k\in\mathbb{I}_{0:N},
        \label{eq:central_OP_coupling_constraints}
        \\
        &(y_{T, i}, y_{T, \mathcal{N}_i}) \in \mathcal{Y}_{T,i}.
        \label{eq:central_OP_artificial_reference_constraints}
    \end{align}
\end{subequations}
The structure of~\eqref{eq:central_OP} agrees with the multi-agent system's communication topology $\mathcal{G}$.
Hence, this optimization problem is amenable to decentralized optimization algorithms, e.g.~\cite{
    Boyd2010, 
    Stewart2011_Cooperativedistributedmodel, 
    Giselsson2013_Acceleratedgradientmethods, 
    Engelmann2020, 
    Stomberg2024 
}.
The choice depends on the convexity and type of the objective function and constraints, as well as considerations of local computation and communication capabilities; a discussion of which goes beyond the scope of this paper.
\change{Lastly, $r_{T,i}$ is uniquely determined by $y_{T,i}$ according to Assumption~\ref{assm:unique_corresponding_equilibrium}. 
The mapping defined in Assumption~\ref{assm:unique_corresponding_equilibrium} does not need to be explicitly known to solve \eqref{eq:central_OP}, as it can be handled implicitly within the optimization. Implementation details are provided in the published code \cite{public_code}.}

The solution of~\eqref{eq:central_OP} at time $t$ depends on $x(t)$ and $y_T^{\mathrm{pr}}(\cdot\vert t)$, which we combine in $\xi(t) = (x(t), y_T^{\mathrm{pr}}(\cdot\vert t))$ or simply $\xi = (x, y_T^{\mathrm{pr}})$ if we do not refer to a specific time. 
We denote it by $u_i^0(\cdot\vert \xi(t))$ and $y_{T,i}^0(\cdot \vert \xi(t))$ with the corresponding $r_{T,i}^0(\cdot\vert \xi(t))$.
Here, e.g. $y_T^{\mathrm{pr}}(k\vert t)$ denotes the $k$-th step of $y_T^{\mathrm{pr}}$ at time $t$ and $u_i^0(k\vert \xi)$ is the $k$-th step optimal prediction given $\xi$.
At $t=0$, $V_i^\Delta$ is omitted in~\eqref{eq:central_OP_objective}, i.e. $y_T^{\mathrm{pr}}(\cdot \vert 0)$ plays no role and can be chosen arbitrarily.
Otherwise, we set $y_T^{\mathrm{pr}}(\cdot\vert t) = y_T^0(\cdot+1\vert \xi(t-1))$.
The set of states for which~\eqref{eq:central_OP} is feasible is denoted by $\mathcal{X}_N$, and it is independent of $y_T^{\mathrm{pr}}$.
We use $\mathcal{J}(\xi)=\mathcal{J}(x, y_T^{\mathrm{pr}})$ accordingly.

In each time step, after solving~\eqref{eq:central_OP}, the first part of the optimal input sequence $\mu_{i}(\xi(t)) =  u_i^0(0\vert \xi(t))$ is applied to the system.
The global closed-loop system is then given by
\begin{subequations}\label{eq:central_global_closed_loop}
    \begin{align}
        x(t + 1) &= f(x(t), \mu(\xi(t))), & x(0) = x_0,
        \label{eq:central_global_closed_loop_a}
        \\
        y(t) &= h(x(t), \mu(\xi(t))), & y_T^{\mathrm{pr}}(\cdot \vert 0) \text{ arbitrary},
    \end{align}
\end{subequations}
with some initial condition $x_0 \in X = \prod_{i=1}^m X_i$.

\subsection{Example: Satellite constellation}\label{ssec:satellite_example}
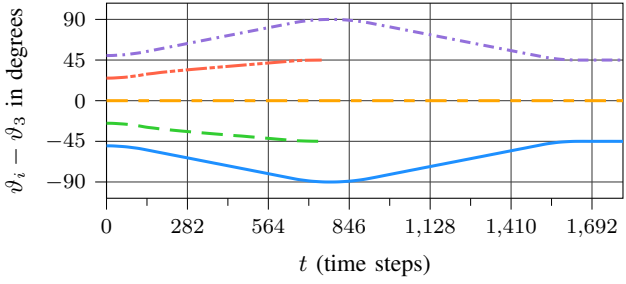
\begin{figure}[tb]
    \setlength\axisheight{0.47\linewidth}
    \setlength\axiswidth{0.95\linewidth}
    \centering
    \begin{tikzpicture}
\definecolor{vivid_blue}{HTML}{1E90FF}  
\definecolor{bright_green}{HTML}{32CD32}  
\definecolor{vibrant_orange}{HTML}{FFA500}  
\definecolor{tomato_red}{HTML}{FF6347}  
\definecolor{medium_purple}{HTML}{9370DB}  
\definecolor{calm_teal}{HTML}{20B2AA}  

\begin{axis}[
height=\axisheight,
legend cell align={left},
legend style={
    at={(0.99,0.99)}, 
    anchor=north east, 
    draw=lightgray, 
    fill opacity=0.8, 
    draw opacity=1,
    text opacity=1,
    font=\small
},
minor x tick num=1,
minor y tick num=0,
tick align=outside,
tick pos=left,
width=\axiswidth,
x grid style={darkgray},
xlabel={$t$ (time steps)},
xmajorgrids,
xmin=0.0, 
xmax=1800,
xtick distance=282,
xtick style={color=black},
y grid style={darkgray},
ymajorgrids,
ytick style={color=black},
ytick={-90, -45, 0, 45, 90},
ylabel={$\vartheta_i - \vartheta_3$ in degrees},
label style={font=\small},
tick label style={font=\footnotesize}
]

\addplot[vivid_blue, line width=0.40mm] table [x=t, y=dtheta] {./plotdata/satellite_constellation/sat_A1_dtheta.tex};

\addplot[bright_green, line width=0.40mm, dash pattern=on 8pt off 4pt,] table [x=t, y=dtheta] {./plotdata/satellite_constellation/sat_A2_dtheta.tex};

\addplot[vibrant_orange, line width=0.40mm, dash pattern=on 8pt off 4pt on 4pt off 4pt,] table [x=t, y=dtheta] {./plotdata/satellite_constellation/sat_A3_dtheta.tex};

\addplot[tomato_red, line width=0.40mm, dash pattern=on 8pt off 1pt on 2pt off 1pt on 2pt off 1pt,] table [x=t, y=dtheta] {./plotdata/satellite_constellation/sat_A4_dtheta.tex};

\addplot[medium_purple, line width=0.40mm, dash pattern=on 4pt off 2pt on 1pt off 2pt,] table [x=t, y=dtheta] {./plotdata/satellite_constellation/sat_A5_dtheta.tex};

\end{axis}
\end{tikzpicture}
    \caption{Relative angular positions of the satellites.
    From the top: $\vartheta_1 - \vartheta_3$, $\vartheta_2 - \vartheta_3$, $\vartheta_3 - \vartheta_3$, $\vartheta_4 - \vartheta_3$, and $\vartheta_5 - \vartheta_3$.}
    \label{fig:satellite_angular_position}
\end{figure}
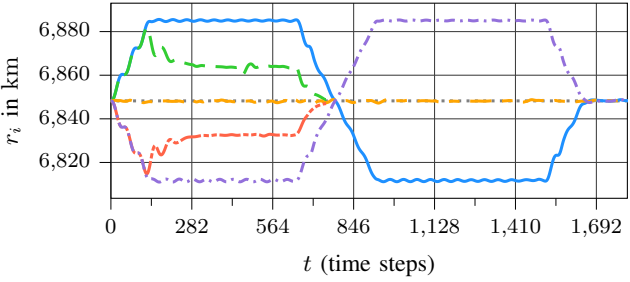
\begin{figure}[tb]
    \setlength\axisheight{0.47\linewidth}
    \setlength\axiswidth{0.95\linewidth}
    \centering
    \begin{tikzpicture}
\definecolor{vivid_blue}{HTML}{1E90FF}  
\definecolor{bright_green}{HTML}{32CD32}  
\definecolor{vibrant_orange}{HTML}{FFA500}  
\definecolor{tomato_red}{HTML}{FF6347}  
\definecolor{medium_purple}{HTML}{9370DB}  
\definecolor{calm_teal}{HTML}{20B2AA}  

\begin{axis}[
height=\axisheight,
legend cell align={left},
legend style={
    at={(0.75, .9)}, 
    anchor=north east, 
    draw=lightgray, 
    fill opacity=0.8, 
    draw opacity=1,
    text opacity=1,
    font=\small
},
minor x tick num=1,
minor y tick num=0,
tick align=outside,
tick pos=left,
width=\axiswidth,
x grid style={darkgray},
xlabel={$t$ (time steps)},
xmajorgrids,
xmin=0.0, 
xmax=1800,
xtick distance=282,
xtick style={color=black},
y grid style={darkgray},
ymajorgrids,
ytick style={color=black},
ylabel={$r_i$ in \unit{\kilo\metre}},
label style={font=\small},
tick label style={font=\footnotesize}
]
\addplot [gray, dotted, line width=0.4mm] coordinates {(-10, 6848.23395327399) (2000, 6848.23395327399)};

\addplot[vivid_blue, line width=0.40mm] table [x=t, y=r] {./plotdata/satellite_constellation/sat_A1_r.tex};

\addplot[bright_green, line width=0.40mm, dash pattern=on 8pt off 4pt,] table [x=t, y=r] {./plotdata/satellite_constellation/sat_A2_r.tex};

\addplot[vibrant_orange, line width=0.40mm, dash pattern=on 8pt off 4pt on 4pt off 4pt,] table [x=t, y=r] {./plotdata/satellite_constellation/sat_A3_r.tex};

\addplot[tomato_red, line width=0.40mm, dash pattern=on 8pt off 1pt on 2pt off 1pt on 2pt off 1pt,] table [x=t, y=r] {./plotdata/satellite_constellation/sat_A4_r.tex};

\addplot[medium_purple, line width=0.40mm, dash pattern=on 4pt off 2pt on 1pt off 2pt,] table [x=t, y=r] {./plotdata/satellite_constellation/sat_A5_r.tex};

\end{axis}
\end{tikzpicture}
    \caption{Orbital radii of the satellites. A grey dotted line marks the initial radius of the satellites. At $t=282$ from the top: $r_1$, $r_2$, $r_3$, $r_4$, and $r_5$.}
    \label{fig:satellite_states}
\end{figure}
Before analysing the theoretical properties of the proposed distributed MPC scheme, we illustrate the scheme with an example inspired by~\cite{Sin2020_PassivityBasedDistributedAcquisition, Pippia2022_Reconfigurationsatelliteconstellation}, where we reconfigure a satellite constellation.
The dynamics of the satellites are given in polar coordinates (cf.~\cite{Sin2020_PassivityBasedDistributedAcquisition}):
\begin{align*}
    \dot{r}_i &= v_i, & \dot{v}_i &= r_i \omega_i^2 - \frac{\mu}{r_i^2} + \frac{F_{r,i}}{m_i},\\
    \dot{\vartheta}_i &= \omega_i, & \dot{\omega}_i &= \frac{-2v_i\omega_i}{r_i} + \frac{F_{\vartheta,i}}{m_i r_i}, 
\end{align*}
where $r_i$ is the orbital radius, $\vartheta_i$ the angular position, $v_i$ and $\omega_i$ are the respective velocities, $F_{r,i}$ and $F_{\vartheta,i}$ are the radial and tangential thrusts, $\mu$ is the standard gravitational parameter (we choose Earth's $\mu = \SI{3.986e14}{\metre^3\per\second^2}$) and $m_i= \SI{200}{\kilo\gram}$ is the satellites' mass.
We have $x_i = \begin{bmatrix} r_i & \vartheta_i & v_i & \omega_i \end{bmatrix}$ and $u_i = \begin{bmatrix} F_{r,i} & F_{\vartheta,i} \end{bmatrix}$.
The continuous-time dynamics are discretized using the Runge-Kutta method (RK4) with a step size of $\SI{120}{\second}$. 
The cooperative task is to achieve a constellation with an angular difference of \SI{45}{\degree} on an orbit with a periodicity of $T = 47$.
The stage cost is chosen as a quadratic stage cost $\ell_i(x_i, u_i, x_{T,i}, u_{T,i}) = \Vert x_i -  x_{T,i}\Vert_Q^2 + \Vert u_i -  u_{T,i}\Vert_R^2$ where $Q$ is a diagonal matrix with diagonal entries $[2 \times 10^{-8}, 0.02, 0.0002, 0.002]$ and $R=10^{-7}I$.
We constrain the cooperation state and cooperation input trajectories such that $r_{T,i}(\tau) \approx \SI{6848.234}{\kilo\metre}$, $v_{T,i}(\tau) = \SI{0}{\frac{\metre}{\second}}$, $\omega_{T,i}(\tau) = \sqrt{\frac{\mu}{r_{T,i}^3(\tau)}}$, and $F_{T,r,i}(\tau) = F_{T,\vartheta,i}(\tau) = \SI{0}{\newton}$, which yields an orbit with the required periodicity without applying thrust.
To achieve the cooperative task, we choose $y_i(t) = \vartheta_i(t)$.
We define $\Delta\vartheta_{ij}(\tau) = y_{T,j}(\tau) - y_{T,i}(\tau) - \mathrm{rad}\,(45)$ for $j<i$ and $\Delta\vartheta_{ij}(\tau) = y_{T,i}(\tau) - y_{T,j}(\tau) - \mathrm{rad}\,(45)$ for $j>i$, where $\mathrm{rad}$ transforms degrees to radians.
We assign indices counter-clockwise on the orbit, i.e. Satellite 1 starts with the smallest $\vartheta_i$, and always define adjacent satellites on the orbit as neighbours.
Then, we choose 
$W^{\mathrm{c}}(y_T) = \frac{1}{2}\sum_{i=1}^{m} \sum_{\tau=0}^{T-1} \sum_{j\in\mathcal{N}_i} \vert \mathcal{N}_i \vert 
\Delta\vartheta_{ij}(\tau)^2  $,
where $\vert \mathcal{N}_i \vert$ is the number of neighbours of Satellite $i$.
Moreover, we impose as input constraints $\change{\vert} F_{r,i} \change{\vert}, \change{\vert} F_{\vartheta,i} \change{\vert} \le \SI{0.237}{\milli\newton}$.
State constraints are also assumed, but are unimportant for the simulation results.
We choose $N=3T$, use terminal equality constraints, and start with $m=5$ satellites with $r_i(0) \approx \SI{6848.234}{\kilo\metre}$, $\vartheta_i(0) = i \cdot \mathrm{rad}\, \SI{25}{\degree}$, $v_{i}(0) = 0$, and $\omega_i(0) = \sqrt{\frac{\mu}{r_i(0)^3}}$. 
Moreover, we use $V_i^{\Delta}(y_{T,i}, y_{T,i}^{\mathrm{pr}}) = \frac{1}{10^4T}\sum_{\tau=0}^{T-1} \Vert y_{T,i}(\tau) - y_{T,i}^{\mathrm{pr}}(\tau) \Vert^2$.
After $750$ steps, we deorbit Satellite 2 and Satellite 4, i.e. remove them from the system. 
The differences in the angular position are shown in Figure~\ref{fig:satellite_angular_position}, and Figure~\ref{fig:satellite_states} shows the orbital radius.
The satellites start transferring to orbits with the desired angular difference, and adapt without issue after two satellites are deorbited.
The scheme allowed for this easy transition since each agent's constraints were unaffected by the change in the topology and the cooperation objective function was designed with respect to neighbours.
Note that no component redesign was necessary during runtime, except that the communication topology was updated.

Because the proposed distributed MPC scheme explicitly handles periodic cooperative tasks, it was easily possible to include the angular position in the satellites' model.
This allowed us to promote directly a desired angular difference between agents.
For a scheme that only deals with cooperative tasks at equilibria, this is difficult, cf.~\cite{Sin2020_PassivityBasedDistributedAcquisition,Pippia2022_Reconfigurationsatelliteconstellation}.
Furthermore, collision avoidance constraints, while unnecessary in this simulation, could have been easily incorporated as $\Vert \vartheta_i(t) - \vartheta_j(t) \Vert \ge \vartheta_{\mathrm{min}}$.

\subsection{Guaranteed achievement of the cooperative task}\label{sec:DMPC_with_terminal_closed_loop_analysis}
We now turn to the analysis of the closed-loop system~\eqref{eq:central_global_closed_loop}.
In the following, we provide general conditions on the design of the stage cost, the cooperation objective function, the penalty on the change in the cooperation output, and the local sets of admissible cooperation outputs.
These are deliberately stated in general terms to constrain their design as little as possible.
After formally stating asymptotic stability and achievement of the cooperative task, we show in Section~\ref{ssec:sufficient_design_central} below that some intuitive choices satisfy these general conditions.

We begin with additional conditions on the cooperation objective function.
We also want to analyse the case in which the cooperative task, characterized by $\mathcal{Y}_T^{\mathrm{c}}$, cannot be achieved, for example, due to constraints or if $\mathcal{Y}_T^{\mathrm{c}} \not\subseteq \mathcal{Y}_{T}$.
For this purpose, we define the best achievable fulfilment of the cooperative task 
\begin{align}
    \smash{\mathcal{Y}_T^{W}} = \argmin_{y_T\in\mathcal{Y}_{T}} W^{\mathrm{c}}(y_T).
\end{align}
\change{Hence, $y_T$ solves the cooperative task as well as possible if $y_T \in \mathcal{Y}_T^W$.}

The following assumption states the existence of a cooperation output that reduces the cooperation objective function if the cooperative task has not been achieved as well as possible.
\begin{assum}\label{assm:better_cooperation_candidate}
    There exist $\omega > 1$, a continuous function $\psi: \mathcal{Y}_T \to \mathbb{R}_{\ge 0}$ positive definite with respect to $\mathcal{Y}_T^W$, and $c_{\psi} > 0$ such that for any $y_T \in \mathcal{Y}_T$ and $\theta\in [0, 1]$ there exists $\hat{y}_T \in \mathcal{Y}_T$ with
    \begin{subequations}
        \begin{align}
            \vert \hat{y}_T \vert_{y_T} &\le \theta c_{\psi}\psi(y_T),
            \label{eq:better_cooperation_candidate_a}
        \\
            W^{\mathrm{c}}(\hat{y}_T) - W^{\mathrm{c}}(y_T) &\le - \theta \psi(y_T)^{\omega}.
            \label{eq:better_cooperation_candidate_b}
        \end{align}
    \end{subequations}
\end{assum}

This assumption imposes a certain growth condition on $W^{\mathrm{c}}$ and a structure on $\mathcal{Y}_T$, for example (but not limited to) convexity, cf.~\cite[Assm. 7]{MatthiasKohler2023_TransientPerformanceMPC}, \cite[Assm. 5]{Soloperto2022} for a similar assumption with the common choice of $\omega = 2$.
Intuitively, it tells us that for any $y_T$ not solving the cooperative task as well as possible, we can find a better cooperation output $\hat{y}_T$ that reduces the cost in proportion to the distance of $y_T$ from solving the cooperative task.

The impact of \change{a change in the cooperation output} on the stage cost is captured by the following assumption, which relates two cooperation outputs to each other, cf.~\cite[Assm. 3]{MatthiasKohler2023_TransientPerformanceMPC} and \cite[Assm. 1]{Soloperto2022}. \change{Recall that $r_{T,i}$ is uniquely determined by $y_{T,i}$ according to Assumption~\ref{assm:unique_corresponding_equilibrium}.}
\begin{assum}\label{assm:stage_cost_comparison}
    There exist $\omega > 1$ and $c_1^{\ell_i}, c_2^{\ell_i} > 0$ satisfying
    \begin{equation}\label{eq:stage_cost_comparison}
        \ell_i(x_i, u_i, \hat{r}_{T,i}(\tau)) \le c_1^{\ell_i} \ell_i(x_i, u_i, r_{T,i}(\tau)) + c_2^{\ell_i} \vert \hat{r}_{T,i} \vert_{r_{T,i}}^{\omega}
    \end{equation}
    for all $y_T, \hat{y}_T \in \mathcal{Y}_T$, $(x_i, u_i) \in Z_i$ and $\tau\in\mathbb{I}_{0:T-1}$.
\end{assum}

This is important to trade off an increase in the tracking cost, caused by moving the reference, \change{against} a decrease in the cooperation objective function.
Note that Assumption~\ref{assm:stage_cost_comparison} holds with $\omega = 2$ for quadratic stage costs on bounded sets, cf.~\cite{Soloperto2022}.
\change{Since quadratic stage costs on bounded sets are standard in the MPC literature, Assumption~\ref{assm:stage_cost_comparison} is not restrictive.}

Furthermore, the penalty functions on the change in the cooperation output, $V_i^\Delta$, need to satisfy the following conditions.
\begin{assum}\label{assm:penalty_function}
    The functions $V_i^{\Delta}$ are continuous.
    Moreover, there exist $\omega > 1$, $c^{\Delta} > 0$ and $\alpha_{\mathrm{lb}}^{\Delta}, \alpha_{\mathrm{ub}}^{\Delta} \in \mathcal{K}_\infty$ such that for any $\hat{y}_{T}, y_{T}, y_{T}^{\mathrm{pr}} \in \mathcal{Y}_{T}$,
    \begin{subequations}
        \begin{align}
            &\alpha_{\mathrm{lb}}^{\Delta}(\vert y_T \vert_{y_T^{\mathrm{pr}}}) \le \sum_{i=1}^{m}V_i^{\Delta}(y_{T,i}, y_{T,i}^{\mathrm{pr}}) \le \alpha_{\mathrm{ub}}^{\Delta}(\vert y_T \vert_{y_T^{\mathrm{pr}}}),
            \label{eq:penalty_cooperation_change_distance}
            \\
            &\sum_{i=1}^{m} V_i^{\Delta}(\hat{y}_{T,i}, y_{T,i}^{\mathrm{pr}}) - 2V_i^{\Delta}(y_{T,i}, y_{T,i}^{\mathrm{pr}}) \le c^{\Delta} \vert \hat{y}_T \vert_{y_T}^{\omega}.
            \label{eq:penalty_cooperation_change_decrease}
        \end{align}
    \end{subequations}
\end{assum}

Condition~\eqref{eq:penalty_cooperation_change_decrease} is similar to~\eqref{eq:stage_cost_comparison} and limits the growth rate of $\sum_{i=1}^{m} V_i^\Delta$. As with Assumption~\ref{assm:stage_cost_comparison}, this means that the resulting decrease in the cooperation objective function can beat the penalty incurred when changing the cooperation output. 
As shown in  Section~\ref{ssec:sufficient_design_central}, a simple quadratic penalty function satisfies this assumption for $\omega = 2$.
\change{As mentioned above, since $\omega = 2$ corresponds to the standard case of quadratic costs on bounded sets, Assumption~\ref{assm:penalty_function} is not restrictive.}

Based on the stated assumptions, we proceed to establish closed-loop constraint satisfaction and stability, which results in closed-loop fulfilment of the cooperative task as well as possible.
First, we prove that~\eqref{eq:central_OP} is recursively feasible, and the constraints are satisfied in closed loop.
\begin{thm}\label{thm:recursive_feasibility}
    Let Assumptions~\ref{assm:unique_corresponding_equilibrium}, \ref{assm:terminal_ingredients}, and~\ref{assm:tightened_coupling_constraints} hold.
    Then, for any initial condition $x_0$ for which~\eqref{eq:central_OP} is feasible,~\eqref{eq:central_OP} is feasible for all future time steps of the closed-loop system~\eqref{eq:central_global_closed_loop}.
    Consequently,~\eqref{eq:central_global_closed_loop} satisfies the constraints, i.e. $\big(x_{i,{\mu_i}}(t), \mu_i(\xi(t))\big) \in Z_i$ and $\big(x_{i,{\mu_i}}(t), x_{\mathcal{N}_i,{\mu_i}}(t)\big)\in \mathcal{C}_i$ for all $t\in\mathbb{N}_0$.
\end{thm}
\begin{proof}
    At $t=0$,~\eqref{eq:central_OP} is feasible.
    Assume for $t\in\mathbb{N}$, that~\eqref{eq:central_OP} was feasible at $t-1$.
    Then, the shifted previously optimal cooperation output $y_T^0(\cdot+1\vert \xi(t-1))$ is a feasible candidate solution of~\eqref{eq:central_OP}.
    Due to Assumptions~\ref{assm:terminal_ingredients}, and~\ref{assm:tightened_coupling_constraints}, a corresponding feasible input sequence is given by shifting the previously optimal one and appending the terminal controller (as is standard in MPC, cf.~\cite{Gruene2017,Rawlings2020}), i.e. $\big(u_i^0(1\vert \xi(t-1)), \dots, u_i^0(N-1 \vert \xi(t-1)), k_i^{\mathrm{f}}\big(x_{i, u_i^0(\cdot\vert\xi(t-1))}(N, x_i(t-1)), r_T^0(N\vert \xi(t-1))\big)\big)$.
    Constraint satisfaction of the closed loop follows from the definition of $\mathbb{U}_i(x_i(t))$ and~\eqref{eq:central_OP_coupling_constraints} with $k=0$.
\end{proof}

The following lemma shows that if the agents are sufficiently close to a cooperation reference, then the stage cost upper bounds the tracking part of the objective function~\eqref{eq:central_OP_objective}.
This is later useful to bound the increase in the tracking part when the cooperation output is incrementally moved.
\begin{lem}\label{lem:tracking_cost_upper_bound}
    Let Assumptions~\ref{assm:unique_corresponding_equilibrium}--\ref{assm:tightened_coupling_constraints} hold.
    Consider an optimization problem similar to~\eqref{eq:central_OP}, but with $y_T$ (and corresponding $r_{T} = (x_T, u_T)$) fixed as a parameter.
    Then, there exists $\varepsilon > 0$ such that for any $x \in X$ with $\sum_{i=1}^{m} \ell_i'(x_i, r_{T,i}(0)) \le \varepsilon$,
    this optimization problem is feasible and its solution $\tilde{u}$ satisfies with $c_i^{\mathrm{f}}$ from Assumption~\ref{assm:terminal_ingredients}
    \begin{equation}\label{eq:tracking_cost_upper_bound}
        \sum_{i=1}^{m} J_i^{\mathrm{tr}}(x_i, \tilde{u}_i, r_{T,i}) \le \sum_{i=1}^{m} c_i^{\mathrm{f}}\ell_i'(x_i, r_{T,i}(0)).
    \end{equation}
\end{lem}
\begin{proof}
    From Assumptions~\ref{assm:stage_cost_lower_and_upper_bound} and \ref{assm:terminal_ingredients}, if $c_i^{\mathrm{b}} > 0$ for all $i\in\mathbb{I}_{1:m}$, then choosing $\varepsilon>0$ such that $(\alpha_{\textrm{lb}}^{\ell_i})^{-1}(\varepsilon) \le c_i^{\mathrm{b}}$ for all $i\in\mathbb{I}_{1:m}$ implies $x_i \in \mathcal{X}_i^{\mathrm{f}}(r_{T,i}(0))$.
    The following steps are standard, cf.~\cite[Prop. 2.35]{Rawlings2020}.
    Within the terminal set, due to Assumptions~\ref{assm:terminal_ingredients} and~\ref{assm:tightened_coupling_constraints}, the terminal control law $k_i^{\mathrm{f}}$ generates a feasible input sequence in the above considered optimization problem.
    Thus, from~\eqref{eq:terminal_cost_decrease}, the terminal costs provide an upper bound on $J_i^{\mathrm{tr}}(x_i, u_i, r_{T,i})$.
    Finally, applying~\eqref{eq:terminal_cost_upper_bound} yields the claimed bound.
    Otherwise, if $c_i^{\mathrm{b}} = 0$ for some $i\in\mathbb{I}_{1:m}$, i.e. terminal equality constraints are used, then the claim follows from~\cite[Prop. 4]{jkoehler_nonlinear_dynamic_2020} and the arguments above for all other agents with $c_i^{\mathrm{b}} > 0$.
\end{proof}

The following theorem states that the stage costs and the cost of changing the cooperation outputs upper bound the cooperation outputs' distance to $\mathcal{Y}_T^W$, i.e. achieving the cooperative task as well as possible.
Hence, one can see the cooperation output in each time step as an intermediate goal towards achieving the cooperative task.
A similar result has been shown in~\cite{Soloperto2022} for tracking of externally given references without terminal constraints and in~\cite{MatthiasKohler2023_TransientPerformanceMPC} with terminal constraints, both without a dependence on $y_{T}^{\mathrm{pr}}$.
\begin{thm}\label{thm:stage_cost_upper_bounds_cooperation_distance}
    Let Assumptions~\ref{assm:compact_cooperation_sets}--\ref{assm:penalty_function} hold with the same $\omega > 1$.
    For any $N\in\mathbb{N}_0$ there exists $\eta_\ell \in \mathcal{K}$ such that for any $x\in\mathcal{X}_N$ and $y_{T}^{\mathrm{pr}} \in \mathcal{Y}_{T}$ the inequality
    \begin{align}\label{eq:stage_cost_upper_bounds_cooperation_distance}
        &\sum_{i=1}^m \ell_i(x_i, \mu_i(\xi), r_{T,i}^0(0 \vert \xi)) + \lambda(N)V_i^{\Delta}(y_{T,i}^0(\cdot \vert \xi), y_{T,i}^{\mathrm{pr}}) \notag\\
        &\ge \eta_\ell(\vert y_T^0(\cdot \vert \xi) \vert_{\mathcal{Y}_T^W})
    \end{align}
    holds. If $\xi = \xi(0)$, then~\eqref{eq:stage_cost_upper_bounds_cooperation_distance} holds also with $V_i^{\Delta} = 0$.
\end{thm}
\begin{proof}
    We prove this by contradiction.
    Let $N\in\mathbb{N}_0$-
    Abbreviate $y_T^{0} = y_T^{0}(\cdot \vert \xi)$ and $\bar{V}_{i}^{\Delta}=\lambda(N)V_i^{\Delta}$. 
    Suppose for all $\eta_\ell \in \mathcal{K}$ there exist $x\in\mathcal{X}_N$ and $y_{T}^{\mathrm{pr}}\in \mathcal{Y}_{T}$ such that
    \begin{align}\label{eq:stage_cost_upper_bounds_cooperation_distance_contradiction}
        &\smash[b]{\sum_{i=1}^m} \ell_i(x_i, \mu_i(\xi), r_{T,i}^0(0 \vert \xi)) + \bar{V}_{i}^{\Delta}(y_{T,i}^0, y_{T,i}^{\mathrm{pr}}) < \eta_\ell(\vert y_T^0 \vert_{\mathcal{Y}_T^W}).
    \end{align}
    Consider $\hat{y}_T = \hat{y}_T(\cdot\vert\xi)$ from Assumption~\ref{assm:better_cooperation_candidate}, based on $y_T^{0}$.
    Since $\mathcal{Y}_T$ is compact and $\psi$ continuous, there exists $\gamma_W = \sup_{y_T \in \mathcal{Y}_T} \vert y_T \vert_{\mathcal{\mathcal{Y}}_T^W}$, and $\gamma_{\psi} = \sup_{y_T \in \mathcal{Y}_T} \psi(y_T)$.
    Define $L_i = \max(L_{x,i}, L_{u,i})$, and $c^{\ell} = \max_i(c_1^{\ell_i}, c_2^{\ell_i}L_i^{\omega})$, then
    \begin{align}
        &\sum_{i=1}^{m} \ell'(x_i, \hat{r}_{T,i}(0\vert\xi)) 
        \le \sum_{i=1}^{m} \ell(x_i, \mu_i(\xi), \hat{r}_{T,i}(0\vert\xi))
        \notag
    \\
        &\stackrel{\text{Rem.}\,\ref{rem:unique_corresponding_equilibrium},\,\eqref{eq:stage_cost_comparison}}{\le} \sum_{i=1}^{m} c_1^{\ell_i} \ell_i(x_i, \mu_i(\xi), r_{T,i}^0(0 \vert \xi)) + c_2^{\ell_i}L_i^{\omega} \vert \hat{y}_{T,i} \vert_{y_{T,i}^0}^{\omega}
        \notag
    \\
        &\le c^{\ell} \big( \vert \hat{y}_{T} \vert_{y_{T}^0}^{\omega} + \smash{\sum_{i=1}^{m}}  \ell_i(x_i, \mu_i(\xi), r_{T,i}^0(0 \vert \xi)) \big)
        \notag
    \\
        &\stackrel{\mathclap{\eqref{eq:stage_cost_upper_bounds_cooperation_distance_contradiction},\eqref{eq:better_cooperation_candidate_a}}}{<}
        \; c^{\ell} \big( \eta_{\ell}(\vert y_T^0 \vert_{\mathcal{Y}_T^W}) - \smash[b]{\sum_{i=1}^{m}}\bar{V}_i^{\Delta}(y_{T,i}^0, y_{T,i}^{\mathrm{pr}}) + \theta^{\omega} c_{\psi}^{\omega} \psi(y_T^0)^{\omega} \big)
        \label{eq:stage_cost_upper_bounds_cooperation_distance_intermed_a}
    \\
        &\le c^{\ell} \big( \eta_{\ell}(\gamma_W) + \theta^{\omega} c_{\psi}^{\omega}\gamma_{\psi}^{\omega} \big)
        \le \varepsilon,
        \notag
    \end{align}
    with $\varepsilon$ from Lemma~\ref{lem:tracking_cost_upper_bound}, and where the last inequality follows if $\theta^{\omega} \le \frac{\varepsilon}{2c^{\ell}c_{\psi}^{\omega}\gamma_{\psi}^{\omega}}$, and $\eta_{\ell}$ such that $\eta_{\ell}(\gamma_W) \le \frac{\varepsilon}{2c^{\ell}}$.
    Hence, Lemma~\ref{lem:tracking_cost_upper_bound} implies the existence of a feasible candidate $(\hat{u}, \hat{y}_T)$, and defining $c^{\mathrm{f}} = \max_i(c_i^{\mathrm{f}})$, we have from Lemma~\ref{lem:tracking_cost_upper_bound}
    \begin{align}
        &\sum_{i=1}^{m} J_i^{\mathrm{tr}}(x_i, \hat{u}_i, \hat{r}_{T,i}(\cdot\vert\xi)) \stackrel{\eqref{eq:tracking_cost_upper_bound}}{\le} \sum_{i=1}^{m} c_i^{\mathrm{f}} \ell_i'(x_i, \hat{r}_{T,i}(0\vert\xi))
        \notag
        \\
        &\stackrel{\eqref{eq:stage_cost_upper_bounds_cooperation_distance_intermed_a}}{<} 
        c^{\mathrm{f}} c^{\ell} \big( \eta_{\ell}(\vert y_T^0 \vert_{\mathcal{Y}_T^W}) - \smash[b]{\sum_{i=1}^{m}} \bar{V}_i^{\Delta}(y_{T,i}^0, y_{T,i}^{\mathrm{pr}}) + \theta^{\omega} c_{\psi}^{\omega} \psi(y_T^0)^{\omega} \big).
        \label{eq:tracking_cost_upper_bound_proof_aux}
    \end{align}
    Now, compare the cost of $(\hat{u}, \hat{y}_T)$ with that of $(u^0, y_T^0) = (u^0(\cdot\vert \xi), y_T^0(\cdot\vert \xi))$. 
    Define $c^{\mathrm{f}\ell} = \max(c^{\mathrm{f}}c^{\ell}, 1)$. 
    Since $J_i^{\mathrm{tr}}$ is non-negative, and with $c_N^{\Delta} = \lambda(N)c^{\Delta}$ as well as $\theta_N = \lambda(N)\theta$,
    \begin{align}\label{eq:stage_cost_upper_bounds_cooperation_distance_before_combination}
        &{\sum_{i=1}^{m}} J_i(x_i, \hat{u}_i, \hat{y}_{T,i}, y_{T,i}^{\mathrm{pr}}, \hat{y}_{T, \mathcal{N}_i}) - J_i(x_i, u_i^0, y_{T,i}^0, y_{T,i}^{\mathrm{pr}}, y_{T,\mathcal{N}_i}^0) 
        \notag\\
        &\le \smash[t]{\sum_{i=1}^{m}} J_i^{\mathrm{tr}}(x_i, \hat{u}_i, \hat{r}_{T,i}(\cdot\vert\xi)) \hspace{-0.5pt} + \hspace{-0.5pt} \bar{V}_i^{\Delta}(\hat{y}_{T,i}, y_{T,i}^{\mathrm{pr}}) \hspace{-0.5pt} - \hspace{-0.5pt} \bar{V}_i^{\Delta}(y_{T,i}^0, y_{T,i}^{\mathrm{pr}})
        \notag\\
        &\phantom{\le} + \lambda(N) \big(W^{\mathrm{c}}(\hat{y}_T) - W^{\mathrm{c}}(y_T^0)\big)
        \notag\\
        &\stackrel{\eqref{eq:tracking_cost_upper_bound_proof_aux}}{<} c^{\mathrm{f}\ell} \big( \eta_{\ell}(\vert y_T^0 \vert_{\mathcal{Y}_T^W}) - \smash[t]{\sum_{i=1}^{m}} \bar{V}_i^{\Delta}(y_{T,i}^0, y_{T,i}^{\mathrm{pr}}) + \theta^{\omega} c_{\psi}^{\omega} \psi(y_T^0)^{\omega} \big) \notag\\
        &\phantom{\stackrel{\eqref{eq:tracking_cost_upper_bound_proof_aux}}{<}} + c^{\mathrm{f}\ell}\big( \smash{\sum_{i=1}^{m}} \bar{V}_i^{\Delta}(\hat{y}_{T,i}, y_{T,i}^{\mathrm{pr}}) - \bar{V}_i^{\Delta}(y_{T,i}^0, y_{T,i}^{\mathrm{pr}}) \big) \notag\\
        &\phantom{\stackrel{\eqref{eq:tracking_cost_upper_bound_proof_aux}}{<}} + \lambda(N) \big(W^{\mathrm{c}}(\hat{y}_T) - W^{\mathrm{c}}(y_T^0)\big)
        \notag\\
        &\stackrel{\eqref{eq:better_cooperation_candidate_b}}{\le} c^{\mathrm{f}\ell} \big( \eta_{\ell}(\vert y_T^0 \vert_{\mathcal{Y}_T^W}) + \theta^{\omega} c_{\psi}^{\omega} \psi(y_T^0)^{\omega} \big) - \theta_N \psi(y_T^0)^{\omega} \notag\\
        &\phantom{\stackrel{\eqref{eq:tracking_cost_upper_bound_proof_aux}}{<}} + c^{\mathrm{f}\ell}\big( \smash{\sum_{i=1}^{m}} \bar{V}_i^{\Delta}(\hat{y}_{T,i}, y_{T,i}^{\mathrm{pr}}) - 2\bar{V}_i^{\Delta}(y_{T,i}^0, y_{T,i}^{\mathrm{pr}}) \big)
        \notag\\
        &\stackrel{\eqref{eq:penalty_cooperation_change_decrease}}{\le} \hspace{-0.1em} c^{\mathrm{f}\ell} \big( \eta_{\ell}(\vert y_T^0 \vert_{\mathcal{Y}_T^W}) \hspace{-2pt} + \hspace{-2pt} c_N^{\Delta} \vert \hat{y}_T \vert_{y_T^0}^{\omega} \big) \hspace{-2pt} + \hspace{-2pt} (c^{\mathrm{f}\ell}c_{\psi}^{\omega}\theta^{\omega} \hspace{-2pt} - \hspace{-2pt} \theta_N) \psi(y_T^0)^{\omega}
        \notag\\
        &\stackrel{\eqref{eq:better_cooperation_candidate_a}}{\le} \hspace{-0.1em} c^{\mathrm{f}\ell} \eta_{\ell}(\vert y_T^0 \vert_{\mathcal{Y}_T^W})  + \big(c^{\mathrm{f}\ell} c_{\psi}^{\omega}(c_N^{\Delta} + 1) \theta^{\omega} - \theta_N \big)\psi(y_T^0)^{\omega}
        \notag\\
        &\le c^{\mathrm{f}\ell} \big(\eta_{\ell}(\vert y_T^0 \vert_{\mathcal{Y}_T^W}) - c_{\theta}^{W} \psi(y_T^0)^{\omega}\big)
    \end{align}
    with $c_{\theta}^{W} = \theta \big(\lambda(N) - c^{\mathrm{f}\ell} c_{\psi}^{\omega}(c_N^{\Delta}+1)\theta^{\omega-1} \big)$, which is positive if $\theta^{\omega-1} < \frac{\lambda(N)}{c^{\mathrm{f}\ell}c_{\psi}^{\omega}(c_{\change{N}}^{\Delta}+1)}$.
    Since $(\cdot)^{\omega} \circ \psi$ is continuous positive definite with respect to $\mathcal{Y}_T^W$ on $\mathcal{Y}_T$, and $\mathcal{Y}_T$ is compact, there exists $\tilde{\eta}_{\ell}\in\mathcal{K}$ (cf.~\cite{Kellett2014}) such that 
    $\psi(y_T)^{\omega} \ge \tilde{\eta}_{\ell}(\vert y_T \vert_{\mathcal{Y}_T^W})$.
    Finally, choosing $\eta_{\ell} = \delta\frac{c_{\theta}^{W}}{2} \tilde{\eta}_{\ell}$,
    with $\delta \in (0, 1]$ so that the earlier condition $\eta_{\ell}(\gamma_W) \le \frac{\epsilon}{2c^{\ell}}$ holds,
    yields 
    \begin{align*}
        &\smash[b]{\sum_{i=1}^{m}} J_i(x_i, \hat{u}_i, \hat{y}_{T,i}, y_{T,i}^{\mathrm{pr}}, \hat{y}_{T, \mathcal{N}_i}) - J_i(x_i, u_i^0, y_{T,i}^0, y_{T,i}^{\mathrm{pr}}, y_{T,\mathcal{N}_i}^0)
        \\
        &< - \frac{c^{\mathrm{f}\ell}c_{\theta}^{W}}{2} \tilde{\eta}_{\ell}(\vert y_T^0 \vert_{\mathcal{Y}_T^W}) \le 0.
    \end{align*}
    Hence, the objective function of $(\hat{u}, \hat{y}_T)$ is \change{less} than that of $(u^0(\cdot\vert \xi), y_T^0(\cdot\vert \xi))$, which is a contradiction.
    Furthermore, if $\xi = \xi(0)$, then~\eqref{eq:stage_cost_upper_bounds_cooperation_distance} follows from the same derivation with $V_i^{\Delta} = 0$ since $V_i^{\Delta}$ is omitted in~\eqref{eq:central_OP} in this case.
\end{proof}

In the following, we will use Lyapunov-based arguments to show stability of a set on which the cooperative task is achieved as well as possible for the closed-loop system~\eqref{eq:central_global_closed_loop}.
However, we cannot use a Lyapunov function that depends solely on the state $x$, since the closed-loop input is also influenced by $y_T^{\mathrm{pr}}$.
Hence, for the same state $x$, different values of $y_T^{\mathrm{pr}}$ in general result in different inputs.
This situation bears some resemblance to that encountered in suboptimal MPC, where the optimization problem is not solved to global optimality and may return different inputs for the same state (e.g. based on a different initial guess); see, e.g.~\cite{Allan2017}, \cite[Sec. 2.7]{Rawlings2020}.
We use a Lyapunov candidate based on $\xi = (x, y_T^{\mathrm{pr}})$.
Consider the function 
$V(\xi) = \mathcal{J}(\xi) - W^{\mathrm{c}}_N$
where $W^{\mathrm{c}}_N = \lambda(N)W^{\mathrm{c}}_0$ with
\begin{equation}\label{eq:def_best_reachable_cost}
    W^{\mathrm{c}}_0 = \smash{\min_{y_T\in\mathcal{Y}_{T}} W^{\mathrm{c}}(y_T)},
\end{equation}
i.e. $W^{\mathrm{c}}_0 = W^{\mathrm{c}}(y_T)$ for all $y_T\in \mathcal{Y}_T^W$.
Moreover, define $\xi_T(k) = [\xi(k)^\top, \dots, \xi(k+T-1)^\top]^\top$, i.e. $T$ instances of $\xi$ collected.
Then, our Lyapunov function candidate is 
\begin{equation*}
    V_T(\xi_T(t)) = \sum_{\tau=t}^{t+T-1} V(\xi(\tau)),
\end{equation*}
and we simply write $V_T(\xi_T)$ if no specific start time of the sequence is considered.

Define the set of feasible state and input trajectories yielding an output achieving the cooperative task as well as possible:
\begin{align*}
    \mathcal{Z}_T^W = \{ &r_T \in \textstyle{\prod}_{i=1}^{m} \mathcal{Z}_{T,i} \mid (x_{T,i}(\tau), x_{T,\mathcal{N}_i}(\tau))\in\mathcal{C}_i\ominus\mathcal{B}_{\eta_i}, \\
    &h(x_T(\tau), u_T(\tau)) \in \mathcal{Y}_T^W,\;\forall \tau \in \mathbb{I}_{0:T-1},\; \forall i \in \mathbb{I}_{1:m} \}
\end{align*}
with $\eta_i$ from Assumption~\ref{assm:tightened_coupling_constraints}.
Let $\mathcal{X}_T^W$ denote the projection of $\mathcal{Z}_T^W$ onto $\textstyle{\prod}_{i=1}^{m} X_i^{T}$, and define the set 
\begin{align}\label{eq:def_stable_set}
    &\Xi_T^W = \{x_T^W \in \mathcal{X}_T^W,\; y_{T}^{W} \in (\mathcal{Y}_T^W)^{T} \mid \notag\\
    &\hspace{2em}\exists u_T^W\in \textstyle{\prod}_{i=1}^{m} U_i^{T},\; (x_T^W, u_T^W) \in \mathcal{Z}_T^W,\\
    &\hspace{2em} h(x_T^W(\tau), u_T^W(\tau)) = y_{T,k+1}^{W}(\tau-k),\; \forall \tau, k \in \mathbb{I}_{0:T-1} \notag \}.
\end{align}
Note that $y_{T}^{W}$ is a stacked vector of $T$ periodic trajectories, each of length $T$, and $y_{T,k}^{W}$ denotes the $k$-th of these.
Definition~\eqref{eq:def_stable_set} is such that each of these periodic trajectories is shifted by one time step, i.e. $y_{T,k+1}^W(\tau) = y_{T, k}^W(\tau+1)$ for $k\in\mathbb{I}_{1:T}$ and $y_{T, 1}^W(\tau) = y_{T, T}^W(\tau+1)$.
Recall $y_{T}^{\mathrm{pr}}(\cdot \vert t+1) = y_T^0(\cdot+1\vert\xi(t))$ for $t\in\mathbb{N}_0$.
Then, our goal is to show asymptotic stability of $\Xi_T^W$ for the (extended) closed-loop system
\begin{align}\label{eq:extended_closed_loop}
    &\xi_T(t+1) =
    \begin{bmatrix}
        f(x(t), \mu(\xi(t))) \\
        y_{T}^{0}(\cdot+1 \vert \xi(t)) \vspace{-0.5em}\\
        \vdots \\
        f(x(t+T-1), \mu(\xi(t+T-1))) \\
        y_{T}^{0}(\cdot+1 \vert \xi(t+T-1)) \\
    \end{bmatrix}
\end{align}
with $x(0) = x_0$ and $y_{T}^{\mathrm{pr}}(\cdot \vert 0)$ arbitrary.
That is, a state trajectory for which there exists an input trajectory such that they are feasible and generate a periodic output trajectory satisfying the cooperative task as well as possible.
Moreover, this output is equal to the $T$ shifted previous cooperation outputs, thus, the closed-loop state generates this output.

First, we show an upper bound on $V_T(\xi_T)$.
\begin{lem}\label{lem:Lyapunov_upper_bound}
    Let Assumptions~\ref{assm:compact_cooperation_sets}--\ref{assm:tightened_coupling_constraints}, and \ref{assm:penalty_function} hold.
    Then, there exists $\alpha_{\mathrm{ub}} \in \mathcal{K}_\infty$ such that $V_T(\xi_T) \le \alpha_{\mathrm{ub}}(\vert \xi_T \vert_{\Xi_T^W})$ for all $\xi_T$ whose first component $x$ satisfies $x\in\mathcal{X}_N$.
\end{lem}
\begin{proof}
    Consider $T$ periodic cooperation output trajectories $y_{\mathcal{T}, \tau+1} \in \mathcal{Y}_T$ for $\tau\in\mathbb{I}_{0:T-1}$.
    Furthermore, let $x_{\mathcal{T}} \in \prod_{i=1}^{m} X_i^{T}$ such that $x_{\mathcal{T}}(0) \in \mathcal{X}_N$.
    Then, by Theorem~\ref{thm:recursive_feasibility}, there exist $\hat{u}_T\in\prod_{i=1}^{m}U_i^{T}$ and $\hat{y}_T\in\mathcal{Y}_T^{T}$ such that $x_{\mathcal{T}}(\tau) \in \mathcal{X}_N$ for all $\tau\in\mathbb{I}_{0:T-1}$.
    Choose ${y}_{\mathcal{T},\tau+1}=y_T^{\mathrm{pr}}(\cdot \vert \tau)$ such that $x_\mathcal{T}(\tau+1) = x_{\mu(\xi(\tau))}(1, x_\mathcal{T}(\tau))$.
    Then, $\xi(\tau) = (x_{\mathcal{T}}(\tau), y_{{\mathcal{T}},\tau+1})$.
    Next, restrict $x_{\mathcal{T}}$ further so that with 
    $(x_{T}^W, y_{T}^W) = \argmin_{(\hat{x}_T^W, \hat{y}_{T}^W) \in \Xi_T^W} \Vert (x_{\mathcal{T}} - \hat{x}_T^W, y_{\mathcal{T}} - \hat{y}_{T}^W) \Vert$,
    we have
    $\vert x_{{\mathcal{T}},i}(\tau) \vert_{x_{T,i}^W(\tau)} \le  \frac{1}{m}{\alpha_{\mathrm{ub}}^{\ell_i}}^{-1}(\varepsilon)$ with $\varepsilon$ from Lemma~\ref{lem:tracking_cost_upper_bound}.
    From~\eqref{eq:def_stable_set} there exists $u_T^W$ such that $h(x_T^W(\tau), u_T^W(\tau)) = y_{T,k+1}^W(\tau-k)$ for all $\tau,k\in\mathbb{I}_{0,T-1}$.
    By Assumption~\ref{assm:stage_cost_lower_and_upper_bound}, this implies $\sum_{i=1}^{m} \ell_i'(x_{{\mathcal{T}},i}(\tau), r_{T,i}^W(\tau)) \le \varepsilon$ for all $\tau \in \mathbb{I}_{0:T-1}$.
    Then, from Lemma~\ref{lem:tracking_cost_upper_bound}, for every $x_{\mathcal{T}}(\tau)$, there exists $\bar{u}(\cdot\vert\tau)$, such that $(\bar{u}(\cdot\vert\tau), y_{T,\tau+1}^W)$ is a feasible candidate solution in~\eqref{eq:central_OP} with $x = x_{\mathcal{T}}(\tau)$ and $y_T^{\mathrm{pr}} = y_{\mathcal{T},\tau+1} = y_T^{\mathrm{pr}}(\cdot\vert\tau)$.
    We abbreviate $J_i(x_{i}, u_i, y_{T,i}^{\mathrm{pr}}, y_{T}) = J_i(x_{i}, u_i, y_{T,i}, y_{T,i}^{\mathrm{pr}}, y_{T,\mathcal{N}_i})$.
    Since $y^{W}_{T, \tau+1} \in \mathcal{Y}_T^W$, we have
    \begin{align}\label{eq:Lyapunov_upper_bound_preliminary}
        &V(\xi(\tau)) \le \sum_{i=1}^{m} J_i(x_{{\mathcal{T}},i}(\tau), \bar{u}_i(\cdot\vert\tau), y_{T,i}^{\mathrm{pr}}(\cdot \vert \tau), y_{T,\tau+1}^{W}) - W^{\mathrm{c}}_N \notag
        \\
        &= \smash[b]{\sum_{i=1}^{m}} \Big(J_i^{\mathrm{tr}}(x_{{\mathcal{T}},i}(\tau), \bar{u}_i(\cdot\vert\tau), r_{T,i}^{W}(\cdot+\tau)) \notag\\
        &\phantom{={}} \hspace{2.3em}+ \lambda(N)V_i^{\Delta}(y_{T,\tau+1,i}^{W}, y_{\mathcal{T},\tau+1,i})\Big) \notag
        \\
        &\stackrel{\mathclap{\eqref{eq:penalty_cooperation_change_distance}}}{\le}\, \smash[t]{\sum_{i=1}^{m}} J_i^{\mathrm{tr}}(x_{{\mathcal{T}},i}(\tau), \bar{u}_i(\cdot\vert\tau), r_{T,i}^{W}(\cdot+\tau)) \notag\\
        &\hspace{1em} + \lambda(N)\alpha_{\mathrm{ub}}^{\Delta}(\vert y_{\mathcal{T},\tau+1} \vert_{y_{T,\tau+1}^{W}}) \notag
        \\
        &\stackrel{\eqref{eq:tracking_cost_upper_bound},\eqref{eq:stage_cost_lower_and_upper_bound}}{\le} \hspace{-2pt}\smash[b]{\sum_{i=1}^{m}} c_i^{\mathrm{f}}\alpha_{\mathrm{ub}}^{\ell_i}(\vert x_{{\mathcal{T}},i}(\tau) \vert_{x_{T,i}^{W}(\tau)}) \hspace{-1.5pt}+\hspace{-1.5pt} \lambda(N)\alpha_{\mathrm{ub}}^{\Delta}(\vert y_{\mathcal{T},\tau+1} \vert_{y_{T,\tau+1}^{W}}).
    \end{align}
    Summing up yields $V_T(\xi_T) \le \tilde{\alpha}_{\mathrm{ub}}(\vert \xi_T \vert_{\Xi_T^W})$
    with $\tilde{\alpha}_{\mathrm{ub}}(\vert \xi_T \vert_{\Xi_T^W}) = \sum_{\tau=0}^{T-1} \big(\sum_{i=1}^{m} c_i^{\mathrm{f}}\alpha_{\mathrm{ub}}^{\ell_i}(\vert x_{{\mathcal{T}},i}(\tau) \vert_{x_{T,i}^{W}(\tau)}) + \lambda(N)\alpha_{\mathrm{ub}}^{\Delta}(\vert y_{T}^{\mathrm{pr}}(\cdot\vert\tau) \vert_{y_{T,\tau+1}^{W}})\big)$.
    The existence of a local upper bound with $\tilde{\alpha}_{\mathrm{ub}}\in\mathcal{K}_{\infty}$ on a compact subset of $\mathcal{X}_N^{T} \times \mathcal{Y}_T^{T}$, together with compactness of $\mathcal{X}_N$ and $\mathcal{Y}_T$, establishes the claim (cf.~\cite[Prop. 2.16]{Rawlings2020}).
\end{proof}

Second, we show a lower bound on $V_T(\xi_T)$.
\begin{lem}\label{lem:Lyapunov_lower_bound}
    Let Assumptions~\ref{assm:compact_cooperation_sets}--\ref{assm:penalty_function} hold with the same $\omega > 1$.
    Then, there exists $\alpha_{\mathrm{lb}} \in \mathcal{K}_\infty$ so that $V_T(\xi_T) \ge \alpha_{\mathrm{lb}}(\vert \xi_T \vert_{\Xi_T^W})$ for all $\xi_T$ whose first component $x$ satisfies $x\in\mathcal{X}_N$.
\end{lem}
\begin{proof}
    Consider $T$ periodic cooperation output trajectories $y_{{\mathcal{T}}, \tau+1} \in \mathcal{Y}_T$ for $\tau\in\mathbb{I}_{0:T-1}$, and let $x_{\mathcal{T}} \in \prod_{i=1}^{m} X_i^{T}$ such that $x_{\mathcal{T}}(0) \in \mathcal{X}_N$.
    As in Lemma~\ref{lem:Lyapunov_upper_bound}'s proof, $x_{\mathcal{T}}(\tau) \in \mathcal{X}_N$ for all $\tau\in\mathbb{I}_{0:T-1}$. 
    Choose ${y}_{\mathcal{T},\tau+1}=y_T^{\mathrm{pr}}(\cdot \vert \tau)$ such that $x_\mathcal{T}(\tau+1) = x_{\mu(\xi(\tau))}(1, x_\mathcal{T}(\tau))$.
    Then, $\xi(\tau) = (x_{\mathcal{T}}(\tau), y_{{\mathcal{T}},\tau+1})$.
    If ${y}_{\mathcal{T},1}$ played no role in~\eqref{eq:central_OP}, it can be chosen arbitrarily, and we use ${y}_{\mathcal{T},1} = y_T^{0}(\cdot+1 \vert \change{\xi(T-1)} )$.
    Omitting non-negative terms, $\lambda(N) \ge 1$, and $W^{\mathrm{c}}(\hat{y}_T) \ge W^{\mathrm{c}}_0$ for all $\hat{y}_T\in \mathcal{Y}_T$, yield
    \begin{align}
        &V(\xi(\tau)) \ge \smash{\sum_{i=1}^{m}} \Big(\ell_i(x_{{\mathcal{T}},i}(\tau), \mu_i(\xi(\tau)), r_{T,i}^0(0\vert \xi(\tau))) \notag\\
        &\hspace{6.5em} + \lambda(N)V_i^{\Delta}(y_{T,i}^0(\cdot\vert\xi(\tau)), y_{\mathcal{T},\tau+1,i})\Big)
        \notag
        \\
        &\stackrel{\mathclap{\eqref{eq:stage_cost_upper_bounds_cooperation_distance}}}{\ge} \frac{1}{2}\eta_{\ell}(\vert y_T^0(\cdot\vert\xi(\tau)) \vert_{\mathcal{Y}_T^W}) \notag\\
        &\hspace{1em} + \frac{1}{2} \smash{\sum_{i=1}^{m}} \Big( \ell_i(x_{{\mathcal{T}},i}(\tau), \mu_i(\xi(\tau)), r_{T,i}^0(0\vert\xi(\tau))) \notag \\
        &\hspace{5em} + V_i^{\Delta}(y_{T,i}^0(\cdot\vert\xi(\tau)), y_{\mathcal{T},\tau+1,i})\Big)\notag
        \\
        &\stackrel{\mathclap{\hspace{0.75em}\eqref{eq:stage_cost_lower_and_upper_bound},\eqref{eq:penalty_cooperation_change_distance}}}{\ge}
        \hspace{1.2em}\frac{1}{2}\Big(\eta_{\ell}(\vert y_T^0(\cdot\vert \xi(\tau)) \vert_{\mathcal{Y}_T^W}) + \alpha_{\mathrm{lb}}^{\ell}(\vert x_{\mathcal{T}}(\tau) \vert_{x_{T}^0(0\vert\xi(\tau))}) \notag\\
        &\phantom{= \frac{1}{2}\Big(}+ \alpha_{\mathrm{lb}}^{\Delta}(\vert y_{\mathcal{T},\tau+1} \vert_{y_T^0(\cdot\vert \xi(\tau))}) \Big) 
        \label{eq:Lyapunov_lower_bound_with_stage_cost_and_penalty_on_change}
    \end{align}
    with $\alpha_{\mathrm{lb}}^{\ell}(s) = \sum_{i=1}^{m}\alpha_{\mathrm{lb}}^{\ell_i}(s)$ and  $\alpha_{\mathrm{lb}}^{\ell}\in \mathcal{K}_{\infty}$.
    Hence,
    \begin{align}\label{eq:Lyapunov_function_intermediate_lower_bound}
        &V_T(\xi_T) 
        \stackrel{\eqref{eq:Lyapunov_lower_bound_with_stage_cost_and_penalty_on_change}}{\ge} 
        \sum_{\tau=0}^{T-1} \frac{1}{2}\Big(\eta_{\ell}(\vert y_T^0(\cdot\vert \xi(\tau)) \vert_{\mathcal{Y}_T^W}) 
        \notag\\
        & \hspace{1em} + \alpha_{\mathrm{lb}}^{\ell}(\vert x_{\mathcal{T}}(\tau) \vert_{x_{T}^0(0\vert\xi(\tau))}) + \alpha_{\mathrm{lb}}^{\Delta}(\vert y_{\mathcal{T},\tau+1} \vert_{y_T^0(\cdot\vert \xi(\tau))}) \Big).
    \end{align}
    We now show that the right-hand side of~\eqref{eq:Lyapunov_function_intermediate_lower_bound} is positive definite with respect to $\Xi_T^W$.
    By Lemma~\ref{lem:Lyapunov_upper_bound}, it must be zero for $\xi_T \in \Xi_T^W$.
    Assume that $\xi_T \notin \Xi_T^W$.
    We have several possibilities.
    First, there may exist $\tau \in \mathbb{I}_{1:T-1}$ such that $\vert y_T^0(\cdot +1 \vert \xi(\tau-1))\vert_{y_T^0(\cdot\vert\xi(\tau))} > 0$ or $\vert y_T^0(\cdot+1\vert\xi(T-1))\vert_{y_T^0(\cdot \vert \xi(0))} > 0$ for $\tau=0$.
    Since $y_{\mathcal{T}, \tau+1} = y_T^{\mathrm{pr}}(\cdot\vert\tau) = y_T^0(\cdot+1 \vert \change{\xi(\tau-1)})$, and by our choice of $y_{\mathcal{T}, 1}$ in the case that $y_T^{\mathrm{pr}}(\cdot\vert 0)$ played no role in~\eqref{eq:central_OP}, the right-hand side of~\eqref{eq:Lyapunov_function_intermediate_lower_bound} is positive.
    Second, if the first case does not hold, then perhaps $h(x_{\mathcal{T}}(k), u_{\mathcal{T}}(k)) \neq y_{T,\tau+2}(k-1) = y_T^{\mathrm{pr}}(k-1 \vert \tau+1) = y_T^0(k \vert \xi(\tau))$ for some $k$. 
    Note that $\tau$ can be fixed here since the first case does not hold, and thus $y_T^0(k \vert \xi(\tau)) = y_T^0(0 \vert \xi(\tilde{\tau}))$ for some $\tilde{\tau} \in \mathbb{I}_{0:T-1}$.
    Due to Assumption~\ref{assm:unique_corresponding_equilibrium}, this entails $\vert x_{\mathcal{T}}(\tilde{\tau}) \vert_{x_T^0(0 \vert \xi(\tilde{\tau}))} > 0$ for some $\tilde{\tau} \in \mathbb{I}_{0:T-1}$
    and the right-hand side of~\eqref{eq:Lyapunov_function_intermediate_lower_bound} is positive.
    If the other two cases do not hold, then the last possibility is that there exists $\tau\in\mathbb{I}_{0:T-1}$ such that $y_T^0(\cdot \vert \xi(\tau)) \notin \mathcal{Y}_T^W$, which implies that the right-hand side of~\eqref{eq:Lyapunov_function_intermediate_lower_bound} is positive.
    Hence, the right-hand side is a continuous function that is positive definite with respect to $\Xi_T^W$.
    Furthermore, if $\vert \xi_T \vert_{\Xi_T^W} \to \infty$, by the above arguments the right-hand side of \eqref{eq:Lyapunov_function_intermediate_lower_bound} also tends to infinity, and $\mathcal{X}_N^T \times \mathcal{Y}_T^T$ is compact.
    Hence, there exists $\alpha_{\mathrm{lb}} \in \mathcal{K}_{\infty}$ (cf.~\cite{Kellett2014}) such that $V_T(\xi_T) \ge \alpha_{\mathrm{lb}}(\vert \xi_T \vert_{\Xi_T^W})$.
\end{proof}

Finally, we prove stability of $\Xi_T^W$ for the closed-loop system~\eqref{eq:extended_closed_loop}, which implies that the cooperative task is achieved as well as possible in closed loop.
\begin{thm}\label{thm:stability}
    Let Assumptions~\ref{assm:compact_cooperation_sets}--\ref{assm:penalty_function} hold with the same $\omega > 1$.
    Then, for any initial condition $x_0$ for which~\eqref{eq:central_OP} is feasible, the set $\Xi_T^W$ is asymptotically stable for the extended closed-loop system~\eqref{eq:extended_closed_loop}.
    Consequently, the closed-loop system~\eqref{eq:central_global_closed_loop} converges to a unique state trajectory that generates an output fulfilling the cooperative task as well as possible, i.e. $\lim_{t\to\infty} [\col_{i=1}^m(y_i(t)), \dots, \col_{i=1}^m(y_i(t+T-1))] \in \mathcal{Y}_T^{W}$.
\end{thm}
\begin{proof}
    Consider $x(t), x(t+1), y_T^{\mathrm{pr}}(\cdot\vert t), y_T^{\mathrm{pr}}(\cdot\vert t+1)$ from two consecutive time steps with $t\in\mathbb{N}_0$, and recall that by definition $y_T^{\mathrm{pr}}(\cdot\vert t+1) = y_T^0(\cdot +1 \vert \xi(t))$ for $t\in\mathbb{N}$.
    In addition, we set $y_T^{\mathrm{pr}}(\cdot \vert 0)=y_T^{0}(\cdot+1 \vert \change{\xi(T-1)})$ since it can be chosen arbitrarily, to align with the proof of Lemma~\ref{lem:Lyapunov_lower_bound}.
    From the solution of~\eqref{eq:central_OP} at time $t$, we build a (standard) candidate input using Assumption~\ref{assm:terminal_ingredients} as
    $\hat{u}_i(\cdot \vert t + 1) = \big(u_i^0(1 \vert \xi(t)), \dots, u_i^0(N-1 \vert \xi(t)), k_i^{\mathrm{f}}(x_{i,u_i^0(\cdot\vert \xi(t))}(N, x_i(t)), r_{T,i}^0(N\vert \xi(t)))\big)$ and consider $\hat{y}_{T,i}(\cdot\vert t+1) = y_{T,i}^0(\cdot+1\vert \xi(t))$.
    Theorem~\ref{thm:recursive_feasibility} showed that this is a feasible candidate solution of~\eqref{eq:central_OP}.
    Abbreviate again
    $J_i(x_i, u_i, y_{T,i}^{\mathrm{pr}}, y_{T}) = J_i(x_{i}, u_i, y_{T,i}, y_{T,i}^{\mathrm{pr}}, y_{T,\mathcal{N}_i})$.
    Now, we bound the difference between two consecutive time steps by inserting the candidate, cancelling some terms, and using the terminal ingredients (Assumption~\ref{assm:terminal_ingredients}) and the shift invariance of the cooperation objective function (Definition~\ref{def:COF}):
    \begin{align}
        &V(\xi(t+1)) - V(\xi(t)) \notag
        \\
        &\le \smash[b]{\sum_{i=1}^{m}} \Big(J_i\big(x_{i}(t+1), \hat{u}_i(\cdot\vert t+1), y_{T,i}^{\mathrm{pr}}(\cdot \vert t+1), \hat{y}_{T}(\cdot\vert t+1)\big) \notag\\
        &\hspace{3.2em} - J_i\big(x_{i}(t), u_i^0(\cdot\vert \xi(t)), y_{T,i}^{\mathrm{pr}}(\cdot \vert t), y_{T}^0(\cdot\vert \xi(t))\big)\Big)\notag
        \\
        &= \smash[t]{\sum_{i=1}^{m}} \bigg( \smash[t]{\sum_{k=1}^{N-1}} \ell_i(x_{i, u_i^0(\cdot\vert \xi(t))}(k, x_i(t)), u_i^0(k\vert\xi(t)), r_{T,i}^0(k\vert\xi(t))) \notag\\
        &\hspace{0.2em} + \ell_i(x_{i, u_i^0(\cdot\vert \xi(t))}(N, x_i(t)),  \hat{u}_i(N{-}1\vert t+1), r_{T,i}^0(N \vert \xi(t))) \notag\\
        &\hspace{0.2em} + V_i^{\mathrm{f}}(x_{i, \hat{u}_i(\cdot\vert t+1)}(N, x_i(t+1)), r_{T,i}^0(0 \vert \xi(t))) \notag\\
        &\hspace{0.2em} + \lambda(N)V_i^{\Delta}(y_{T,i}^0(\cdot+1\vert \xi(t)), y_{T,i}^0(\cdot +1 \vert \xi(t))) \notag\\
        &\hspace{0.2em} - \sum_{k=0}^{N-1} \ell_i(x_{i, u_i^0(\cdot\vert \xi(t))}(k, x_i(t)), u_i^0(k\vert \xi(t)), r_{T,i}^0(k \vert \xi(t)))  \notag\\
        &\hspace{0.2em} - V_i^{\mathrm{f}}(x_{i, u_i^0(\cdot \vert \xi(t))}(N, x_i(t)), r_{T,i}^0(N \vert \xi(t))) \notag\\
        &\hspace{0.2em} - \lambda(N)V_i^{\Delta}(y_{T,i}^0(\cdot\vert \xi(t)), y_{T,i}^{\mathrm{pr}}(\cdot\vert t)) \bigg)\notag\\
        &\hspace{0.2em} +\lambda(N)\Big(W^{\mathrm{c}}(y_{T}^0(\cdot+1\vert \xi(t))) - W^{\mathrm{c}}(y_{T}^0(\cdot\vert \xi(t)))\Big)\notag
        \\
        &\stackrel{\substack{\eqref{eq:terminal_cost_decrease},\,\text{Def.~\ref{def:COF}},\\\eqref{eq:penalty_cooperation_change_distance}}}{\le} \smash[b]{\sum_{i=1}^{m}} \Big( -\ell_i(x_i(t), \mu_i(\xi(t)), r_{T,i}^0(0 \vert \xi(t))) \notag\\
        &\hspace{6.1em} - \lambda(N) V_i^{\Delta}(y_{T,i}^0(\cdot\vert \xi(t)), y_{T,i}^{\mathrm{pr}}(\cdot\vert t)) \Big). \label{eq:Lyapunov_decrease_with_costs}
    \end{align} 

    As when showing the lower bound in Lemma~\ref{lem:Lyapunov_lower_bound}, using~\eqref{eq:Lyapunov_decrease_with_costs},~\eqref{eq:Lyapunov_lower_bound_with_stage_cost_and_penalty_on_change}, and~\eqref{eq:Lyapunov_function_intermediate_lower_bound} by summing up, we arrive at 
    \begin{align}\label{eq:Lyapunov_decrease}
        &V_T(\xi_T(t+1)) -  V_T(\xi_T(t)) \notag
        \\
        &
        = {\sum_{\tau=t}^{t+T-1}} V(\xi(\tau+1)) - V(\xi(\tau))
        \le -\alpha_{\mathrm{lb}}(\vert \xi_T(t) \vert_{\Xi_T^W}),
    \end{align}
    with $\alpha_{\mathrm{lb}}\in\mathcal{K}_{\infty}$ from Lemma~\ref{lem:Lyapunov_lower_bound}.
    Together with the upper bound of Lemma~\ref{lem:Lyapunov_upper_bound} and the lower bound of Lemma~\ref{lem:Lyapunov_lower_bound}, asymptotic stability of $\Xi_T^W$ for the closed-loop system~\eqref{eq:extended_closed_loop} follows from standard Lyapunov arguments, cf.~\cite[Thm. B.13]{Rawlings2020}.
    In particular, this implies convergence of $y_T^0(\cdot\vert\xi(\tau))$ to a unique solution achieving the cooperation goal as well as possible, i.e. $\lim_{t\to\infty}\vert y_T^0(\cdot\vert \xi(t)) \vert_{\mathcal{Y}_T^W} = 0$ as well as $\lim_{t\to\infty} \vert y_T^0(\cdot+1\vert \xi(t+1)) \vert_{y_T^0(\cdot\vert \xi(t))} = 0$, and the closed-loop state~\eqref{eq:central_global_closed_loop_a} follows a (periodic) trajectory that realizes this output.
\end{proof}

If the involved bounds are quadratic functions and Assumption~\ref{assm:better_cooperation_candidate} holds with a simple distance function, exponential stability can be guaranteed, as we show next.
\begin{thm}\label{thm:exponential_stability}
    Let Assumptions~\ref{assm:compact_cooperation_sets}--\ref{assm:penalty_function} hold with $\omega = 2$.
    Moreover, assume $\alpha_{\mathrm{lb}}^{\mathrm{c}}$ and $\alpha_{\mathrm{ub}}^{\mathrm{c}}$ of Definition~\ref{def:COF}, $\alpha_{\mathrm{lb}}^{\ell_i}$ and $\alpha_{\mathrm{ub}}^{\ell_i}$ of Assumption~\ref{assm:stage_cost_lower_and_upper_bound}, and $\alpha_{\mathrm{lb}}^{\Delta}$ and $\alpha_{\mathrm{ub}}^{\Delta}$ of Assumption~\ref{assm:penalty_function} are quadratic functions, e.g. $\alpha_{\mathrm{ub}}^{\mathrm{c}}(s) = a_{\mathrm{ub}}^{\mathrm{c}}s^2$ with $a_{\mathrm{ub}}^{\mathrm{c}} > 0$. 
    In addition, assume for $\psi$ of Assumption~\ref{assm:better_cooperation_candidate} that $\psi(y_T) = a_{\psi}\vert y_T \vert_{\mathcal{Y}_T^W}$ with $a_{\psi} > 0$.
    Then, for any initial condition $x_0$ for which~\eqref{eq:central_OP} is feasible, the set $\Xi_T^W$ is exponentially stable for the extended closed-loop system~\eqref{eq:extended_closed_loop}.
\end{thm}

The proof of Theorem~\ref{thm:exponential_stability} is given in the appendix.

This concludes our stability analysis of the closed-loop system.
Theorem~\ref{thm:recursive_feasibility} establishes that recursive feasibility generally holds as soon as a feasible cooperation output is found.
This may enable short prediction horizons and a large region of attraction since cooperation outputs close to the initial states may be chosen, similar to standard MPC for tracking~\cite{Limon2018}.
Furthermore, we proved asymptotic stability for bounds with $\mathcal{K}_{\infty}$-functions, and exponential stability for quadratic bounds. 
\begin{rem}
    We assumed throughout our analysis that the optimization problem~\eqref{eq:central_OP} is solved to global optimality.
    This is a very common assumption in the (distributed) MPC literature, even though it is unlikely to hold for a non-convex problem.
    If this is not assumed, and only a suboptimal solution of~\eqref{eq:central_OP} (e.g. a local minimum) is obtained, we are dealing with so-called suboptimal MPC.
    In the centralized case, additional care in the optimization problem still leads to stability based on a warm start with the standard MPC candidate solution (cf. Theorem \ref{thm:recursive_feasibility}), see, e.g.,~\cite[Sec.~2.7]{Rawlings2020}.
    We conjecture that we would need the following for our derived guarantees to hold.
    First, \change{either} the decentralized optimization algorithm solving~\eqref{eq:central_OP} should return feasible iterates, the returned suboptimal solution is projected onto the constraints, or suitable constraint tightening is used (cf.~\cite{JKoehler2019_inexact}).
    Second, the norm of the optimal input trajectory is proportional to the distance from the state to the optimal cooperation state reference trajectory, cf.~\cite[Sec.~2.7]{Rawlings2020}.
    Third, a candidate as in Theorem~\ref{thm:stage_cost_upper_bounds_cooperation_distance} is explicitly available as a warm start,
    i.e. it reduces the cooperation objective function more than it increases the tracking cost.
\end{rem}


\subsection{Design of sufficient ingredients for cooperation}\label{ssec:sufficient_design_central}
We complete this section by outlining specific choices for the cooperation objective function, the set of admissible cooperation outputs, and the penalty on the change in the cooperation output such that Assumptions~\ref{assm:better_cooperation_candidate} and~\ref{assm:penalty_function} hold.
As noted, the commonly used quadratic stage costs satisfy Assumptions~\ref{assm:stage_cost_lower_and_upper_bound} and \ref{assm:stage_cost_comparison} with $\omega=2$ given bounded constraints.

We start by stating an assumption on the cooperation objective function and the set of admissible cooperation outputs inspired by convex optimization (cf.~\cite{Bertsekas2016}).
\begin{assum}\label{assm:convexity_and_lipschitz_continuous_gradient}
    The sets $\mathcal{Y}_{T,i}$ and the cooperation objective function $W^{\mathrm{c}}$ are convex.
    Furthermore, $W^{\mathrm{c}}$ is continuously differentiable, and its gradient is Lipschitz continuous, i.e.
    there exists $L_W > 0$ such that
    \begin{equation}\label{eq:Lipschitz_continuous_cooperation_objective_function}
        \Vert \nabla W^{\mathrm{c}}(y_T) - \nabla W^{\mathrm{c}}(y_T') \Vert \le L_W \Vert y_T - y_T' \Vert
    \end{equation}
    holds for all $y_T, y_T' \in \mathcal{Y}_T$.
\end{assum}

Note that knowledge of $L_W$ is not required for implementing the proposed MPC scheme.
Since $\mathcal{Y}_T$ is compact,~\eqref{eq:Lipschitz_continuous_cooperation_objective_function} follows, e.g., if $W^{\mathrm{c}}$ is twice continuously differentiable.
Assumption~\ref{assm:convexity_and_lipschitz_continuous_gradient} implies Assumption~\ref{assm:better_cooperation_candidate}.
\begin{lem}\label{lem:sufficient_conditions_for_better_candidate_Lipschitz}
    Suppose Assumptions~\ref{assm:compact_cooperation_sets} and \ref{assm:convexity_and_lipschitz_continuous_gradient} hold.
    Define $p(y_T) = \mathcal{P}_{\mathcal{Y}_T}[y_T - s\nabla W^{\mathrm{c}}(y_T)]$ with $s > 0$.
    Then, $\hat{y}_T = y_T + \theta (p(y_T) - y_T)$ with $\theta\in[0,1]$ satisfies Assumption~\ref{assm:better_cooperation_candidate} with $\omega = 2$ and $c_{\psi} = 1$.
\end{lem}
\begin{proof}
    From the Projection Theorem~\cite[Prop. 1.1.4]{Bertsekas2016},
    \begin{equation*}
        (y - p(y_T))^\top(y_T - s\nabla W^{\mathrm{c}}(y_T) - p(y_T)) \le 0, \; \forall y \in \mathcal{Y}_T.
    \end{equation*}
    Inserting $y = y_T$ yields
    \begin{equation}\label{eq:aux_for_descent_lemma}
        \nabla W^{\mathrm{c}}(y_T)^\top(p(y_T) - y_T) \le \smash{-\frac{1}{s}} \Vert y_T - p(y_T)\Vert^2.
    \end{equation}
    Since~\eqref{eq:Lipschitz_continuous_cooperation_objective_function} holds, we can apply the Descent Lemma~\cite[Prop. A.24]{Bertsekas2016} and insert~\eqref{eq:aux_for_descent_lemma} to get
    \begin{align*}
            &W^{\mathrm{c}}(y_T + \theta(p(y_T) - y_T)) - W^{\mathrm{c}}(y_T) 
            \\
            &\le \nabla{W^{\mathrm{c}}}\change{(y_T)}^\top(\theta(p(y_T) - y_T)) + \frac{L_W}{2}\Vert \theta (p(y_T) - y_T) \Vert^2
            \\
            &\le -\frac{\theta}{s} \Vert y_T - p(y_T)\Vert^2 + \frac{L_W}{2}\Vert \theta (p(y_T) - y_T) \Vert^2.
    \end{align*}
    Choosing $s = \frac{2}{L_W\theta + 2}$, yields $W^{\mathrm{c}}(\hat{y}_T) - W^{\mathrm{c}}(y_T) \le - \theta\Vert p(y_T) - y_T \Vert^2$.
    Consider the mapping $\psi(y_T) = \Vert p(y_T) - y_T \Vert$.
    Since $W^{\mathrm{c}}$ is convex, $\psi(y_T) = 0$ if and only if $y_T$ minimizes $W^{\mathrm{c}}(y_T)$ over $\mathcal{Y}_T$ (cf.~\cite[Prop. 3.1.1, Prop. 1.1.4, Fig. 3.3.2]{Bertsekas2016}).
    Otherwise, $\psi(y_T) > 0$.
    Further, $p(y_T)$ is continuous, and hence, so is $\psi(y_T)$.
    Thus, $\psi(y_T)$ is a continuous and locally (i.e. on $\mathcal{Y}_T$) positive definite function with respect to $\mathcal{Y}_T^W$.
    The claim follows with
    $\vert \hat{y}_T \vert_{y_T}  = \Vert \theta (p(y_T) - y_T) \Vert = \theta  \psi(y_T)$ and $W^{\mathrm{c}}(\hat{y}_T) - W^{\mathrm{c}}(y_T) \le - \theta \psi(y_T)^2$.
\end{proof}

Next, we show that a weakened form of strong convexity allows us to establish Assumption~\ref{assm:better_cooperation_candidate} with a simple distance function, which we used to show exponential stability.
\begin{assum}\label{assm:weak_strong_convexity}
    The sets $\mathcal{Y}_{T,i}$ and the cooperation objective function $W^{\mathrm{c}}$ are convex.
    Furthermore, $W^{\mathrm{c}}$ is continuously differentiable, and there exists $\sigma > 0$ such that for all $y_T \in \mathcal{Y}_T$
    \begin{equation}\label{eq:weak_strong_convexity}
        W^{\mathrm{c}}(y_T) \ge W^{\mathrm{c}}_0 \hspace{-1pt} + \hspace{-1pt} \nabla W^{\mathrm{c}}(\bar{y}_T)^\top(y_T - \bar{y}_T) + \frac{\sigma}{2} \Vert \bar{y}_T - y_T \Vert^2 
    \end{equation}
    with $\bar{y}_T = \argmin_{\tilde{y}_T\in\mathcal{Y}_T^W} \vert y_T \vert_{\tilde{y}_T}$ (recall $W^{\mathrm{c}}_0 = W^{\mathrm{c}}(\bar{y}_T)$).
\end{assum}

Strong convexity, which implies a unique minimizer of $W^{\mathrm{c}}$, is sufficient but not necessary for Assumption~\ref{assm:weak_strong_convexity} to hold.

Assumption~\ref{assm:weak_strong_convexity} entails the following inequality:
\begin{align}\label{eq:weak_strong_convexity_implication}
    &\frac{\sigma}{2} \Vert \bar{y}_T \hspace{-1pt}-\hspace{-1pt} y_T \Vert^2 \stackrel{\eqref{eq:weak_strong_convexity}}{\le} W^{\mathrm{c}}(y_T) \hspace{-1pt}-\hspace{-1pt} W^{\mathrm{c}}(\bar{y}_T) \hspace{-1pt}-\hspace{-1pt} \nabla W^{\mathrm{c}}(\bar{y}_T)^\top(y_T \hspace{-1pt}-\hspace{-1pt} \bar{y}_T) \notag
    \\
    &\le - \nabla W^{\mathrm{c}}(y_T)^\top (\bar{y}_T - y_T) - \nabla W^{\mathrm{c}}(\bar{y}_T)^\top(y_T - \bar{y}_T) \notag
    \\
    &= (\nabla W^{\mathrm{c}}(y_T) - \nabla W^{\mathrm{c}}(\bar{y}_T))^\top (y_T - \bar{y}_T),
\end{align}
where the second inequality follows from convexity of $W^{\mathrm{c}}$.
This leads to the desired condition on $\psi$ of Theorem~\ref{thm:exponential_stability}.
\begin{lem}\label{lem:sufficient_conditions_for_better_candidate_weak_strong_convexity}
    Let Assumptions~\ref{assm:convexity_and_lipschitz_continuous_gradient} and~\ref{assm:weak_strong_convexity} hold.
    Define $p(y_T) = \mathcal{P}_{\mathcal{Y}_T}[y_T - s\nabla W^{\mathrm{c}}(y_T)]$ with $s \in (0, \frac{2}{L_W})$.
    Then, the candidate $\hat{y}_T = y_T + \theta (p(y_T) - y_T)$ with $\theta\in[0,1]$ satisfies Assumption~\ref{assm:better_cooperation_candidate} with $\psi(y_T) = a_{\psi} \vert y_T \vert_{\mathcal{Y}_T^W}$, $a_{\psi} > 0$ and $\omega=2$.
\end{lem}
\begin{proof}
    We first note that 
    \begin{align}\label{eq:weak_strong_convexity_lower_bound}
        &(\nabla W^{\mathrm{c}}(y_T) - \nabla W^{\mathrm{c}}(y_T'))^\top (y_T - y_T') \notag
        \\
        &\ge \frac{L_W\bar{\sigma}}{L_W+\bar{\sigma}}\Vert y_T - y_T'\Vert^2 + \frac{\Vert \nabla W^{\mathrm{c}}(y_T) - \nabla W^{\mathrm{c}}(y_T') \Vert^2}{L_W+\bar{\sigma}}
    \end{align}
    holds for all $y_T', y_T \in \mathcal{Y}_T$.
    The proof of equation~\eqref{eq:weak_strong_convexity_lower_bound} follows the proof of~\cite[Prop. B.5]{Bertsekas2016} and can be found in the appendix.

    Next, we show that the candidate satisfies~\eqref{eq:better_cooperation_candidate_a} with $\psi(y_T) = a_{\psi} \vert y_T \vert_{\mathcal{Y}_T^W}$ and $a_{\psi} > 0$.
    Define $\tilde{p}(y_T) = y_T - s\nabla W^{\mathrm{c}}(y_T)$.
    We then have 
    \begin{align*}
        &\Vert \tilde{p}(y_T) - \tilde{p}(y_T') \Vert^2 \notag
        \\
        &= \Vert y_T - y_T' \Vert^2 + s^2 \Vert \nabla W^{\mathrm{c}}(y_T) - \nabla W^{\mathrm{c}}(y_T')\Vert^2 \notag \\
        &\phantom{={}} - 2s(y_T - y_T')^\top(\nabla W^{\mathrm{c}}(y_T) - \nabla W^{\mathrm{c}}(y_T')) \notag
        \\
        &\stackrel{\eqref{eq:weak_strong_convexity_lower_bound}}{\le} \Vert y_T - y_T' \Vert^2 + s^2 \Vert \nabla W^{\mathrm{c}}(y_T) - \nabla W^{\mathrm{c}}(y_T')\Vert^2 \notag \\
        & \hspace{1em} -\hspace{-2pt} \frac{2s\bar{\sigma}L_W}{\bar{\sigma}\hspace{-2pt}+\hspace{-2pt}L_W}\Vert y_T \hspace{-2pt}- \hspace{-2pt}y_T' \Vert^2 \hspace{-2pt}-\hspace{-2pt} \frac{2s}{\bar{\sigma}\hspace{-2pt}+\hspace{-2pt}L_W}\Vert \nabla W^{\mathrm{c}}(y_T) \hspace{-2pt}-\hspace{-2pt} \nabla W^{\mathrm{c}}(y_T')\Vert^2 \notag 
        \\
        &= \big(1 - \frac{2s\bar{\sigma}L_W}{\bar{\sigma}+L_W}\big) \Vert y_T - y_T' \Vert^2 \notag\\
        &\phantom{={}} + \big( s^2 - \frac{2s}{\bar{\sigma}+L_W} \big) \Vert \nabla W^{\mathrm{c}}(y_T) - \nabla W^{\mathrm{c}}(y_T')\Vert^2.
    \end{align*}
    From here, in contrast to~\cite[Prop. B.5]{Bertsekas2016}, we rely on~\eqref{eq:weak_strong_convexity} instead of strong convexity.
    Choosing $y_T' = \bar{y}_T$ with $\vert y_T \vert_{\bar{y}_T} = \vert y_T \vert_{\mathcal{Y}_T^W}$ and applying~\eqref{eq:Lipschitz_continuous_cooperation_objective_function} if $s > \frac{2}{\bar{\sigma} + L_W}$ or~\eqref{eq:weak_strong_convexity_implication} if $s < \frac{2}{\bar{\sigma} + L_W}$ to bound the second term yields 
    \begin{equation}\label{eq:gradient_descent_is_contraction}
        \Vert \tilde{p}(y_T) \hspace{-1pt}-\hspace{-1pt} \tilde{p}(\bar{y}_T) \Vert^2 \hspace{-1pt}\le\hspace{-1pt} \max\hspace{-1pt}\big((1\hspace{-1.5pt}-\hspace{-1.5pt}s L_W)^2, (1\hspace{-1pt}-\hspace{-1pt}s\bar{\sigma})^2\big) \Vert y_T \hspace{-1pt}-\hspace{-1pt} \bar{y}_T \Vert^2.
    \end{equation}
    Then, for the projected gradient descent $p(y_T)$, where $p(y_T) = y_T$ if and only if $y_T\in \mathcal{Y}_T^W$ (cf.~\cite[Prop. 3.1.1, Prop. 1.1.4, Fig. 3.3.2]{Bertsekas2016}), we obtain 
    \begin{align}\label{eq:projected_gradient_descent_is_contraction}
        &\Vert p(y_T) - \bar{y}_T \Vert = \Vert p(y_T) - p(\bar{y}_T) \Vert \le \Vert \tilde{p}(y_T) - \tilde{p}(\bar{y}_T) \Vert \notag
        \\
        &\stackrel{\eqref{eq:gradient_descent_is_contraction}}{\le} \max(\vert 1-s L_W\vert, \vert 1-s \bar{\sigma}\vert) \Vert y_T - \bar{y}_T \Vert,
    \end{align}
    where the first inequality follows from the non-expansive property of the projection, see~\cite[Prop. 1.1.4]{Bertsekas2016}.
    Hence, 
    \begin{align*}
        &\vert \hat{y}_T \vert_{y_T} = \theta \Vert p(y_T) - y_T \Vert \le \theta (\Vert p(y_T) - \bar{y}_T \Vert + \Vert \bar{y}_T - y_T \Vert)
        \\
        &\stackrel{\eqref{eq:projected_gradient_descent_is_contraction}}{\le} \theta \big(1 + \max(\vert 1-s L_W\vert, \vert 1-s \bar{\sigma}\vert )\big)\vert y_T \vert_{\mathcal{Y}_T^W},
    \end{align*}
    which shows~\eqref{eq:better_cooperation_candidate_a} with $c_{\psi}a_{\psi} = 1 + \max(\vert 1-s L_W\vert, \vert 1-s \bar{\sigma}\vert)$.

    Finally, we show that the candidate also satisfies~\eqref{eq:better_cooperation_candidate_b} with $\psi(y_T) = a_{\psi} \vert y_T \vert_{\mathcal{Y}_T^W}$.
    Following the proof of Lemma~\ref{lem:sufficient_conditions_for_better_candidate_Lipschitz}, we get $W^{\mathrm{c}}(\hat{y}_T) - W^{\mathrm{c}}(y_T) \le - \theta\Vert p(y_T) - y_T \Vert^2$.
    Moreover, from~\eqref{eq:projected_gradient_descent_is_contraction}, $\vert y_T \vert_{\mathcal{Y}_T^W} \le \Vert y_T - p(y_T) \Vert + \Vert p(y_T) - \bar{y}_T \Vert \le \Vert y_T - p(y_T) \Vert + \max(\vert 1-s L_W\vert, \vert 1-s \bar{\sigma}\vert) \vert y_T \vert_{\mathcal{Y}_T^W}.$
    Thus, $\Vert y_T - p(y_T) \Vert \ge \big(1-\max(\vert 1-s L_W\vert, \vert 1-s \bar{\sigma}\vert)\big) \vert y_T \vert_{\mathcal{Y}_T^W}$.
    Hence,~\eqref{eq:better_cooperation_candidate_b} holds with $\psi(y_T) = a_{\psi} \vert y_T \vert_{\mathcal{Y}_T^W}$ and $a_{\psi} = 1-\max(\vert 1-s L_W\vert, \vert 1-s \bar{\sigma}\vert) > 0$ since $s < \frac{2}{L_W}$ and $\bar{\sigma} < L_W$.
\end{proof}

Finally, we prove in the following lemma that a simple quadratic penalty function on the change in the cooperation output suffices for Assumption~\ref{assm:penalty_function}.
\begin{lem}
    Define $V_i^{\Delta}(y_{T,i}, y_{T,i}^{\mathrm{pr}}) = \delta_i \smash[t]{\sum_{\tau=0}^{T-1}} \Vert y_{T,i}(\tau) - y_{T,i}^{\mathrm{pr}}(\tau) \Vert^2$
    with $\delta_i > 0$.
    Then, Assumption~\ref{assm:penalty_function} holds with $\omega=2$.
\end{lem}
\begin{proof}
    Condition~\eqref{eq:penalty_cooperation_change_distance} holds trivially; we proceed to show~\eqref{eq:penalty_cooperation_change_decrease}.
    First, $\Vert \hat{y}_{T,i}(\tau) - y_{T,i}^{\mathrm{pr}}(\tau) \Vert^2 
    \le 2\Vert \hat{y}_{T,i}(\tau) - y_{T,i}(\tau) \Vert^2 + 2 \Vert y_{T,i}(\tau) - y_{T,i}^{\mathrm{pr}}(\tau) \Vert^2$.
    With $\delta^{\Delta} = \max_i (\delta_i)$, this yields condition~\eqref{eq:penalty_cooperation_change_decrease}:
    $\sum_{i=1}^{m} V_i^{\Delta}(\hat{y}_{T,i}, y_{T,i}^{\mathrm{pr}}) - 2V_i^{\Delta}(y_{T,i}, y_{T,i}^{\mathrm{pr}}) \le 2\delta^{\Delta} \sum_{\tau=0}^{T-1} \Vert \hat{y}_{T}(\tau) - y_{T}(\tau) \Vert^2  \le 2\delta^{\Delta} \big( \sum_{\tau=0}^{T-1} \Vert \hat{y}_{T}(\tau) - y_{T}(\tau) \Vert \big)^2$.
\end{proof}

\section{Closed-loop performance bounds}\label{sec:performance_bounds}
In this section, we derive a closed-loop performance bound of the proposed distributed MPC scheme. 
Based on~\cite{MatthiasKohler2023_TransientPerformanceMPC}, we derive a transient performance bound and show optimal performance for an infinite prediction horizon under certain conditions on the cooperative task.
Hence, we are interested in bounds on 
\begin{equation*}
    \mathfrak{J}_K(x,u,r_T) = \smash[t]{\sum_{k=0}^{K-1} \sum_{i=1}^{m}} \ell_i(x_{i,u_i}(k, x_i), u_i(k), r_{T,i}(k)).
\end{equation*}
We establish a performance bound with respect to the closed-loop input trajectory and the cooperative reference that solves the cooperative task for an infinite horizon. 
\change{For simplicity, we consider only the case of a terminal region with nonempty interior, i.e. Assumption \ref{assm:terminal_ingredients} holds with $c_i^{\mathrm{b}}>0$.}

It is helpful to define the set of cooperation outputs that are part of a feasible candidate solution of~\eqref{eq:central_OP} given a specific initial condition $x$.
This set is independent of $y_T^{\mathrm{pr}}$ because $y_T^{\mathrm{pr}}$ enters only through the objective function.
Define $\mathbb{Y}_{N}(x) = \{ y_T \in \mathcal{Y}_T \mid \exists u\in\mathbb{U}^N(x): (u, y_T) \text{ is feasible in}~\eqref{eq:central_OP} \}$.

We introduce the 'standard' MPC problem for tracking a given periodic trajectory $r_T$ from an initial state $x$:
\begin{equation}\label{eq:standard_MPC_problem}
    V_{N}^{\mathrm{s}}(x, r_T) = \min_{u\in\mathbb{U}^N(x)} \sum_{i=1}^{m} J_i^{\mathrm{tr}}(x_i, u_i, r_{T,i})
\end{equation}%
subject to, for all $i\in\mathbb{I}_{1:m}$,~\eqref{eq:central_OP_terminal_constraint} and ~\eqref{eq:central_OP_coupling_constraints}.
The solution is denoted by $u_{\mathrm{s}}^0(\cdot\vert x, r_T)$ with $\mu_{\mathrm{s}}(x, r_T) = u_{\mathrm{s}}^0(0\vert x, r_T)$,
which coincides with the solution of~\eqref{eq:central_OP} if $r_T^0$ is inserted in~\eqref{eq:standard_MPC_problem}.
The set of feasible states is denoted by $\mathbb{X}_{N}^{\mathrm{s}}(r_T)$.

Asymptotic stability of the periodic trajectory $r_T$ for the closed loop $x(t+1) = f(x(t), \mu_{\mathrm{s}}(x(t), r_T))$
with $x(0) \in \mathbb{X}_{N}^{\mathrm{s}}(r_T)$ follows directly from~\cite[Thm. 5.13]{Gruene2017} (cf.~\cite[Rem. 5.17]{Gruene2017}) if
Assumptions \ref{assm:stage_cost_lower_and_upper_bound} and \ref{assm:terminal_ingredients} hold.
Hence, there exists $\beta_{\mathrm{s}} \in \mathcal{KL}$ such that $\vert x_{\mu_{\mathrm{s}}}(t,x) \vert_{x_T(t)} \le \beta_{\mathrm{s}}(\vert x \vert_{x_T(0)}, t)$ for all $x\in\mathbb{X}_N^{\mathrm{s}}(r_T)$.
This also entails that for all $\widetilde{N}\in\mathbb{N}_0$ there exists $\alpha_{\widetilde{N}}^{\mathrm{s}} \in \mathcal{K}_{\infty}$ such that
\begin{align}\label{eq:standard_value_function_upper_bound}
    V_{N}^{\mathrm{s}}(x, r_T) \le V_{\widetilde{N}}^{\mathrm{s}}(x, r_T) \le \alpha_{\widetilde{N}}^{\mathrm{s}}(\vert x \vert_{x_T(0)})
\end{align}
holds for all $x \in \mathbb{X}_{\widetilde{N}}^{\mathrm{s}}(r_T)$ and $N\ge\widetilde{N}$.
Moreover, 
\begin{align}\label{eq:terminal_cost_is_upper_bound}
    V_{N}^{\mathrm{s}}(x, r_T) \le \smash[t]{\sum_{i=1}^{m}{}} V_i^{\mathrm{f}}(x_i, r_{T,i}(0))
\end{align}
for all $N\in\mathbb{N}$ and $x\in\prod_{i=1}^{m}\mathcal{X}_i^{\mathrm{f}}(r_{T,i}(0))$, cf.~\cite[Thm. 5.13]{Gruene2017}.

In the following proposition, we show a performance bound for the standard MPC scheme similar to~\cite[Thm. 8.22]{Gruene2017} adapted to the case of tracking a periodic trajectory.
\begin{prop}\label{prop:standard_MPC_performance_bound}
Let Assumptions
\ref{assm:stage_cost_lower_and_upper_bound}, and
\ref{assm:terminal_ingredients} \change{with $c_i^{\mathrm{b}}>0$},
hold.
Then, for all $\widetilde{N}$, there exist $\delta_1, \delta_2 \in \mathcal{L}$ such that for all $r_T\in\prod_{i=1}^{m}\mathcal{Z}_{T,i}$ and $x\in\mathbb{X}_{\widetilde{N}}^{\mathrm{s}}(r_T)$ the inequality
\begin{align}\label{eq:standard_MPC_performance_bound}
    V_{N}^{\mathrm{s}}(x, r_T) \le \hspace{-2em} \inf_{\substack{u\in\mathbb{U}^K(x) \\ x_u(K,x) \in \mathcal{B}_{\kappa}(r_T(K))}} \hspace{-2em} \mathfrak{J}_K(x, u, r_T) + \delta_1(N) + \delta_2(K)
\end{align}
holds for $N\ge\widetilde{N}$ with $\kappa = \beta_{\mathrm{s}}(\vert x \vert_{x_T(0)}, K)$.
\end{prop}

The proof is an adaption of the proof of~\cite[Thm 8.22]{Gruene2017} and can be found in the appendix.
Based on the proof of~\cite[Thm. 8.22]{Gruene2017}, the right-hand side of~\eqref{eq:standard_MPC_performance_bound} also upper bounds the closed-loop cost by showing $\mathfrak{J}_K(x, \mu_{\mathrm{s}}, r_T) \le V_N^{\mathrm{s}}(x, r_T)$ (cf. \cite[Thm. 8.21]{Gruene2017}).
However,~\eqref{eq:standard_MPC_performance_bound} suffices for our purpose.
We refer to~\cite[Chap. 2]{Gruene2017} for a discussion of $\delta_1$ and $\delta_2$ in~\eqref{eq:standard_MPC_performance_bound}.

For a meaningful performance bound, we require existence of a uniformly reachable $y_T'\in \mathcal{Y}_T^W$, which we establish in the following lemma.
\begin{lem}\label{lem:uniform_reachability}
    Let Assumptions \ref{assm:compact_cooperation_sets}--\ref{assm:stage_cost_lower_and_upper_bound}, \ref{assm:terminal_ingredients} with $c_i^{\mathrm{b}}>0$, \ref{assm:tightened_coupling_constraints}, and \ref{assm:better_cooperation_candidate} hold.
    Then, for all $\widetilde{N} \in \mathbb{N}_0$ with $\mathcal{X}_{\widetilde{N}}\neq \emptyset$, there exists $\widehat{N}\in\mathbb{N}_0$ such that for any $x\in\mathcal{X}_{\widetilde{N}}$, $y_T^\mathrm{pr}\in \mathcal{Y}_T$ there exist $\hat{u}\in\mathbb{U}^N(x)$ and $\hat{y}_T \in \mathcal{Y}_T^W$ so that $(\hat{u}, \hat{y}_T)$ is a feasible candidate solution of~\eqref{eq:central_OP} for $N\ge \widehat{N}$.
\end{lem}

The proof, which can be found in the appendix, is inspired by the proof of~\cite[Lem. 3]{MatthiasKohler2023_TransientPerformanceMPC}, but adapted to bounds with comparison functions, and where $\mathcal{Y}_T^W$ is not a singleton.

Based on Lemma~\ref{lem:uniform_reachability}, we are able to show that a certain invariance property holds for the closed-loop states, and the cooperation outputs converge uniformly to the closed-loop solution of the cooperative task for growing prediction horizons.
\begin{lem}\label{lem:invariance_and_uniform_convergence}
    Let Assumptions
    \ref{assm:compact_cooperation_sets}--\ref{assm:stage_cost_lower_and_upper_bound}, \ref{assm:terminal_ingredients} with $c_i^{\mathrm{b}}>0$,~\ref{assm:tightened_coupling_constraints}--\ref{assm:penalty_function}, and~\ref{assm:stage_cost_comparison_with_1} hold with the same $\omega>1$, and let $M \in \mathbb{N}_0$.
    Then, the following two properties hold.
    \begin{enumerate}
        \item There exist $P \ge M$ and $\widehat{N}$ such that $x_{\mu}(k,x)\in \mathcal{X}_P$ for all $x\in \mathcal{X}_{M}$, $N\ge \widehat{N}$ and $k\in\mathbb{N}_0$.
        \item Let $\xi(0) = (x, y_T^{\mathrm{pr}})$ with $x\in\mathcal{X}_M$ and $y_T^{\mathrm{pr}}$ arbitrary.
        Then, $\lim_{N\to\infty} y_{T}^0(\cdot\vert\xi(k)) = y_T'(\cdot+k)$ uniformly on $\mathcal{X}_M$ for all $k\in\mathbb{N}_0$ where $y_T'$ is the eventual closed-loop solution of the cooperative task for $N\to\infty$.
        Thus, $\lim_{N\to\infty} \vert y_{T}^0(\cdot\vert \xi(k)) \vert_{\mathcal{Y}_T^W} = 0$ uniformly on $\mathcal{X}_M$.
    \end{enumerate}
\end{lem}
\begin{proof}
    Let $M\in\mathbb{N}_0$.
    We start by showing uniform convergence of $y_T^0(\cdot\vert \xi(0))$ to $\bar{y}_T(\cdot \vert \xi(0))$ on $\mathcal{X}_M$ where $\bar{y}_T(\cdot \vert \xi(k)) = \argmin_{y_T\in\mathcal{Y}_T^W} \vert y_T^0(\cdot\vert \xi(k)) \vert_{y_T} $.
    Suppose there exist $c_i > 0$, $i\in\mathbb{I}_{1:m}$ such that for all $N \ge \change{M}$ there exists $j\in\mathbb{I}_{1:m}$ and $x \in \mathcal{X}_{M}$ with $\vert y_{T,j}^{\change{*}}(\cdot\vert \xi(0))\vert_{\mathcal{Y}_T^W} = \vert y_{T,j}^{\change{*}}(\cdot\vert \xi(0))\vert_{\bar{y}_{T,j}(\cdot\vert \xi(0))} \ge c_j$, where $\xi(0) = (x, y_T^{\mathrm{pr}})$ and $y_T^{\mathrm{pr}}$ is arbitrary.
    Note that by Definition~\ref{def:COF} and~\eqref{eq:def_best_reachable_cost}, there exist $\tilde{\alpha}_{\mathrm{lb}}^{\mathrm{c}}, \tilde{\alpha}_{\mathrm{ub}}^{\mathrm{c}} \in \mathcal{K}_\infty$ such that $\tilde{\alpha}_{\mathrm{lb}}^{\mathrm{c}}(\vert y_T \vert_{\mathcal{Y}_T^W}) \le W^{\mathrm{c}}(y_T) - W_0^{\mathrm{c}} \le \tilde{\alpha}_{\mathrm{ub}}^{\mathrm{c}}(\vert y_T \vert_{\mathcal{Y}_T^W})$.
    From Lemma~\ref{lem:uniform_reachability}, there exist $\widehat{N} \in \mathbb{N}_0$, $\hat{y}_T \in \mathcal{Y}_T^W$, and $\hat{u}\in\mathbb{U}^{N}(x)$ such that $(\hat{u}, \hat{y}_T)$ is a feasible candidate in~\eqref{eq:central_OP} for all $N\ge \widehat{N}$ and $x \in \mathcal{X}_{M}$. 
    Moreover, since $\lambda(N)W^{\mathrm{c}}(\hat{y}_T) - W^{\mathrm{c}}_N = 0$ and  $V_i^{\Delta}$ is omitted in~\eqref{eq:central_OP} for $\xi(0)$, we get $V(\xi(0)) \le V_N^{\mathrm{s}}(x, \change{\hat{r}_T})$ for all $N\ge \widehat{N}$.
    Then, we have $\lambda(N)(W^{\mathrm{c}}(y_{T}^0(\cdot\vert \xi(0)))-W_0^{\mathrm{c}}) \le V(\xi(0)) \le V_N^{\mathrm{s}}(x, \hat{r}_T) \stackrel{\eqref{eq:standard_value_function_upper_bound}}{\le} \alpha_{M}^{\mathrm{s}}(\vert x \vert_{\hat{x}_T(0)}) \le \alpha_{M}^{\mathrm{s}}(\gamma^{\mathrm{r}})$ for all $N\ge \widehat{N}$ with $\gamma^{\mathrm{r}} = \sup_{(x_i,u_i) \in \change{Z_i}, \, r_{T,i}\in\mathcal{Z}_{T,i}} (\vert x \vert_{r_{T}(0)})$ because $Z_i$ are compact, as well as using the candidate solution in the second inequality.
    But then $\lambda(N) \le \frac{\alpha_{M}^{\mathrm{s}}(\gamma^{\mathrm{r}})}{W^{\mathrm{c}}(y_{T}^0(\cdot\vert \xi(0))) - W_0^{\mathrm{c}}} \le \frac{\alpha_{M}^{\mathrm{s}}(\gamma^{\mathrm{r}})}{\tilde{\alpha}_{\mathrm{lb}}^{\mathrm{c}}(c_j)}$
    which yields a contradiction for sufficiently large $N$.

    Next, to prove the invariance property, we first show a turnpike property.
    We prove that for all $\Gamma > 0$, there exists $\sigma_{\Gamma} \in \mathcal{L}$ such that for all $N,P\in \mathbb{N}$, $x \in X$, $u\in\mathbb{U}^N(x)$ and $y_T \in \mathcal{Y}_T$ with $\sum_{i=1}^{m} J_i(x_i, u_i, y_{T,i}, y_{T,i}^{\mathrm{pr}}, y_{T,\mathcal{N}_i}) \le \Gamma$, the set $\bar{Q} = \{ k\in \mathbb{I}_{0:N-1} \mid \vert x_u(k,x) \vert_{\bar{x}_T(k)} \ge \sigma_{\Gamma}(P) \}$ has at most $P$ elements, where $\bar{y}_T \in \mathcal{Y}_T^W$ such that $\vert y_T \vert_{\bar{y}_T} = \vert y_T \vert_{\mathcal{Y}_T^W}$.
    With $\tilde{\alpha}_{\mathrm{lb}}^{\mathrm{c}} \in \mathcal{K}_{\infty}$ from before, with~\eqref{eq:stage_cost_lower_and_upper_bound} and Assumption~\ref{assm:unique_corresponding_equilibrium}, we have $\sum_{i=1}^{m} \ell_i(x_i, u_i, r_{T,i}(\tau)) + W^{\mathrm{c}}(y_T) \ge \sum_{i=1}^{m} \alpha_{\mathrm{lb}}^{\ell_i}(\vert x_i \vert_{x_{T,i}(\tau)}) + \tilde{\alpha}_{\mathrm{lb}}^{\mathrm{c}}(\vert y_T \vert_{\bar{y}_T}) \ge \bar{\rho}(\vert x \vert_{\bar{x}_T(\tau)})$ for some $\bar{\rho} \in \mathcal{K}_{\infty}$ and $\tau\in\mathbb{I}_{0:T-1}$.
    We now prove the turnpike property by contradiction. 
    Fix $\Gamma > 0$ and choose $\sigma_{\Gamma}(P) = \bar{\rho}^{-1}(\frac{\Gamma}{P})$.
    Suppose there exist $N,P,x,u,r_T$ such that $\sum_{i=1}^{m} J_i(x_i, u_i, y_{T,i},y_{T,i}^{\mathrm{pr}},y_{T,\mathcal{N}_i}) \le \Gamma$ but $\bar{Q}$ has at least $P+1$ elements.
    However, then $\sum_{i=1}^{m} J_i(x_i, u_i, y_{T,i},y_{T,i}^{\mathrm{pr}},y_{T,\mathcal{N}_i}) \ge \sum_{i=1}^{m}\sum_{k=0}^{\change{N-1}} \ell_i(x_{i,u_i}(k,x_i), u_i(k), r_{T,i}(k)) + NW^{\mathrm{c}}(y_T) \ge \sum_{k=0}^{N} \bar{\rho}(\vert x_u(k,x) \vert_{\bar{x}_T(k)}) \ge \sum_{k\in\bar{Q}} \bar{\rho}(\sigma_{\Gamma}(P)) \ge \frac{(P+1)\Gamma}{P} > \Gamma$, which is a contradiction.

    Now, we use this turnpike property to show that there exist $P\ge M$ and $\widehat{N}$ such that for all $x\in \mathcal{X}_{M}$, $N\ge \widehat{N}$ and $k\in\mathbb{N}_0$, we have $x_{\mu}(k,x)\in \mathcal{X}_P$.
    Let $x\in \mathcal{X}_{M}$.
    As previously in the proof, we have $V(\xi(0)) \le V_N^{\mathrm{s}}(x, \hat{r}_T) \le \alpha_{M}^{\mathrm{s}}(\gamma^{\mathrm{r}})$ for all $N\ge \widehat{N}$.
    We invoke the turnpike property with $\Gamma = \alpha_{M}^{\mathrm{s}}(\gamma^{\mathrm{r}})$ and choose $P$ such that $\sigma_{\mathrm{\Gamma}}(P) \le \min_{i} c_i^{\mathrm{b}}$ with $c_i^{\mathrm{b}}$ from Assumption~\ref{assm:terminal_ingredients}.
    Thus, the set $\{k\in\mathbb{I}_{0:N-1} \mid \vert x_{u^0(\cdot\vert \xi(0))}(k) \vert_{\bar{x}_T(k)} \ge \min_i c_i^{\mathrm{b}} \}$ has at most $P$ elements for all $x\in\mathcal{X}_{M}$ and $N \ge \widehat{N}$.
    Hence, $x_{\mu}(1, x) \in \mathcal{X}_P$, since at most $N-P$ elements of $x_{i,u^0(\cdot\vert \xi(0))}(k)$ are outside the terminal region of $\bar{r}_{T,i}$ by~\eqref{eq:terminal_non_empty_interior}.
    In addition, from~\eqref{eq:Lyapunov_decrease_with_costs}, $V(\xi(k)) \le V(\xi(0)) \le \alpha_{M}^{\mathrm{s}}(\gamma^{\mathrm{r}})$ for all $k\in\mathbb{N}_0$ with $\xi(k) = (x_{\mu}(k,x), y_T^{\mathrm{pr}}(\cdot \vert k))$.
    Therefore, we can apply the turnpike property for all $k\in\mathbb{N}_0$ with the same $\Gamma$ and $P$.
    This yields $x_{\mu}(k,x) \in \mathcal{X}_P$ for all $k\in \mathbb{N}_0$.

    To finish the second claim, assume that $y_T^0(\cdot\vert \xi(k-1))$ uniformly converges to $\bar{y}_T(\cdot \vert \xi(k-1))$ on $\mathcal{X}_M$ for all $k\in\mathbb{I}_{1:\tilde{k}}$, e.g. $\tilde{k}=1$ as shown before, and where $\xi(0) = (x, y_T^{\mathrm{pr}})$ with $y_T^{\mathrm{pr}}$ arbitrary and $x\in\mathcal{X}_M$.
    By Assumption~\ref{assm:terminal_ingredients} and the invariance property, there exists $u'\in\mathbb{U}^N(x(k))$ such that $(u', y_T^0(\cdot+1\vert \xi(k-1)))$ is feasible in~\eqref{eq:central_OP} for $\xi(k)$ and $N\ge P$, where $u'$ solves~\eqref{eq:standard_MPC_problem} for $r_T = r_T^0(\cdot+1\vert \xi(k-1))$.
    Then,
    \begin{align*}
        &\lambda(N)W^{\mathrm{c}}(y_T^0(\cdot\vert \xi(k))) \\
        &+ \lambda(N)\sum_{i=1}^{m} V_i^{\Delta}(y_{T,i}^0(\cdot\vert\xi(k)), y_{T,i}^0(\cdot+1\vert\xi(k-1)))
        \\
        &\le \lambda(N) W^{\mathrm{c}}(y_T^0(\cdot\vert \xi(k)))\hspace{-2pt} \\
        &\phantom{\le{}} + \smash[b]{\sum_{i=1}^{m}} J_i^{\mathrm{tr}}(x_i(k), u_i^0(\cdot\vert\xi(k)), r_{T,i}^0(\cdot\vert\xi(k))) \\
        &\phantom{\le{}} + \lambda(N)\sum_{i=1}^{m} V_i^{\Delta}(y_{T,i}^0(\cdot\vert\xi(k)), y_{T,i}^0(\cdot+1\vert\xi(k-1)))
        \\
        &\le V_N^{\mathrm{s}}(x(k), r_T^0(\cdot\hspace{-2pt}+\hspace{-2pt}1\vert \xi(k\hspace{-2pt}-\hspace{-2pt}1))) \hspace{-1pt}+\hspace{-1pt} \lambda(N)W^{\mathrm{c}}(y_T^0(\cdot\hspace{-2pt}+\hspace{-2pt}1\vert \xi(k\hspace{-2pt}-\hspace{-2pt}1)))
        \\
        &\le \alpha_{P}^{\mathrm{s}}(\gamma^{\mathrm{r}}) + \lambda(N)W^{\mathrm{c}}(y_T^0(\cdot+1\vert \xi(k-1))).
    \end{align*}
    The last inequality follows as before from~\eqref{eq:standard_value_function_upper_bound} and compactness of $Z_i$, and for the second we used the candidate.
    Thus,
    \begin{align}\label{eq:aux_uniform_convergence_closed_loop_cooperation_output}
        \alpha_{P}^{\mathrm{s}}(\gamma^{\mathrm{r}}) &\ge \lambda(N)\sum_{i=1}^{m} V_i^{\Delta}(y_{T,i}^0(\cdot\vert\xi(k)), y_{T,i}^0(\cdot+1\vert\xi(k-1)))
        \notag\\
        &\hspace{0.7em} + \lambda(N)\big(W_0^{\mathrm{c}} - W^{\mathrm{c}}(y_T^0(\cdot+1\vert \xi(k-1)))\big) \notag\\
        &\hspace{0.7em} + \lambda(N)\big(W^{\mathrm{c}}(y_T^0(\cdot\vert \xi(k))) - W_0^{\mathrm{c}}\big)
        \notag\\
        &\ge \lambda(N)\sum_{i=1}^{m} V_i^{\Delta}(y_{T,i}^0(\cdot\vert\xi(k)), y_{T,i}^0(\cdot+1\vert\xi(k-1)))
        \notag\\
        &\hspace{0.7em} + \lambda(N)\big( \tilde{\alpha}_{\mathrm{lb}}^{\mathrm{c}}(\vert y_T^0(\cdot \vert \xi(k)) \vert_{\bar{y}_T(\cdot \vert \xi(k))}) \big) \notag\\
        &\hspace{0.7em} - \lambda(N)\big( \tilde{\alpha}_{\mathrm{ub}}^{\mathrm{c}}(\vert y_T^0(\cdot \vert \xi(k-1)) \vert_{\bar{y}_T(\cdot \vert \xi(k-1))})\big).
    \end{align}
    Since $y_T^0(\cdot\vert \xi(k-1))$ converges uniformly, for all $c_0>0$ there exists $N_0\in\mathbb{N}$ such that 
    $\tilde{\alpha}_{\mathrm{ub}}^{\mathrm{c}}(\vert y_T^0(\cdot \vert \xi(k-1)) \vert_{\bar{y}_T(\cdot \vert \xi(k-1))}) < \tilde{\alpha}_{\mathrm{ub}}^{\mathrm{c}}(c_0)$.
    First, suppose there exists $c_j > 0$ with $j\in\mathbb{I}_{1:m}$, such that for all $N\in\mathbb{N}_0$, and all $\hat{y}_{T,j}\in \mathcal{Y}_T^W$, $\vert y_{T,j}^0(\cdot\vert \xi(k))\vert_{\hat{y}_{T,j}} \ge c_j$.
    Thus, $\tilde{\alpha}_{\mathrm{lb}}^{\mathrm{c}}(\vert y_T^0(\cdot \vert \xi(k)) \vert_{\bar{y}_T(\cdot \vert \xi(k))}) \ge \tilde{\alpha}_{\mathrm{lb}}^{\mathrm{c}}(c_j)$. 
    Choose $c_0$ such that $\tilde{\alpha}_{\mathrm{ub}}^{\mathrm{c}}(c_0) \le \frac{\tilde{\alpha}_{\mathrm{lb}}^{\mathrm{c}}(c_j)}{2}$.
    Then, since $V_i^{\Delta}$ is non-negative from Assumption~\ref{assm:penalty_function}, we have $\alpha_{P}^{\mathrm{s}}(\gamma^{\mathrm{r}}) \ge \lambda(N) \frac{\tilde{\alpha}_{\mathrm{lb}}^{\mathrm{c}}(c_j)}{2}$ for all $N\ge N_0$ from~\eqref{eq:aux_uniform_convergence_closed_loop_cooperation_output}, which is a contradiction.
    Hence, $\lim_{N\to\infty} \vert y_{T}^0(\cdot\vert \xi(k)) \vert_{\bar{y}_T(\cdot\vert\xi(k))} = 0$ uniformly.
    Second, assume there exists $c_{\Delta}$ such that for all $N\in\mathbb{N}_0$ the inequality $\vert y_T^0(\cdot\vert \xi(k)) \vert_{y_T^0(\cdot+1\vert \xi(k-1))} \ge c_{\Delta}$ holds for some $k\in\mathbb{N}$. 
    Choose now $c_0$ such that $\tilde{\alpha}_{\mathrm{ub}}^{\mathrm{c}}(c_0) \le \frac{\alpha_{\mathrm{lb}}^{\Delta}(c_{\Delta})}{2}$ with $\alpha_{\mathrm{lb}}^{\Delta}$ from Assumption~\ref{assm:penalty_function}. 
    Then, from~\eqref{eq:aux_uniform_convergence_closed_loop_cooperation_output} and Assumption~\ref{assm:penalty_function}, $\alpha_{P}^{\mathrm{s}}(\gamma^{\mathrm{r}}) \ge \lambda(N)\frac{\alpha_{\mathrm{lb}}^{\Delta}(c_{\Delta})}{2}$, which is also a contradiction.
    Thus, $\lim_{N\to\infty} \sum_{i=1}^{m} V_i^{\Delta}(y_{T,i}^0(\cdot\vert\xi(k)), y_{T,i}^0(\cdot+1\vert\xi(k-1))) = 0$ uniformly.
    Finally, by induction, $\lim_{N\to\infty}y_{T}^0(\cdot\vert\xi(k)) = y_T'(\cdot+k)$ uniformly where $y_T'(\cdot+k) = \lim_{N\to\infty} \bar{y}_T(\cdot+k \vert \xi(0))$ is the closed-loop solution to the cooperative task for $N\to\infty$.
\end{proof}

We established uniform convergence of the cooperation outputs to the closed-loop cooperation output that solves the cooperative task on a fixed set of initial states.
This is the first result \change{in this paper} that relies on the scaling in the objective function of~\eqref{eq:central_OP}.
\change{It} shows that Assumption~\ref{assm:scaling} is an important ingredient for a well-behaved asymptotic performance of the closed-loop system.
Furthermore, we showed a turnpike property that implies that the closed-loop system starting from a set $\mathcal{X}_{\widetilde{N}}$ also stays in a set $\mathcal{X}_P$ for some $P\ge\widetilde{N}$. 

The following assumption is similar to Assumption~\ref{assm:stage_cost_comparison}, except we require a coefficient-free comparison between the stage costs of two references, cf.~\cite[Assm.~4]{MatthiasKohler2023_TransientPerformanceMPC}.
\begin{assum}\label{assm:stage_cost_comparison_with_1}
    There exist $c_3^{\ell_i}, c_4^{\ell_i} > 0$ such that for all $y_T, \hat{y}_T \in \mathcal{Y}_T$, $(x_i, u_i) \in Z_i$ and $\tau\in\mathbb{I}_{0:T-1}$:
    \begin{align}\label{eq:stage_cost_comparison_with_1}
        &\ell_i(x_i, u_i, \hat{r}_{T,i}(\tau)) \le \ell_i(x_i, u_i, r_{T,i}(\tau)) + c_3^{\ell_i} \vert \hat{r}_{T,i}(\tau) \vert_{r_{T,i}(\tau)}^2 \notag\\
        &\phantom{+ c_3^{\ell_i} \vert \hat{r}_{T,i}(\tau) \vert_{r_{T,i}(\tau)}^2} + c_4^{\ell_i} \vert \hat{r}_{T,i}(\tau) \vert_{r_{T,i}(\tau)}.
    \end{align}
\end{assum}

As stated in~\cite[Rem.\ 1]{MatthiasKohler2023_TransientPerformanceMPC}, Assumption~\ref{assm:stage_cost_comparison_with_1} holds for quadratic stage costs on bounded constraint sets $Z_i$ and $\mathcal{Y}_{T,i}$.
\change{As pointed out previously, boundedness of the coupling constraints $\mathcal{C}_i$ is not required for $\mathcal{Y}_{T,i}$ to be bounded.}

Finally, we derive a \change{transient} closed-loop performance bound with respect to the infinite-horizon closed-loop solution of the cooperative task.
\begin{thm}\label{thm:transient_performance_bound}
    Let Assumptions
    \ref{assm:compact_cooperation_sets}--\ref{assm:stage_cost_lower_and_upper_bound},
    \ref{assm:terminal_ingredients} with $c_i^{\mathrm{b}}>0$,
    \ref{assm:tightened_coupling_constraints}--\ref{assm:penalty_function},
    and \ref{assm:stage_cost_comparison_with_1}
    hold, all with $\omega = 2$.
    Then, for any $\widetilde{N}\in\mathbb{N}_0$, there exist $\delta_1, \delta_2 \in \mathcal{L}$, $N' \in \mathbb{N}$, such that for all $x\in\mathcal{X}_{\widetilde{N}}$, $N\ge N'$ and $K\in\mathbb{N}$,
    \begin{align}\label{eq:transient_performance_bound}
        &\mathfrak{J}_K(x, \mu, r_T') \le \hspace{-2em} \inf_{\substack{u\in\mathbb{U}^K(x) \\ x_u(K,x) \in \mathcal{B}_{\kappa}(r_T'(K))}} \hspace{-2em} \mathfrak{J}_K(x, u, r_T') + \delta_1(N) + \delta_2(K) \notag
        \\
        &\hspace{-1em} + \hspace{-2pt}{\sum_{k=0}^{K-1}} {\sum_{i=1}^{m}} c_3^{\ell_i} \vert r_{T,i}^0(0\vert\xi(k))\vert_{r_{T,i}'(k)}^2 \hspace{-2pt}+\hspace{-1pt} c_4^{\ell_i} \vert r_{T,i}^0(0\vert\xi(k))\vert_{r_{T,i}'(k)}
    \end{align}
    holds with $y_T' = \lim_{k\to\infty} \lim_{N\to\infty} y_T^{\mathrm{pr}}(\cdot\vert t_k) \in \mathcal{Y}_T^W$, $t_k = kT$, $k\in\mathbb{N}_0$, and $\kappa = \beta_{\mathrm{s}}(\vert x \vert_{x_T'(0)}, K)$.
\end{thm}
\begin{proof}
    Let $\widetilde{N}\in\mathbb{N}_0$ and $x\in\mathcal{X}_{\widetilde{N}}$.
    We only need to show~\eqref{eq:transient_performance_bound} for sufficiently large $N$ and $K$ as in the proof of Proposition~\ref{prop:standard_MPC_performance_bound}.
    From Lemma~\ref{lem:invariance_and_uniform_convergence}, $\lim_{N\to\infty} \vert y_T^0(\cdot\vert\xi(0)) \vert_{y_T'} = 0$ uniformly.
    Hence, there exists $N_{\mathrm{b}}'$ such that $\vert y_T^0(\cdot\vert\xi(0)) \vert_{y_T'} \le \frac{\min_i c_i^{\mathrm{b}}}{2L_x}$ with $L_{x} = \max_i L_{x,i}$ holds for all $N\ge N_{\mathrm{b}}'$.
    From Lemma~\ref{lem:uniform_reachability}'s proof, we know that there exists $N'\ge N_{\mathrm{b}}'$, $\tau'\in\mathbb{I}_{0:T-1}$, and $u^{\mathrm{b}} \in \mathbb{U}^{N'}(x)$ such that $\vert x_{i,u_i^{\mathrm{b}}}(N',x) \vert_{x_{T,i}^0(\tau'\vert\xi(0))} \le \frac{c_i^{\mathrm{b}}}{2}$.
    Thus, with Assumption~\ref{assm:unique_corresponding_equilibrium}, $\vert x_{i,u_i^{\mathrm{b}}}(N',x) \vert_{x_{T,i}'(\tau'\vert\xi(0))} \le \vert x_{i,u_i^{\mathrm{b}}}(N',x) \vert_{x_{T,i}^0(\tau'\vert\xi(0))} + \vert x_{T,i}^0(\tau'\vert\xi(0)) \vert_{x_{T,i}'(\tau'\vert\xi(0))} \le \frac{c_i^{\mathrm{b}}}{2} + L_x \vert y_T^0(\cdot\vert\xi(0)) \vert_{y_T'} \le c_i^{\mathrm{b}}$.
    Therefore, there exists $N' \in \mathbb{N}$ and $u'\in\mathbb{U}^N(x)$ such that $(u', y_T')$ is a feasible candidate in~\eqref{eq:central_OP} for all $N\ge N'$ and $\xi(0)$ with $x\in\mathbb{X}_{\widetilde{N}}$. 
    Consequently,
    \begin{align*}
        &\mathfrak{J}_K(x, \mu, r_T') = \sum_{k=0}^{K-1}\sum_{i=1}^{m} \ell_i(x_{i,\mu_i}(k, x_i), \mu_i(k), r_{T,i}'(k))
        \\
        &\stackrel{\eqref{eq:stage_cost_comparison_with_1}}{\le} \smash[t]{\sum_{k=0}^{K-1}\sum_{i=1}^{m}} \Big( \ell_i(x_{i,\mu_i}(k, x_i), \mu_i(k), r_{T,i}^0(0\vert\xi(k)))
        \\
        &\hspace{2em} + c_3^{\ell_i} \vert r_{T,i}^0(0\vert\xi(k)) \vert_{r_{T,i}'(k)}^2 + c_4^{\ell_i} \vert r_{T,i}^0(0\vert\xi(k)) \vert_{r_{T,i}'(k)}\Big)
        \\
        &\smash[t]{\stackrel{\substack{\eqref{eq:penalty_cooperation_change_distance}\\\eqref{eq:Lyapunov_decrease_with_costs}}}{\le}} V(\xi(0)) - V(\xi(K)) 
        \\
        &\hspace{0.5em} + \sum_{k=0}^{K-1}\sum_{i=1}^{m} c_3^{\ell_i} \vert r_{T,i}^0(0\vert\xi(k)) \vert_{r_{T,i}'(k)}^2 + c_4^{\ell_i} \vert r_{T,i}^0(0\vert\xi(k)) \vert_{r_{T,i}'(k)}.
    \end{align*}
    Since at $\xi(0)$,~\eqref{eq:central_OP} is solved without taking $y_T^{\mathrm{pr}}$ and $V_i^{\Delta}$ into account, and there exists a feasible candidate solution $(u', y_T')$ as outlined above, the solution of~\eqref{eq:standard_MPC_problem} is also a feasible candidate solution, i.e. $(u_{\mathrm{s}}^0(\cdot\vert x, r_T'), y_T')$, and $V(\xi(0)) \le V_N^{\mathrm{s}}(x, y_T')$.
    Note that $V(\xi(k)) \ge 0$.
    Hence, $\mathfrak{J}_K(x, \mu, r_T') \le  V_N^{\mathrm{s}}(x, y_T') + \sum_{k=0}^{K-1}\sum_{i=1}^{m}c_3^{\ell_i} \vert r_{T,i}^0(0\vert\xi(k)) \vert_{r_{T,i}'(k)}^2 + \sum_{k=0}^{K-1}\sum_{i=1}^{m} c_4^{\ell_i} \vert r_{T,i}^0(0\vert\xi(k)) \vert_{r_{T,i}'(k)}$,
    and Proposition~\ref{prop:standard_MPC_performance_bound} entails the claimed inequality.
\end{proof}

Compared to the performance bound~\eqref{eq:standard_MPC_performance_bound} for a standard MPC scheme, the derived transient performance bound~\eqref{eq:transient_performance_bound} contains error terms that depend on the optimal cooperation outputs in each time step, the shape of the stage cost, and the shape of the set of admissible cooperation outputs.
This is a similar structure to the one in~\cite[Thm. 2]{MatthiasKohler2023_TransientPerformanceMPC} as expected, since we adapted the analysis of~\cite{MatthiasKohler2023_TransientPerformanceMPC}.

The transient performance bound~\eqref{eq:transient_performance_bound} enables the derivation of an asymptotic performance bound, i.e. for $K\to\infty$ and $N\to\infty$.
Here, the error terms in~\eqref{eq:transient_performance_bound} depending on $K$ and $N$ vanish.
We exploit exponential convergence and the uniform convergence of the cooperation outputs to show this.
This results in the following asymptotic performance bound.
\change{%
Define the set of inputs that have the generated state sequence converge to a periodic reference $r_T$ as
\begin{align*}
    \mathbb{U}_{\{r_T\}}^{\infty}(x) = \{&u \in \mathbb{U}^{\infty}(x) \mid \lim_{k\to\infty} \vert x_u(t_k + \tau, x) \vert_{x_T(t_k + \tau)} = 0,
    \\
    &t_{k+1} = t_k + T,\; t_0 = 0, \; \forall \tau \in \mathbb{I}_{0:T-1}\}.
\end{align*}
}
\begin{thm}\label{thm:asymptotic_performance_bound}
    Let Assumptions~\ref{assm:compact_cooperation_sets}--\ref{assm:stage_cost_lower_and_upper_bound},~\ref{assm:terminal_ingredients} with $c_i^{\mathrm{b}}>0$,~\ref{assm:tightened_coupling_constraints}--\ref{assm:penalty_function}, and~\ref{assm:stage_cost_comparison_with_1}
    hold with $\omega=2$.
    Moreover, assume the conditions in Theorem~\ref{thm:exponential_stability} are satisfied.
    Then, for any $x\in\mathcal{X}_{\widetilde{N}}$ with $\widetilde{N}\in\mathbb{N}_0$,
    \begin{equation}\label{eq:asymptotic_performance_bound}
        \lim_{N\to\infty}\mathfrak{J}_\infty(x, \mu, r_T') = \inf_{\change{u \in \mathbb{U}_{\{r_T'\}}^{\infty}(x)}} \mathfrak{J}_{\infty}(x, u, r_T').
    \end{equation}
\begin{proof}
    First, we show that the series terms in~\eqref{eq:transient_performance_bound} converge for $K\to\infty$ due to exponential convergence of the closed-loop system.
    Let $N\in \mathbb{N}_0$ and $x\in\mathcal{X}_N$.
    From exponential stability (Theorem~\ref{thm:exponential_stability}), there exist $a_{\mathrm{e}}>0$ and $b_{\mathrm{e}}\in(0,1)$ such that the inequality $\vert \xi_T(k) \vert_{\Xi_T^W} \le a_{\mathrm{e}} \vert \xi_T(0) \vert_{\Xi_T^W} b_{\mathrm{e}}^{k}$ holds for all $k\in\mathbb{N}_0$.
    Thus, with $L = \max_i(L_{x,i}, L_{u,i})$ from Assumption~\ref{assm:unique_corresponding_equilibrium}, $\vert r_{T}^0(0\vert \xi(k)) \vert_{r_{T}'(k)} \le L\vert y_{T}^0(\cdot\vert \xi(k)) \vert_{y_{T}'} = L\vert y_{T}^{\mathrm{pr}}(\cdot-1 \vert k+1) \vert_{y_{T}'} \le L \sum_{\tau=k+1}^{T+k+1} \vert y_{T}^{\mathrm{pr}}(\cdot \vert \tau) \vert_{y_{T,\tau+1}^W} \le L \vert \xi_T(k+1) \vert_{\Xi_T^W} \le L a_{\mathrm{e}} \vert \xi_T(0) \vert_{\Xi_T^W} b_{\mathrm{e}}^{k+1}$,
    where we define $y_{T, k}^{\change{W}} = y_{T, k - (T+1)\floor*{\frac{k - 1}{T+1}}}^{\change{W}}$ for a well-defined sum after the second inequality.
    This involves a modulo operation with an offset of one to ensure consistent indexing.
    Then, $\sum_{k=0}^{\infty} \sum_{i=1}^{m} c_3^{\ell_i} \vert r_{T}^0(0\vert \xi(k)) \vert_{r_{T}'(k)}^2 \le \sum_{k=0}^{\infty} c_3^{\ell} L^2 a_{\mathrm{e}}^2 \vert \xi_T(0) \vert_{\Xi_T^W}^2 b_{\mathrm{e}}^{2(k+1)} \le \frac{c_3^{\ell} L^2 a_{\mathrm{e}}^2 \vert \xi_T(0) \vert_{\Xi_T^W}^2}{1-b_{\mathrm{e}}}$
    with $c_3^{\ell} = \max_i c_3^{\ell_i}$ and because $b_{\mathrm{e}}^{2(k+1)} \le b_{\mathrm{e}}^{k}$.
    Similarly, $\sum_{k=0}^{\infty} \sum_{i=1}^{m} c_4^{\ell_i} \vert r_{T}^0(0\vert \xi(k)) \vert_{r_{T}'(k)} \le \frac{c_4^{\ell} L a_{\mathrm{e}} \vert \xi_T(0) \vert_{\Xi_T^W}}{1-b_{\mathrm{e}}}$ with $c_4^{\ell} = \max_i c_4^{\ell_i}$.
    Hence, because $Z_i$ is bounded, both series converge.

    Now, we combine all parts to show that the error terms in the transient bound~\eqref{eq:transient_performance_bound} vanish for $K\to\infty$ and $N\to\infty$.
    Since there exists $P\ge\widetilde{N}$ such that $\lim_{N\to\infty} y_{T}^0(\cdot \vert \xi(k)) = y_T'(\cdot+k)$ uniformly on $\mathcal{X}_P$ due to Lemma~\ref{lem:invariance_and_uniform_convergence}, we have
    \begin{align*}
        &\lim_{N\to\infty} {\sum_{k=0}^{\infty}} c_3^{\ell_i} \hspace{-1pt}\vert r_{T,i}^0(0\vert\xi(k))\vert_{r_{T,i}'\hspace{-1pt}(k)}^2 \hspace{-2.5pt}+ \hspace{-2pt}c_4^{\ell_i}\hspace{-1pt} \vert r_{T,i}^0(0\vert\xi(k))\vert_{r_{T,i}'\hspace{-1pt}(k)}
        \\
        &={\sum_{k=0}^{\infty}} \hspace{-2pt}\lim_{N\to\infty} \hspace{-1pt} c_3^{\ell_i} \hspace{-1pt}\vert r_{T,i}^0(0\vert\xi(k))\vert_{r_{T,i}'\hspace{-1pt}(k)}^2 \hspace{-2.5pt}+ \hspace{-2pt}c_4^{\ell_i}\hspace{-1pt} \vert r_{T,i}^0(0\vert\xi(k))\vert_{r_{T,i}'\hspace{-1pt}(k)} \hspace{-2pt}=\hspace{-2pt} 0
    \end{align*}
    Hence, for any $x\in\mathcal{X}_{\widetilde{N}}$ with $\widetilde{N} \in \mathbb{N}_0$, the inequality
    $\lim_{N\to\infty}\mathfrak{J}_\infty(x, \mu, r_T') \le \inf_{\change{u\in\mathbb{U}_{\{r_T'\}}^{\infty}(x)}} \mathfrak{J}_{\infty}(x, u, r_T')$ follows from \eqref{eq:transient_performance_bound}.
    The other direction, i.e. $\lim_{N\to\infty}\mathfrak{J}_\infty(x, \mu, r_T') \ge \inf_{\change{u \in \mathbb{U}_{\{r_T'\}}^{\infty}(x)}} \mathfrak{J}_{\infty}(x, u, r_T')$
    follows from \change{stability of the closed-loop system (Theorem \ref{thm:exponential_stability}) and} optimality.
\end{proof}
\end{thm}

The asymptotic performance bound~\eqref{eq:asymptotic_performance_bound} proves that the proposed scheme is able to recover infinite optimal performance for $N\to\infty$, with additional advantages, e.g. a larger region of attraction, and the important flexibility of providing an emerging solution to the cooperative task by optimized cooperation instead of having to specify one \emph{a priori}.

\section{Numerical examples}
In this section, we provide two additional examples that complement the first example in Section~\ref{ssec:satellite_example}.
All three examples were implemented in Python using CasADi~\cite{Andersson2019}, IPOPT~\cite{Waechter2005}, CVXPY~\cite{Diamond2016_CVXPYPythonembeddedmodeling}, and Gurobi.
In all simulations, a (suboptimal) solution of~\eqref{eq:central_OP} was obtained using the decentralized sequential quadratic programming scheme presented in~\cite{Stomberg2024}.
The implementation is available at~\cite{public_code}.

\subsection{Crossing a narrow path}
First, we consider two agents that must keep a safe distance from each other, but also need to both cross a narrow pathway.
We illustrate how the cooperation objective function can be designed to avoid getting stuck in the pathway.
Consider two agents with simple double integrator dynamics with a two-dimensional position.
The agents start on opposite sides of a narrow pathway and want to reach the other's initial position.
In addition, they must not get closer than 0.8, and hence cannot pass each other in the narrow pathway.
We use the Pseudo-Huber loss function $L_{\delta}(a) = \delta^2 \big(\sqrt{1 + \frac{a^2}{\delta^2}} - 1\big)$ since it approximates $\Vert a \Vert$ for large values of $a$.
We choose $W_1^{\mathrm{c}}(y_{T,1}) = 2000 L_{0.01}((y_{T,1,1} - 20)^2 + y_{T,1,2}^2)$ and $W_2^{\mathrm{c}}(y_{T,2}) = 1000 L_{0.01}((y_{T,2,1} + 20)^2 + y_{T,2,2}^2)$.
In addition, $V_i^{\Delta}(y_{T,i}, y_{T,i}^{\mathrm{pr}}) = \frac{1}{10^4} \Vert y_{T,i} - y_{T,i}^{\mathrm{pr}} \Vert^2$.
The result for $N=20$ and $T=1$ can be seen in Figure~\ref{fig:bridge}.
Once the agents meet in the narrow pathway, they cannot pass each other, so the agent with the larger weight on the cooperation cost function (Agent 1) pushes the other out.
The cooperative task succeeds since our cooperation objective function satisfies Assumption~\ref{assm:better_cooperation_candidate}.

\change{%
Note that the agents would be stuck in the narrow pathway if a quadratic cooperation objective function is chosen. 
Namely, at some point the agents could not pass each other, as there always exists a local minimum where neither agent benefits from moving, since a reduction in one agent's cooperative cost would not offset the increase in the other's. Overcoming this would require carefully tuned weights that depend on the length of the narrow pathway. 
This is not the case for $L_{\delta}$, which can be tuned to be approximately linear outside a neighbourhood of $a=0$ corresponding to the terminal region.
Since the derivative of a linear cost is constant, the local minimum is eliminated, ensuring success independent of the length of the narrow pathway.
}

\begin{figure}[tb]
    \setlength\axisheight{0.45\linewidth}
    \setlength\axiswidth{0.95\linewidth}
    \centering
    \input{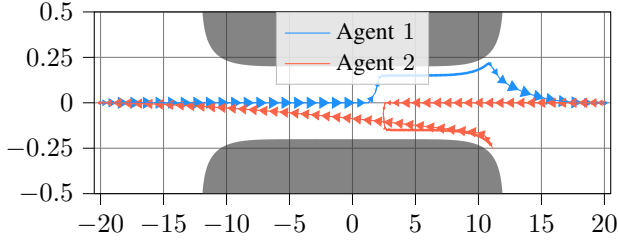}
    \caption{Two agents exchanging positions through a narrow pathway that they cannot cross simultaneously due to collision avoidance constraints. The shaded area indicates the boundary of the narrow pathway. The triangles indicate the direction of travel and velocity. The agents slow down while the collision avoidance constraint is active due to a suboptimal solution returned by the employed decentralized optimization algorithm.}
    \label{fig:bridge}
\end{figure}%

\subsection{Synchronization and flocking}
In this example, we consider a multi-agent system comprised of $m=4$ quadrotors with dynamics
\change{%
\begin{align*}
    \dot{x}_{i,1}(t) &= x_{i,6}(t),
    &
    \dot{x}_{i,6}(t) &= 9.81 \tan(x_{i,4}(t)),
    \\
    \dot{x}_{i,2}(t) &= x_{i,7}(t),
    &
    \dot{x}_{i,7}(t) &= 9.81 \tan(x_{i,5}(t)),
    \\
    \dot{x}_{i,3}(t) &= x_{i,8}(t),
    & 
    \dot{x}_{i,8}(t) &= 0.91 u_{i,3}(t) \hspace{-2pt}-\hspace{-2pt} 9.81,
    \\
    \dot{x}_{i,4}(t) &=  x_{i,9}(t) \hspace{-2pt}- \hspace{-2pt}8 x_{i,4}(t),
    & 
    \dot{x}_{i,9}(t) &= 10 (u_{i,1}(t) \hspace{-2pt}-\hspace{-2pt} x_{i,4}(t)),
    \\
    \dot{x}_{i,5}(t) &=  x_{i,10}(t) \hspace{-2pt}-\hspace{-2pt} 8 x_{i,5}(t),
    &
    \dot{x}_{i,10}(t) &= 10 (u_{i,2}(t) \hspace{-2pt}-\hspace{-2pt} x_{i,5}(t)),
\end{align*}
}
adapted from~\cite{Hu2018}.
We discretize the dynamics using the Euler method with a step-size of $h = \SI{0.1}{\second}$.
Terminal costs and constraints are computed offline following~\cite{JKoehler2020a}.
As the output, we choose the position of the quadrotors, i.e., $y_{i,1} = x_{i,1}$, $y_{i,2} = x_{i,2}$, and $y_{i,3} = x_{i,3}$.
Here, $x_{i,3}$ is the quadrotor's altitude.
We impose as constraints: $\change{\vert} x_{i,k} \change{\vert} \le 21$ for $k\in\mathbb{I}_{1:3}$, $\change{\vert} x_{i,k} \change{\vert} \le \frac{\pi}{4}$ for $k\in\mathbb{I}_{4:5}$, $\change{\vert} x_{i,k} \change{\vert} \le 2$ for $k\in\mathbb{I}_{6:8}$, $\change{\vert} x_{i,k} \change{\vert} \le 3$ for $k\in\mathbb{I}_{9:10}$, $\change{\vert} u_{i,k} \change{\vert} \le \frac{\pi}{9}$ for $k\in\mathbb{I}_{1:2}$, and $0 \le u_{i,3} \le 19.62$.
Constraints on the cooperation references are tightened to: 
$\change{\vert}x_{T,i,k} \change{\vert} \le 20.95$ for $k\in\mathbb{I}_{1:3}$, $\change{\vert} x_{T,i,k} \change{\vert} \le 0.75$ for $k\in\mathbb{I}_{4:5}$, $\change{\vert} x_{T,i,k} \change{\vert} \le 1.95$ for $k\in\mathbb{I}_{6:8}$, $\change{\vert} x_{T,i,k} \change{\vert} \le 2.9$ for $k\in\mathbb{I}_{9:10}$, $\change{\vert} u_{T,i,k} \change{\vert} \le 0.3$ for $k\in\mathbb{I}_{1:2}$, and $0.05 \le u_{T,i,3} \le 19.5$.

We aim to illustrate that the proposed scheme is flexible with respect to switches in the cooperative objective function, and it finds a solution to the cooperative task as well as possible despite conflicting objectives.
Until $t=349$, the cooperative task will be for the agents to converge to a trajectory that follows a circle.
The quadrotors should agree on the circle's radius, centre and altitude.
Beginning at $t=350$, the first agent should follow an externally provided reference signal, whereas the other agents should converge to the position of the first agent, i.e. follow it to achieve output consensus.
However, at all times, the quadrotors must maintain a minimum distance of $\SI{0.4}{\metre}$, which conflicts with the desire for consensus.
Hence, the coupling constraint is defined as $\mathcal{C}_i = \{(y_{T,i}, y_{T,\mathcal{N}_i}) \mid \Vert y_{T,i} - x_{T,j} \Vert^2 \ge 0.4 \quad\forall i\in\mathcal{N}_j\}$.

We choose $N = 10$, $T = 50$, $V_i^{\Delta}(y_{T,i}(\cdot\vert t), y_{T,i}^{\mathrm{pr}}(\cdot \vert t)) = \frac{1}{10^4T} \sum_{t=0}^{T-1}\Vert y_{T,i}(\tau \vert t) - y_{T,i}^{\mathrm{pr}}(\tau \vert t) \Vert^2$, and a quadratic stage cost where $Q$ and $R$ are diagonal matrices with \change{diagonals} $[0.5, 0.5, 0.5, 0.375, 0.375, 0.25, 0.25, 0.25, 0.125, 0.125]$ and $[0.25, 0.25, 0.005]$.

For $t \le 349$, we augment the cooperation output by the parametric decision variables $1 \le y_{T,i}^{\mathrm{r}} \le 2$ for the circle's radius and $y_{T,i}^{\mathrm{c}}$ for the circle's centre.
Let $\tilde{y}_{T,i} = (y_{T,i}, y_{T,i}^{\mathrm{r}}, y_{T,i}^{\mathrm{c}})$.
Then, we define 
\begin{align*}
    &W^{\mathrm{c}}_i(\tilde{y}_{T,i}(\cdot\vert t), \tilde{y}_{T,\mathcal{N}_i}(\cdot\vert t)) \\
    &= \sum_{\tau=0}^{T-1} \frac{1}{T}\bigg(\Big( y_{T,i,1}(\tau \vert t) - y_{T,i}^{\mathrm{r}}(t)\cos\left(\frac{2\pi \tau}{T} + (i-1)\frac{45\pi}{180} \right) 
    \\
    &\hspace{4em} - y_{T,i,1}^{\mathrm{c}}(t)\Big)^2
    + \Big(y_{T,i,2}(\tau \vert t) - y_{T,i,1}^{\mathrm{c}}(t) \\
    &\hspace{4em} - y_{T,i}^{\mathrm{r}}(t)\sin\left(\frac{2\pi \tau}{T} + (i-1)\frac{45\pi}{180}\right) \Big)^2 \bigg)\\
    &\hspace{1em} + \sum_{\tau=0}^{T-1} \frac{1}{10T}\sum_{j=1}^m \Vert y_{T,i,3}(\tau \vert t) - y_{T,j,3}(\tau \vert t) \Vert^2 
    \\
    &\hspace{1em} - y_{T,i}^{\mathrm{r}} + \sum_{j=1}^m \Vert y_{T,i}^{\mathrm{c}} - y_{T,j}^{\mathrm{c}}\Vert^2 + \Vert y_{T,i}^{\mathrm{r}} - y_{T,j}^{\mathrm{r}}\Vert^2.
\end{align*}
The first two terms reward trajectories that follow a circle with an agent-specific phase, the third rewards consensus on the altitude, the fourth pushes for a large radius, and the remaining terms reward consensus \change{on} the circle's centre and radius.

For $t \ge 350$, we switch to $N = 30$ and $T = 1$, and use
\begin{align*}
    &W^{\mathrm{c}}_1(y_{T,1}(\cdot\vert t)) = \Vert y_{T,1}( 0 \vert t) - y_r(t) \Vert^2, \\
    &W^{\mathrm{c}}_i(y_{T,i}(\cdot\vert t)) = \Vert 
    y_{T,i}(0\vert t) - y_{T,1}(0\vert t) \Vert_{\change{G}}^2, \quad i > 1,
\end{align*}
\change{with a diagonal matrix $G$ with diagonal $[1, 1, 0.1]$.}
The external reference $y_r$ is defined as $y_r(t) = \bigl(-10 + 20\frac{t-350}{350}\bigr)[1,\,1,\,0]^\top.$

The simulation results are depicted in Figure~\ref{fig:quadrotor_first_task} for the first phase and in Figure~\ref{fig:quadrotor_second_task} for the second phase.
At all times, the quadrotors are further than \SI{0.4}{\metre} apart.

\begin{figure}[tb]
    \setlength\axisheight{0.35\linewidth}
    \setlength\axiswidth{0.95\linewidth}
    \centering
    \begin{tikzpicture}
    \definecolor{vivid_blue}{HTML}{1E90FF}  
    \definecolor{bright_green}{HTML}{32CD32}  
    \definecolor{vibrant_orange}{HTML}{FFA500}  
    \definecolor{tomato_red}{HTML}{FF6347}  
    \definecolor{medium_purple}{HTML}{9370DB}  
    \definecolor{calm_teal}{HTML}{20B2AA}  
    
    \begin{groupplot}[
      group style={
        group name=my plots,
        group size=1 by 3,
        vertical sep=0.25cm
      },
      height=\axisheight,
      width=\axiswidth,
      legend cell align={left},
      legend style={
          at={(1,1)},
          anchor=north east,
          draw=lightgray,
          fill opacity=0.8, 
          draw opacity=1,
          text opacity=1,
          font=\small
      },
      minor x tick num=4,
      minor y tick num=0,
      tick align=outside,
      tick pos=left,
      x grid style={gray},
      xmajorgrids,
      xmin=-1, 
      xmax=350,
      xtick distance=50,
      xtick style={color=black},
      y grid style={gray},
      ymajorgrids,
      ytick style={color=black},
      label style={font=\small},
      tick label style={font=\footnotesize},
    ]
    
    \nextgroupplot[ylabel={$x_{i,1}$ in \si{\metre}}, xticklabels={}]
    
    \addplot [vivid_blue, line width=0.35mm]
    table [col sep=space] {./plotdata/quadrotor/A1_x1_phase1.tex};
    \addlegendentry{$x_{1,1}$}
    
    \addplot [bright_green, dash pattern=on 8pt off 4pt, line width=0.35mm]
    table [col sep=space] {./plotdata/quadrotor/A2_x1_phase1.tex};
    \addlegendentry{$x_{2,1}$}
    
    \addplot [vibrant_orange, dash pattern=on 6pt off 2pt on 1pt off 2pt, line width=0.35mm]
    table [col sep=space] {./plotdata/quadrotor/A3_x1_phase1.tex};
    \addlegendentry{$x_{3,1}$}
    
    \addplot [tomato_red, dash pattern=on 3pt off 1pt on 1pt off 1pt, line width=0.35mm]
    table [col sep=space] {./plotdata/quadrotor/A4_x1_phase1.tex};
    \addlegendentry{$x_{4,1}$}
    
    
    
    \nextgroupplot[ylabel={$x_{i,2}$ in \si{\metre}}, xticklabels={}]
    
    \addplot [vivid_blue, line width=0.35mm]
    table [col sep=space] {./plotdata/quadrotor/A1_x2_phase1.tex};
    
    \addplot [bright_green, dash pattern=on 8pt off 4pt, line width=0.35mm]
    table [col sep=space] {./plotdata/quadrotor/A2_x2_phase1.tex};
    
    \addplot [vibrant_orange, dash pattern=on 6pt off 2pt on 1pt off 2pt, line width=0.35mm]
    table [col sep=space] {./plotdata/quadrotor/A3_x2_phase1.tex};
    
    \addplot [tomato_red, dash pattern=on 3pt off 1pt on 1pt off 1pt, line width=0.35mm]
    table [col sep=space] {./plotdata/quadrotor/A4_x2_phase1.tex};
    
    

    \nextgroupplot[ylabel={$x_{i,3}$ in \si{\metre}}, xlabel={$t$ (time steps)}]
    
    \addplot [vivid_blue, line width=0.35mm]
    table [col sep=space] {./plotdata/quadrotor/A1_x3_phase1.tex};
    
    \addplot [bright_green, dash pattern=on 8pt off 4pt, line width=0.35mm]
    table [col sep=space] {./plotdata/quadrotor/A2_x3_phase1.tex};
    
    \addplot [vibrant_orange, dash pattern=on 6pt off 2pt on 1pt off 2pt, line width=0.35mm]
    table [col sep=space] {./plotdata/quadrotor/A3_x3_phase1.tex};
    
    \addplot [tomato_red, dash pattern=on 3pt off 1pt on 1pt off 1pt, line width=0.35mm]
    table [col sep=space] {./plotdata/quadrotor/A4_x3_phase1.tex};
    
    
    
    \end{groupplot}
\end{tikzpicture}
    
    \caption{Positions of the quadrotors during the first phase of the cooperative task from $t=0$ to $t=350$.
    The agents converge to a circular trajectory with a common centre and radius, but follow it with different phase, as was desired.}
    \label{fig:quadrotor_first_task}
\end{figure}
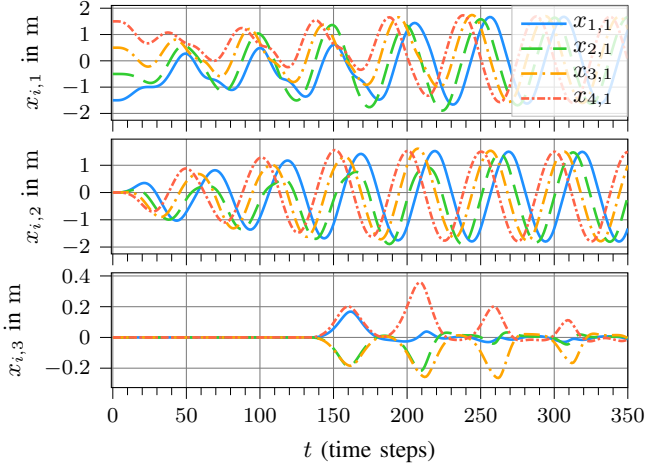%
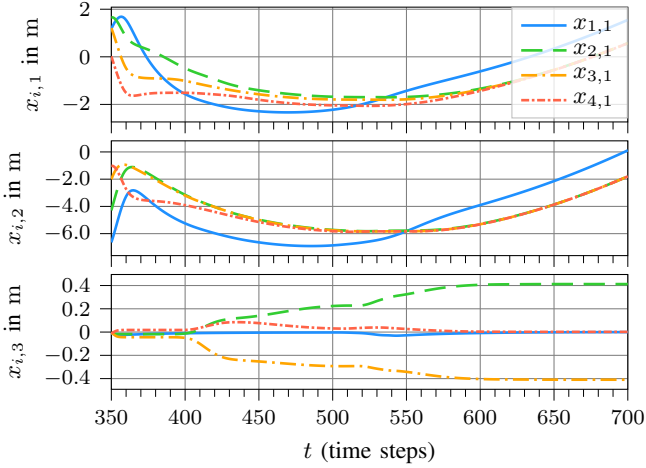
\begin{figure}[tb]
    \setlength\axisheight{0.35\linewidth}
    \setlength\axiswidth{0.95\linewidth}
    \centering
    \begin{tikzpicture}
    \definecolor{vivid_blue}{HTML}{1E90FF}  
    \definecolor{bright_green}{HTML}{32CD32}  
    \definecolor{vibrant_orange}{HTML}{FFA500}  
    \definecolor{tomato_red}{HTML}{FF6347}  
    \definecolor{medium_purple}{HTML}{9370DB}  
    \definecolor{calm_teal}{HTML}{20B2AA}  
    
    \begin{groupplot}[
      group style={
        group name=my plots,
        group size=1 by 3,
        vertical sep=0.25cm
      },
      height=\axisheight,
      width=\axiswidth,
      legend cell align={left},
      legend style={
          at={(1,1)},
          anchor=north east,
          draw=lightgray,
          fill opacity=0.8, 
          draw opacity=1,
          text opacity=1,
          font=\small
      },
      minor x tick num=4,
      minor y tick num=0,
      tick align=outside,
      tick pos=left,
      x grid style={gray},
      xmajorgrids,
      xmin=350, 
      xmax=700,
      xtick distance=50,
      xtick style={color=black},
      y grid style={gray},
      ymajorgrids,
      ytick style={color=black},
      label style={font=\small},
      tick label style={font=\footnotesize},
    ]
    
    \nextgroupplot[ylabel={$x_{i,1}$ in \si{\metre}}, xticklabels={}]
    
    \addplot [vivid_blue, line width=0.35mm]
    table [col sep=space] {./plotdata/quadrotor/A1_x1_phase2.tex};
    \addlegendentry{$x_{1,1}$}
    
    \addplot [bright_green, dash pattern=on 8pt off 4pt, line width=0.35mm]
    table [col sep=space] {./plotdata/quadrotor/A2_x1_phase2.tex};
    \addlegendentry{$x_{2,1}$}
    
    \addplot [vibrant_orange, dash pattern=on 6pt off 2pt on 1pt off 2pt, line width=0.35mm]
    table [col sep=space] {./plotdata/quadrotor/A3_x1_phase2.tex};
    \addlegendentry{$x_{3,1}$}
    
    \addplot [tomato_red, dash pattern=on 3pt off 1pt on 1pt off 1pt, line width=0.35mm]
    table [col sep=space] {./plotdata/quadrotor/A4_x1_phase2.tex};
    \addlegendentry{$x_{4,1}$}
    
    
    
    \nextgroupplot[ylabel={$x_{i,2}$ in \si{\metre}}, xticklabels={}, yticklabels={-2.0, $-6.0$, $-4.0$, $-2.0$, $0$}]
    
    \addplot [vivid_blue, line width=0.35mm]
    table [col sep=space] {./plotdata/quadrotor/A1_x2_phase2.tex};
    
    \addplot [bright_green, dash pattern=on 8pt off 4pt, line width=0.35mm]
    table [col sep=space] {./plotdata/quadrotor/A2_x2_phase2.tex};
    
    \addplot [vibrant_orange, dash pattern=on 6pt off 2pt on 1pt off 2pt, line width=0.35mm]
    table [col sep=space] {./plotdata/quadrotor/A3_x2_phase2.tex};
    
    \addplot [tomato_red, dash pattern=on 3pt off 1pt on 1pt off 1pt, line width=0.35mm]
    table [col sep=space] {./plotdata/quadrotor/A4_x2_phase2.tex};
    
    

    \nextgroupplot[ylabel={$x_{i,3}$ in \si{\metre}}, xlabel={$t$ (time steps)}]
    
    \addplot [vivid_blue, line width=0.35mm]
    table [col sep=space] {./plotdata/quadrotor/A1_x3_phase2.tex};
    
    \addplot [bright_green, dash pattern=on 8pt off 4pt, line width=0.35mm]
    table [col sep=space] {./plotdata/quadrotor/A2_x3_phase2.tex};
    
    \addplot [vibrant_orange, dash pattern=on 6pt off 2pt on 1pt off 2pt, line width=0.35mm]
    table [col sep=space] {./plotdata/quadrotor/A3_x3_phase2.tex};
    
    \addplot [tomato_red, dash pattern=on 3pt off 1pt on 1pt off 1pt, line width=0.35mm]
    table [col sep=space] {./plotdata/quadrotor/A4_x3_phase2.tex};
    
    
    
    \end{groupplot}
\end{tikzpicture}
    
    \caption{Positions of the quadrotors during the second phase of the cooperative task from $t=350$ to $t=700$.
    The first agent follows a reference, whereas the rest follow it without colliding.}
    \label{fig:quadrotor_second_task}
\end{figure}%

\section{Conclusion}
We presented a distributed MPC framework for dynamic cooperative control of multi-agent systems.
The scheme decouples individual agent dynamics from the cooperative objective, enabling flexible and scalable coordination.
We provided conditions for the asymptotic achievement of the cooperative task, along with transient and asymptotic performance guarantees.

The framework was demonstrated in three scenarios:
(i) periodic motion and robustness to changing communication topologies in satellite constellations,
(ii) deadlock avoidance in narrow-passage traversal via appropriate objective design, and
(iii) sequential execution of cooperative tasks in a system of quadrotors.
These examples underscore the flexibility of the proposed approach.

While we assumed a time-invariant topology in the theoretical analysis, the framework \change{can remain} applicable under time-varying communication structures, provided that initial feasibility is ensured at each topology change.
\change{%
This is because relevant components, such as terminal constraints, do not need to be redesigned.
For each interval with constant topology, the theoretical results, e.g. asymptotic stability, hold anew.
The penalty on changes in the cooperation output should be removed in the first optimization problem after a topology change.
While persistent topology changes preclude conclusions about overall asymptotic behaviour, recursive feasibility is preserved due to the decentralized design of the terminal components.
}
Moreover, although the scheme does not require an external coordinator, it can seamlessly incorporate external references when hierarchical control structures or multiscale coordination are desired.

\appendices



\section{Proofs}
\subsection{Proof of Theorem~\ref{thm:exponential_stability}}
\begin{proof}
    First, we show that $\eta_{\ell}$ in Theorem~\ref{thm:stage_cost_upper_bounds_cooperation_distance} is also a quadratic function.
    We follow the proof of Theorem~\ref{thm:stage_cost_upper_bounds_cooperation_distance} until~\eqref{eq:stage_cost_upper_bounds_cooperation_distance_before_combination}.
    Then, we have 
    \begin{align*}
        &{\sum_{i=1}^{m}} J_i(x_i, \hat{u}_i, \hat{y}_{T,i}, y_{T,i}^{\mathrm{pr}}, \hat{y}_{T, \mathcal{N}_i}) - J_i(x_i, u_i^0, y_{T,i}^0, y_{T,i}^{\mathrm{pr}}, y_{T,\mathcal{N}_i}^0)
        \\
        &< c^{f\ell} \big(\eta_{\ell}(\vert y_T^0 \vert_{\mathcal{Y}_T^W}) - c_{\theta}^{W} a_{\psi}^2\vert y_T^0 \vert_{\mathcal{Y}_T^W}^2 \big).
    \end{align*}
    Now, choosing $\eta_{\ell}(\vert y_T \vert_{\mathcal{Y}_T^W}) = \delta\frac{c_{\theta}^{W}}{2}a_{\psi}^2 \vert y_T \vert_{\mathcal{Y}_T^W}^2 = c_{\eta}\vert y_T \vert_{\mathcal{Y}_T^W}^2$ with $\delta\in(0,1]$ and $\delta \le \frac{\varepsilon}{c^{\ell}c_{\theta}^{W} a_{\psi}^2\gamma_W^2}$ leads to the same contradiction as in the proof of Theorem~\ref{thm:stage_cost_upper_bounds_cooperation_distance} and establishes $\eta_{\ell}$ to be quadratic as well.

    Second, we show that the upper bound shown in Lemma~\ref{lem:Lyapunov_upper_bound} is quadratic.
    We follow the proof of Lemma~\ref{lem:Lyapunov_upper_bound} until~\eqref{eq:Lyapunov_upper_bound_preliminary}, from which, with $\tilde{a}_{\mathrm{ub}} = \max_i (c_i^{\mathrm{f}} a_{\mathrm{ub}}^{\ell_i}, \lambda(N)a_{\mathrm{ub}}^{\Delta})$, we have
    \begin{align*}
        &V(\xi(\hspace{-0.5pt}\tau\hspace{-0.5pt})) \hspace{-1pt}\le\hspace{-1pt} {\sum_{i=1}^{m}} c_i^{\mathrm{f}}a_{\mathrm{ub}}^{\ell_i}\vert x_{{\mathcal{T}},i}(\hspace{-0.5pt}\tau\hspace{-0.5pt}) \vert_{x_{T,i}^{W}(\hspace{-0.5pt}\tau\hspace{-0.5pt})}^2 \hspace{-0.8pt}+\hspace{-0.8pt} \lambda(N)a_{\mathrm{ub}}^{\Delta}\vert y_{{\mathcal{T}},\tau+1} \vert_{y_{T,\tau+1}^{W}}^2
        \\
        &\le \hspace{-2pt}\tilde{a}_{\mathrm{ub}} \hspace{-1pt}(\vert x_{\mathcal{T}}(\hspace{-0.5pt}\tau\hspace{-0.5pt})\vert_{x_{T}^W\hspace{-2pt}(\hspace{-0.5pt}\tau\hspace{-0.5pt})}^2 \hspace{-2.5pt}+\hspace{-2pt} \vert y_{{\mathcal{T}},\hspace{-0.5pt}\tau+1} \vert_{y_{T,\hspace{-0.5pt}\tau+1}^{W}}^2\hspace{-1pt}) \hspace{-2pt}=\hspace{-2pt} \tilde{a}_{\mathrm{ub}}\vert \xi(\hspace{-0.5pt}\tau\hspace{-0.5pt}) \vert^2_{(x_{T}^W\hspace{-2pt}(\hspace{-0.5pt}\tau\hspace{-0.5pt}), y_{T\hspace{-0.5pt},\tau+1}^{W})}.
    \end{align*}
    This yields $V_T(\xi_T) \le \tilde{a}_{\mathrm{ub}} \sum_{\tau=0}^{T-1}\vert \xi(\tau) \vert^2_{(x_{T}^W(\tau), y_{T,\tau+1}^{W})} = \tilde{a}_{\mathrm{ub}}\vert \xi_T \vert_{\Xi_T^W}^2$.
    The local upper bound on a compact subset of $\mathcal{X}_N^{T} \times \mathcal{Y}_T^{T}$ and the compactness of $\mathcal{X}_N$ and $\mathcal{Y}_T$ entail $V_T(\xi_T) \le a_{\mathrm{ub}} \vert \xi_T \vert_{\Xi_T^W}^2$ with $a_{\mathrm{ub}} > 0$ for all $\xi_T$ whose first component $x$ satisfies $x\in\mathcal{X}_N$ (cf.~\cite[Prop. 2.16]{Rawlings2020}).

    Third, we establish a quadratic lower bound.
    We follow the derivation in the proof of Lemma~\ref{lem:Lyapunov_lower_bound} until~\eqref{eq:Lyapunov_lower_bound_with_stage_cost_and_penalty_on_change}, yielding
    $V(\xi(\tau)) \ge \hat{a}_{\mathrm{lb}}\big( \vert y_T^0(\cdot\vert\xi(\tau)) \vert_{\mathcal{Y}_T^W}^2 + \vert x_{\mathcal{T}}(\tau) \vert_{x_T^0(0 \vert \xi(\tau))}^2
         + \vert y_{{\mathcal{T}},\tau+1} \vert_{y_T^0(\cdot\vert\xi(\tau))}^2 \big)$
    with $\hat{a}_{\mathrm{lb}} = \frac{1}{2}\min_{i} (a_{\mathrm{lb}}^{\ell_i},c_{\eta},a_{\mathrm{lb}}^{\Delta})$.
    \change{%
    The same arguments as in the proof of~\ref{lem:Lyapunov_lower_bound} show that the right-hand side
    is positive definite with respect to $\Xi_T^W$.
    From here, it is possible to show that there exists $a_{\mathrm{lb}}>0$ such that $V_T(\xi_T) \ge a_{\mathrm{lb}} \vert \xi_T \vert_{\Xi_T^W}^2$.
    }%
    Now, we want to find $a_{\mathrm{lb}}>0$ such that $V_T(\xi_T) \ge a_{\mathrm{lb}} \vert \xi_T \vert_{\Xi_T^W}^2$.
    Define $(x_T^W, y_T^W) = \argmin_{(\hat{x}_T^W, \hat{y}_T^W) \in \Xi_T^W} \Vert \xi_T - (\hat{x}_T^W, \hat{y}_T^W) \Vert$ and $\bar{y}_T(\cdot\vert \xi)$ such that $\vert y_T^0 \vert_{\mathcal{Y}_T^W} = \Vert y_T^0(\cdot\vert \xi) - \bar{y}_T(\cdot\vert \xi) \Vert$.
    Additionally, define $(\hat{x}_T, \hat{y}_T) = \smash[b]{\argmin_{(\tilde{x}_T,\tilde{y}_T) \in \Xi_T^W} \sum_{\tau=0}^{T-1} \Vert \bar{y}_T(\cdot\vert \xi(\tau)) - \tilde{y}_{T,\tau+1}\Vert}$ subject to $\Vert \bar{y}_T(\cdot\vert \xi(0)) - \tilde{y}_{T,1}\Vert = 0$.
    We now introduce some preparatory results before deriving the lower bound.
    We have, with $\hat{y}_{T,\tau+1} = \hat{y}_{T,\tau}(\cdot+1)$, for $\tau\in\mathbb{I}_{1:T-1}$,
    \begin{align*}
        &\vert \bar{y}_T(\cdot\vert\xi(\tau))\vert_{\hat{y}_{T,\tau+1}}^2 
        \\
        &\le 2\vert \bar{y}_T(\cdot\vert\xi(\tau)) \vert_{\bar{y}_T(\cdot+1\vert\xi(\tau-1))}^2 + 2\vert\bar{y}_T(\cdot+1\vert\xi(\tau-1)) \vert_{\hat{y}_{T,\tau+1}}^2
        \\
        &= 2\vert \bar{y}_T(\cdot\vert\xi(\tau)) \vert_{\bar{y}_T(\cdot+1\vert\xi(\tau-1))}^2 + 2\vert\bar{y}_T(\cdot\vert\xi(\tau-1)) \vert_{\hat{y}_{T,\tau}}^2
    \end{align*}
    By iterating these steps until $\vert\bar{y}_T(\cdot\vert\xi(0)) \vert_{\hat{y}_{T,1}}^2 = 0$ appears, we obtain, for $\tau\in\mathbb{I}_{1:T-1}$,
    \begin{equation}\label{eq:quadratic_bound_bar_to_hat}
        \vert \bar{y}_T(\cdot\vert\xi(\tau))\vert_{\hat{y}_{T,\tau+1}}^2 \le {\sum_{k=0}^{\tau-1}} 2^{k+1} \vert \bar{y}_T(\cdot\vert\xi(\tau-k))\vert_{\bar{y}_T(\cdot+1\vert\xi(\tau-k-1))}^2.
    \end{equation}
    Next, for $\tau\in\mathbb{I}_{1:T-1}$,
    \begin{align}\label{eq:quadratic_bound_bar_to_previous_bar}
        &\vert \bar{y}_T(\cdot\vert\xi(\tau))\vert_{\bar{y}_T(\cdot+1\vert\xi(\tau-1))}^2 \notag
        \\
        &\le 2\vert \bar{y}_T(\cdot\vert\xi(\tau))\vert_{y_T^0(\cdot\vert\xi(\tau))}^2 + 2 \vert y_T^0(\cdot\vert\xi(\tau)) \vert_{\bar{y}_T(\cdot+1\vert\xi(\tau-1))}^2 \notag
        \\
        &\le 2\vert \bar{y}_T(\cdot\vert\xi(\tau))\vert_{y_T^0(\cdot\vert\xi(\tau))}^2 + 4 \vert y_T^0(\cdot\vert\xi(\tau)) \vert_{y_T^0(\cdot+1\vert\xi(\tau-1))}^2 \notag
        \\
        &\phantom{\le{}} + 4\vert y_T^0(\cdot+1\vert\xi(\tau-1)) \vert_{\bar{y}_T(\cdot+1\vert\xi(\tau-1))}^2 \notag
        \\
        &= 2\vert \bar{y}_T(\cdot\vert\xi(\tau))\vert_{y_T^0(\cdot\vert\xi(\tau))}^2 + 4 \vert y_{{\mathcal{T}},\tau+1} \vert_{ y_T^0(\cdot\vert\xi(\tau))}^2 \notag
        \\
        &\phantom{\le{}} + 4\vert y_T^0(\cdot\vert\xi(\tau-1)) \vert_{\bar{y}_T(\cdot\vert\xi(\tau-1))}^2.
    \end{align}
    Combining these two, we get
    \begin{align}\label{eq:quadratic_bound_sum_bar_to_hat}
    &{\sum_{\tau=1}^{T-1}} \vert \bar{y}_T(\cdot\vert\xi(\tau))\vert_{\hat{y}_{T,\tau+1}}^2 \notag
    \\
    &\stackrel{\eqref{eq:quadratic_bound_bar_to_hat}}{\le} \smash[t]{\sum_{\tau=1}^{T-1} \sum_{k=0}^{\tau-1}} 2^{k+1} \vert \bar{y}_T(\cdot\vert\xi(\tau-k))\vert_{\bar{y}_T(\cdot+1\vert\xi(\tau-k-1))}^2 \notag
    \\
    &\stackrel{\eqref{eq:quadratic_bound_bar_to_previous_bar}}{\le} \smash[b]{\sum_{\tau=1}^{T-1}} 2^{\tau+2} \smash[b]{\sum_{k=0}^{\tau-1}} \Big( \vert \bar{y}_T(\cdot\vert\xi(\tau-k))\vert_{y_T^0(\cdot\vert\xi(\tau-k))}^2 \notag
    \\
    &\hspace{7.2em} + \vert y_{{\mathcal{T}},\tau-k+1} \vert_{ y_T^0(\cdot\vert\xi(\tau-k))}^2 \notag 
    \\
    &\hspace{7.2em} + \vert y_T^0(\cdot\vert\xi(\tau-k-1)) \vert_{\bar{y}_T(\cdot\vert\xi(\tau-k-1))}^2\Big) \notag
    \\
    &\le 2^{T+2} \Big(\smash[t]{\sum_{\tau=1}^{T-1}} (T-\tau+1)\big( \vert \bar{y}_T(\cdot\vert\xi(\tau))\vert_{y_T^0(\cdot\vert\xi(\tau))}^2 \notag
    \\
    &\phantom{\le{}} + \vert y_{{\mathcal{T}},\tau+1} \vert_{ y_T^0(\cdot\vert\xi(\tau))}^2\big) + \smash[t]{\sum_{\tau=0}^{T-2}} (T-\tau) \vert \bar{y}_T(\cdot\vert\xi(\tau))\vert_{y_T^0(\cdot\vert\xi(\tau))}^2 \Big)\notag
    \\
    &\le 2^{T+2}T \Big(\smash[t]{\sum_{\tau=1}^{T-1}} \big( \vert y_T^0(\cdot\vert\xi(\tau))\vert_{\bar{y}_T(\cdot\vert\xi(\tau))}^2 \notag
    \\
    &\phantom{\le{}} + \vert y_{{\mathcal{T}},\tau+1} \vert_{ y_T^0(\cdot\vert\xi(\tau))}^2\big) + {\sum_{\tau=0}^{T-2}} \vert y_T^0(\cdot\vert\xi(\tau)) \vert_{\bar{y}_T(\cdot\vert\xi(\tau))}^2 \Big).
    \end{align}
    Next, consider $\vert y_{{\mathcal{T}},\tau+1} \vert_{\hat{y}_{T,\tau+1}}^2$.
    Since $ \vert \bar{y}_{T}(\cdot\vert \xi(0)) \vert_{\hat{y}_{T,1}} = 0$, 
    \begin{align}\label{eq:quadratic_bound_y_to_hat_tau_0}
        &\vert y_{{\mathcal{T}},1} \vert_{\hat{y}_{{\mathcal{T}},1}}^2 \hspace{-3pt}=\hspace{-2pt} (\hspace{-1pt}\vert y_{\mathcal{T},1} \vert_{\bar{y}_{T}(\hspace{-1pt}\cdot\vert \xi(\hspace{-1pt}0\hspace{-1pt}))} \hspace{-3pt}+\hspace{-2pt}  \vert \bar{y}_{T}(\cdot\vert \xi(\hspace{-1pt}0\hspace{-1pt})) \vert_{\hat{y}_{T,1}}\hspace{-1pt})^2 \hspace{-2pt}=\hspace{-2pt} \vert y_{\mathcal{T},1} \vert_{\bar{y}_{T}(\cdot\vert \xi(0))}^2\notag\\
        &\le 2 (\vert y_{{\mathcal{T}},1} \vert_{y_{T}^0(\cdot\vert \xi(0))}^2 + \vert y_{T}^0(\cdot\vert \xi(0)) \vert_{\bar{y}_{T}(\cdot\vert \xi(0))}^2),
    \end{align}
    and for $\tau\in\mathbb{I}_{1:T-1}$,
    \begin{align}\label{eq:quadratic_bound_y_to_hat_tau_1_to_T}
        &\vert y_{{\mathcal{T}},\tau+1} \vert_{\hat{y}_{T,\tau+1}}^2 \le 2 \vert y_{{\mathcal{T}},\tau+1} \vert_{\bar{y}_T(\cdot\vert\xi(\tau))}^2 + 2\vert \bar{y}_{T}(\cdot\vert \xi(\tau))\vert_{\hat{y}_{T,\tau+1}}^2 \notag
        \\
        &\stackrel{\eqref{eq:quadratic_bound_bar_to_hat}}{\le } 2 \vert y_{{\mathcal{T}},\tau+1} \vert_{\bar{y}_T(\cdot\vert\xi(\tau))}^2 \notag\\
        &\phantom{\le{}} + {\sum_{k=0}^{\tau-1}} 2^{k+1} \vert \bar{y}_T(\cdot\vert\xi(\tau-k))\vert_{\bar{y}_T(\cdot+1\vert\xi(\tau-k-1))}^2 \notag
        \\
        &\le 4 \vert y_{{\mathcal{T}},\tau+1} \vert_{y_T^0(\cdot\vert\xi(\tau))}^2 + 4 \vert y_T^0(\cdot\vert\xi(\tau)) \vert_{\bar{y}_T(\cdot\vert\xi(\tau))}^2 \notag\\
        &\phantom{\le{}} + {\sum_{k=0}^{\tau-1}} 2^{k+1} \vert \bar{y}_T(\cdot\vert\xi(\tau-k))\vert_{\bar{y}_T(\cdot+1\vert\xi(\tau-k-1))}^2.
    \end{align}
    Hence,
    \begin{align}\label{eq:quadratic_bound_sum_y_to_hat}
        &\sum_{\tau=0}^{T-1} \vert y_{{\mathcal{T}},\tau+1}\vert_{\hat{y}_{T,\tau+1}}^2 \notag
        \\
        &\stackrel{\eqref{eq:quadratic_bound_y_to_hat_tau_0},\,\eqref{eq:quadratic_bound_y_to_hat_tau_1_to_T}}{\le} 2 \vert y_{{\mathcal{T}},1} \vert_{y_{T}^0(\cdot\vert \xi(0))}^2 + 2 \vert y_{T}^0(\cdot\vert \xi(0)) \vert_{\bar{y}_{T}(\cdot\vert \xi(0))}^2\notag\\
        &\phantom{\le{}} + \sum_{\tau=1}^{T-1}  4 \vert y_{{\mathcal{T}},\tau+1} \vert_{y_T^0(\cdot\vert\xi(\tau))}^2 + 4 \vert y_T^0(\cdot\vert\xi(\tau)) \vert_{\bar{y}_T(\cdot\vert\xi(\tau))}^2 \notag \notag\\
        &\phantom{\le{}} + \sum_{\tau=1}^{T-1}\sum_{k=0}^{\tau-1} 2^{k+1} \vert \bar{y}_T(\cdot\vert\xi(\tau-k))\vert_{\bar{y}_T(\cdot+1\vert\xi(\tau-k-1))}^2 \notag
        \\
        &\stackrel{\eqref{eq:quadratic_bound_sum_bar_to_hat}}{\le} 2 \vert y_{{\mathcal{T}},1} \vert_{y_{T}^0(\cdot\vert \xi(0))}^2 + 2 \vert y_{T}^0(\cdot\vert \xi(0)) \vert_{\bar{y}_{T}(\cdot\vert \xi(0))}^2\notag\\
        &\phantom{\le{}} + \sum_{\tau=1}^{T-1}  4 \vert y_{{\mathcal{T}},\tau+1} \vert_{y_T^0(\cdot\vert\xi(\tau))}^2 + 4 \vert y_T^0(\cdot\vert\xi(\tau)) \vert_{\bar{y}_T(\cdot\vert\xi(\tau))}^2 \notag \\
        &\phantom{\le{}} + 2^{T+2}T \Big(\smash[t]{\sum_{\tau=1}^{T-1}} \big( \vert y_T^0(\cdot\vert\xi(\tau))\vert_{\bar{y}_T(\cdot\vert\xi(\tau))}^2 \notag
        \\
        &\phantom{\le{}} + \vert y_{{\mathcal{T}},\tau+1} \vert_{ y_T^0(\cdot\vert\xi(\tau))}^2\big) + \smash[t]{\sum_{\tau=0}^{T-2}} \vert y_T^0(\cdot\vert\xi(\tau)) \vert_{\bar{y}_T(\cdot\vert\xi(\tau))}^2 \Big) \notag
        \\
        &\le 2^{T+5}T \Big(\smash[t]{\sum_{\tau=0}^{T-1}} \big( \vert y_T^0(\cdot\vert\xi(\tau))\vert_{\bar{y}_T(\cdot\vert\xi(\tau))}^2 + \vert y_{{\mathcal{T}},\tau+1} \vert_{ y_T^0(\cdot\vert\xi(\tau))}^2\big)\Big)
    \end{align}
    We now turn to $\vert x_{\mathcal{T}}(\tau)\vert_{\hat{x}_T(\tau)}^2$, where we use Assumption~\ref{assm:unique_corresponding_equilibrium} with $L_{x} = \min_i L_{x,i}$:
    \begin{align}\label{eq:quadratic_bound_x_to_hat}
        &\vert x_{\mathcal{T}}(\tau)\vert_{\hat{x}_T(\tau)}^2 \le 2 \vert x_T^0(0\vert\xi(\tau)) \vert_{\hat{x}_T(\tau)}^2 + 2 \vert x_{\mathcal{T}}(\tau)\vert_{x_T^0(0\vert\xi(\tau))}^2
        \notag\\
        &\le 2L_x \vert y_T^0(\cdot\vert\xi(\tau)) \vert_{\hat{y}_{T,\tau+1}}^2 + 2 \vert x_{\mathcal{T}}(\tau)\vert_{x_T^0(0\vert\xi(\tau))}^2
        \notag\\
        &\le 4L_x \vert y_{{\mathcal{T}},\tau+1} \vert_{\hat{y}_{T,\tau+1}}^2 + 4L_x \vert y_{{\mathcal{T}},\tau+1} \vert_{y_T^0(\cdot\vert\xi(\tau))}^2 \notag\\
        &\phantom{\le{}} + 2 \vert x_{\mathcal{T}}(\tau)\vert_{x_T^0(0\vert\xi(\tau))}^2.
    \end{align}
    With these preliminary steps combined, we get 
    \begin{align}\label{eq:Lyapuno_quadratic_lower_bound_from_below}
        &\vert \xi_T \vert_{\Xi_T^W}^2 = \sum_{\tau=0}^{T-1} \vert y_{{\mathcal{T}},\tau+1}\vert_{y_{T,\tau+1}^W}^2 + \vert x_{\mathcal{T}}(\tau) \vert_{x_T^W(\tau)}^2 \notag
        \\
        &\le \sum_{\tau=0}^{T-1} \vert y_{{\mathcal{T}},\tau+1}\vert_{\hat{y}_{T,\tau+1}}^2 + \vert x_{\mathcal{T}}(\tau) \vert_{\hat{x}_T(\tau)}^2 \notag
        \\
        &\stackrel{\eqref{eq:quadratic_bound_x_to_hat}}{\le} \sum_{\tau=0}^{T-1} \Big((1+4L_x)\vert y_{{\mathcal{T}},\tau+1} \vert_{\hat{y}_{T,\tau+1}}^2 + 4L_x \vert y_{{\mathcal{T}},\tau+1} \vert_{y_T^0(\cdot\vert\xi(\tau))}^2 \notag\\
        &\hspace{4em} + 2 \vert x_{\mathcal{T}}(\tau)\vert_{x_T^0(0\vert\xi(\tau))}^2\Big)\notag
        \\
        &\stackrel{\eqref{eq:quadratic_bound_sum_y_to_hat}}{\le} \sum_{\tau=0}^{T-1} \Big((1 + 4L_x)(2^{T+5}T) \vert y_T^0(\cdot\vert\xi(\tau))\vert_{\bar{y}_T(\cdot\vert\xi(\tau))}^2 \notag
        \\
        &\hspace{4em} + (1 + 4L_x)(2^{T+5}T)\vert y_{{\mathcal{T}},\tau+1} \vert_{ y_T^0(\cdot\vert\xi(\tau))}^2 \notag
        \\
        &\hspace{4em}+ 4L_x \vert y_{{\mathcal{T}},\tau+1} \vert_{y_T^0(\cdot\vert\xi(\tau))}^2 + 2 \vert x_{\mathcal{T}}(\tau)\vert_{x_T^0(0\vert\xi(\tau))}^2\Big)\notag
        \\
        &\le \big((1 + 4L_x)(2^{T+5}T) + 4L_x\big)\Big(\sum_{\tau=0}^{T-1} \vert x_{\mathcal{T}}(\tau)\vert_{x_T^0(0\vert\xi(\tau))}^2\notag
        \\
        &\hspace{4em} + \vert y_T^0(\cdot\vert\xi(\tau))\vert_{\bar{y}_T(\cdot\vert\xi(\tau))}^2 + \vert y_{{\mathcal{T}},\tau+1} \vert_{ y_T^0(\cdot\vert\xi(\tau))}^2 \Big).
    \end{align}
    Finally, comparing~\eqref{eq:Lyapunov_function_intermediate_lower_bound} with~\eqref{eq:Lyapuno_quadratic_lower_bound_from_below} entails $V_T(\xi_T) \ge a_{\mathrm{lb}} \vert \xi_T \vert_{\Xi_T^W}^2$ with $a_{\mathrm{lb}} = \frac{\hat{a}_{\mathrm{lb}}}{(1 + 4L_x)(2^{T+5}T) + 4L_x}$.

    Furthermore, we also get from~\eqref{eq:Lyapunov_decrease}, $V_T(\xi_T(t+1)) - V_T(\xi_T(t)) \le - a_{\mathrm{lb}} \vert \xi_T \vert_{\Xi_T^W}^2$.
    Exponential stability then follows from this and the quadratic upper and lower bounds using standard arguments (cf.~\cite[Thm. B.19]{Rawlings2020}).
\end{proof}
\subsection{Proof of equation~\eqref{eq:weak_strong_convexity_lower_bound} in Lemma~\ref{lem:sufficient_conditions_for_better_candidate_weak_strong_convexity}'s proof}
We follow the proof of~\cite[Prop. B.5]{Bertsekas2016} with adapted notation.
Choose $\bar{\sigma} < \max(L_W, \frac{\sigma}{2})$ with $\sigma$ from Assumption~\ref{assm:weak_strong_convexity}, let $y_T, y_T'\in \mathcal{Y}_T$, and consider $\phi(y_T) = W^{\mathrm{c}}(y_T) - \frac{\bar{\sigma}}{2}\Vert y_T \Vert^2$.
We show that $\nabla \phi(y_T) = \nabla W^{\mathrm{c}}(y_T) - \bar{\sigma}y_T$ is Lipschitz continuous with constant $L_W - \bar{\sigma}$.
From~\cite[Prop. B.3]{Bertsekas2016},~\eqref{eq:Lipschitz_continuous_cooperation_objective_function} is equivalent to 
\begin{equation*}
    (\nabla W^{\mathrm{c}}(y_T) - \nabla W^{\mathrm{c}}(y_T'))^\top (y_T - y_T') \le L_W \Vert y_T - y_T' \Vert^2.
\end{equation*}
Moreover,
\begin{align*}
    &(\nabla W^{\mathrm{c}}(y_T) - \nabla W^{\mathrm{c}}(y_T'))^\top (y_T - y_T') 
    \\
    &= (\nabla \phi(y_T) + \bar{\sigma}y_T - \nabla \phi(y_T') - \bar{\sigma}y_T')^\top (y_T - y_T')
    \\
    &= (\nabla \phi(y_T) - \nabla \phi(y_T'))^\top (y_T - y_T') + \bar{\sigma} \Vert y_T - y_T' \Vert^2
\end{align*}
which implies
\begin{equation*}
    (\nabla \phi(y_T) - \nabla \phi(y_T'))^\top (y_T - y_T') \le (L_W-\bar{\sigma}) \Vert y_T - y_T' \Vert^2.
\end{equation*}
Again from~\cite[Prop. B.3]{Bertsekas2016}, this shows the claimed Lipschitz continuity of $\nabla \phi(y_T)$, and we obtain from~\cite[Prop. B.3]{Bertsekas2016}
\begin{equation*}
    (\nabla \phi(y_T) - \nabla \phi(y_T'))^\top (y_T - y_T') \ge \frac{\Vert \nabla \phi(y_T) - \nabla \phi(y_T') \Vert^2}{L_W-\bar{\sigma}}.
\end{equation*}
Inserting the definition of $\phi(y_T)$ yields 
\begin{align*}
    &(\nabla W^{\mathrm{c}}(y_T) - \nabla W^{\mathrm{c}}(y_T'))^\top (y_T - y_T') - \bar{\sigma}\Vert y_T - y_T' \Vert^2 
    \\
    &\ge \frac{1}{L_W-\bar{\sigma}}\Vert \nabla W^{\mathrm{c}}(y_T) - \nabla W^{\mathrm{c}}(y_T') - \bar{\sigma}(y_T - y_T') \Vert^2
    \\
    &\ge \frac{\bar{\sigma}^2}{L_W-\bar{\sigma}}\Vert y_T - y_T'\Vert^2 + \frac{\Vert \nabla W^{\mathrm{c}}(y_T) - \nabla W^{\mathrm{c}}(y_T') \Vert^2}{L_W - \bar{\sigma}} \\
    &\phantom{\ge{}} - \frac{2\bar{\sigma}}{L_W - \bar{\sigma}}(\nabla W^{\mathrm{c}}(y_T) - \nabla W^{\mathrm{c}}(y_T'))^\top (y_T - y_T')
\end{align*}
which can be reordered into~\eqref{eq:weak_strong_convexity_lower_bound}.

\subsection{Proof of Proposition~\ref{prop:standard_MPC_performance_bound}}
\begin{proof}
    We begin by showing a turnpike property as in~\cite[Prop. 8.15]{Gruene2017}.
    Fix $\gamma > 0$ and define $\delta_i^{\gamma} = (\alpha_{\mathrm{lb}}^{\ell_i})^{-1}(\frac{\gamma}{P})$.
    We show that for all $N,P\in\mathbb{N}$, $x\in X$, $u\in\mathbb{U}^N(x)$ and $r_T\in\mathcal{Z}_T$ with $\mathcal{P}_N(x,u,r_T) \le \gamma$, the set $Q(x,u,P,N,r_T) = \{ k \in \mathbb{I}_{0:N-1} \mid \vert x_{i,u_i}(k,x_i) \vert_{x_{T,i}(k)} \ge \delta_i^{\gamma}\; \forall i\in\mathbb{I}_{1:m} \}$ contains at most $P$ elements.
    For this, assume that there exist $N, P, x, u, r_T$ with $\mathcal{P}_N(x,u,r_T) \le \gamma$ but $Q(x,u,P,N,r_T)$ has at least $P+1$ elements.
    However, this implies the following contradiction:
    \begin{align*}
        &\mathcal{P}_N(x,u,r_T) 
        \stackrel{\eqref{eq:stage_cost_lower_and_upper_bound}}{\ge} 
        \smash[t]{\sum_{k=0}^{N-1} \sum_{i=1}^{m}} \alpha_{\mathrm{lb}}^{\ell_i}(\vert x_{i,u_i}(k,x_i) \vert_{x_{T,i}(k)}) 
        \\
        &\ge \smash[t]{\sum_{\substack{k\in\mathbb{I}_{0:N-1} \\ \vert x_{i,u_i}(k,x_i) \vert_{x_{T,i}(k)} \ge \delta_i^{\gamma}}} \sum_{i=1}^{m}} \alpha_{\mathrm{lb}}^{\ell_i}(\delta_i^{\gamma}) \ge (P+1) \sum_{i=1}^{m} \frac{\gamma}{P}
        > \gamma.
    \end{align*}

    We now prove the performance bound adapting the proof of~\cite[Thm. 8.22]{Gruene2017}.
    Note that we need to prove the assertion only for sufficiently large $N$ and $K$ since all involved functions are bounded. Hence, $\delta_1(N)$ and $\delta_2(K)$ can always be chosen sufficiently large for small $N$ and $K$. 
    Consider $u_{\epsilon} \in \mathbb{U}^K(x)$ such that $x_{u_\epsilon}(K,x) \in \mathcal{B}_{\kappa}(r_T(K))$ with 
    \begin{equation}\label{eq:epsilon_input_bound}
        \mathcal{P}_K(x,u_{\epsilon},r_T) \le \inf_{\substack{u\in\mathbb{U}^K(x) \\ x_u(K,x) \in \mathcal{B}_{\kappa}(r_T(K))}} \mathcal{P}_K(x, u, r_T) + \epsilon
    \end{equation}
    with an arbitrary but fixed $\epsilon\in(0,1)$.
    The standard stability proof~\cite[Thm. 5.13]{Gruene2017} and~\eqref{eq:standard_value_function_upper_bound} entail 
    \begin{equation}\label{eq:value_function_upper_bounds_best_cost}
        \hspace{-0.5em} \inf_{\substack{u\in\mathbb{U}^K(x) \\ x_u(K,x) \in \mathcal{B}_{\kappa}(r_T(K))}} \hspace{-2.5em} \mathcal{P}_K(x, u, r_T) \le V_K(x, r_T) \le \alpha_{\widetilde{N}}^{\mathrm{s}}(\vert x \vert_{x_T(0)}).
    \end{equation}
    We apply the turnpike property with $\gamma = \sup_{x\in \mathcal{Z}_X} \alpha_{\widetilde{N}}^{\mathrm{s}}(\vert x \vert_{x_T(0)}) + \epsilon$, where $\mathcal{Z}_X = \{ x\in X \mid \exists u\in U,\; (x_i, u_i) \in Z_i\; \forall i\in\mathbb{I}_{1:m}\}$.
    Since the set $Q(x,u,\min(\floor*{\frac{N}{2}}, K-1),K,r_T)$ contains at most $\min(\floor*{\frac{N}{2}}, K-1)$ elements, there exists $k \in \mathbb{I}_{0:\min(\floor*{\frac{N}{2}}, K-1)}$ with $\vert x_{i,u_{\epsilon,i}}(k,x_i) \vert_{x_{T,i}(k)} \le \delta_i^{\gamma}(\min(\floor*{\frac{N}{2}}, K-1))$.
    By choosing $N$ and $K$ sufficiently large, we can ensure $\delta_i^{\gamma}(\min(\floor*{\frac{N}{2}}, K-1)) \le c_i^{\mathrm{b}}$, and hence from Assumption~\ref{assm:terminal_ingredients}, $u_{\epsilon} \in \mathbb{U}^k(x)$, $x_{i,u_{\epsilon,i}}(k, x_i) \in \mathcal{X}_i^{\mathrm{f}}(r_{T,i}(k))$, and thus also $x_{u_{\epsilon}}(k, x) \in \mathbb{X}_{N-k}^{\mathrm{s}}(r_T(\cdot+k))$.
    Then, the dynamic programming principle entails
    \begin{align*}
        &V_N^{\mathrm{s}}(x, r_T) = \hspace{-2.75em}\smash[b]{\inf_{\substack{u\in\mathbb{U}^k(x) \\ x_{u}(k,x)\in \mathbb{X}_{N-k}^{\mathrm{s}}(r_T(\cdot+k))}}} \Big(\mathcal{P}_k(x, u, r_T) 
        \\&\hspace{12em}+ V_{N-k}^{\mathrm{s}}(x_{u}(k,x), r_T(\cdot+k))\Big)
        \\
        &\le \mathcal{P}_k(x, u_{\epsilon}, r_T) + V_{N-k}^{\mathrm{s}}(x_{u_{\epsilon}}(k,x), r_T(\cdot+k)) 
        \\
        &\stackrel{\eqref{eq:epsilon_input_bound},\,\mathrlap{\eqref{eq:terminal_cost_is_upper_bound} }}{\le} \epsilon + \hspace{-2.8em}\inf_{\substack{u\in\mathbb{U}^K(x) \\ x_u(K,x) \in \mathcal{B}_{\kappa}(r_T(K))}}\hspace{-3em} \mathcal{P}_K(x, u, r_T) + \sum_{i=1}^{m} V_{i}^{\mathrm{f}}(x_{i,u_{\epsilon,i}}(k,x), r_{T,i}(k))
        \\
        &\stackrel{\eqref{eq:terminal_cost_upper_bound},\mathrlap{\eqref{eq:stage_cost_lower_and_upper_bound}}}{\le} \epsilon + \hspace{-2.8em}\inf_{\substack{u\in\mathbb{U}^K(x) \\ x_u(K,x) \in \mathcal{B}_{\kappa}(r_T(K))}}\hspace{-3em} \mathcal{P}_K(x, u, r_T) + \sum_{i=1}^{m} c_i^{\mathrm{f}} \alpha_{\mathrm{ub}}^{\ell_i}(\vert x_{i,u_{\epsilon,i}}(k,x) \vert_{r_{T,i}(k)})
        \\
        &\le \epsilon + \hspace{-2.8em}\inf_{\substack{u\in\mathbb{U}^K(x) \\ x_u(K,x) \in \mathcal{B}_{\kappa}(r_T(K))}}\hspace{-3em} \mathcal{P}_K(x, u, r_T) \hspace{-0.1em} + \hspace{-0.1em}\sum_{i=1}^{m} c_i^{\mathrm{f}} \alpha_{\mathrm{ub}}^{\ell_i}(\delta_i^{\gamma}(\min(\floor*{\hspace{-1.2pt}\frac{N}{2}\hspace{-1.2pt}}, K-1)))
        \\
        &\le \hspace{-2pt}\epsilon \hspace{-1pt}+ \hspace{-1.0em}\inf_{\substack{u\in\mathbb{U}^K(x) \\ \hspace{-3.3em}\mathrlap{x_u(K,x) \in \mathcal{B}_{\kappa}(r_T(K))}}} \hspace{-1.0em}\mathcal{P}_K\hspace{-1pt}(\hspace{-1pt}x, u, r_T\hspace{-1pt}) \hspace{-2.5pt} + \hspace{-3.5pt}\smash[b]{\sum_{i=1}^{m}} c_i^{\mathrm{f}} \alpha_{\mathrm{ub}}^{\ell_i}\hspace{-1pt}(\hspace{-1pt}\delta_i^{\gamma}\hspace{-1pt}(\hspace{-1pt}\floor*{\hspace{-1.5pt}\frac{N}{2}\hspace{-1.5pt}}\hspace{-1pt})\hspace{-1pt})\hspace{-1pt} \hspace{-2pt}+\hspace{-2pt} c_i^{\mathrm{f}} \alpha_{\mathrm{ub}}^{\ell_i}\hspace{-1pt}(\hspace{-1pt}\delta_i^{\gamma}\hspace{-1pt}(\hspace{-1pt}K\hspace{-2pt}-\hspace{-2pt}1\hspace{-1pt})\hspace{-1pt}).
    \end{align*}
    This shows~\eqref{eq:standard_MPC_performance_bound} with $\delta_1(N) = \sum_{i=1}^{m} c_i^{\mathrm{f}} \alpha_{\mathrm{ub}}^{\ell_i}(\delta_i^{\gamma}(\floor*{\frac{N}{2}}))$ and $\delta_2(K) = \sum_{i=1}^{m} c_i^{\mathrm{f}} \alpha_{\mathrm{ub}}^{\ell_i}(\delta_i^{\gamma}(K-1))$.
\end{proof}
\subsection{Proof of Lemma~\ref{lem:uniform_reachability}}
\begin{proof}
    Let $x\in\mathcal{X}_{\widetilde{N}}$, $y_T^{(0)} \in \mathbb{Y}_{\widetilde{N}}(x)$, and define $\bar{y}_T^{(0)} = \argmin_{\tilde{y}_T \in \mathcal{Y}_T^W} \vert y_T^{(0)} \vert_{\tilde{y}_T}$.
    There exists $u\in \mathbb{U}^{\widetilde{N}}(x)$ such that $x_{i,u_i}(\widetilde{N}) \in \mathcal{X}_i^{\mathrm{f}}(r_{T,i}^{(0)}(\widetilde{N}))$.
    Recursively define the input trajectory $u_i^{(0)}(k) = k_i^{\mathrm{f}}(x_{i,u_i^{(0)}}(k, x_{i,u_i}(\widetilde{N})), r_{T,i}^{(0)}(\widetilde{N}+k))$ for $k\in\mathbb{I}_{0:N_1}$ and $N_1 \in \mathbb{N}$, i.e. applying the terminal control law from Assumption~\ref{assm:terminal_ingredients} repeatedly.
    Then, 
    \begin{align*}
        &V_i^{\mathrm{f}}(x_{i,u_i^{(0)}}(k+1, x_{i,u_i}(\widetilde{N})), r_{T,i}^{(0)}(\widetilde{N}+k+1)) 
        \\
        &\stackrel{\eqref{eq:terminal_cost_decrease}}{\le} V_i^{\mathrm{f}}(x_{i,u_i^{(0)}}(k, x_{i,u_i}(\widetilde{N})), r_{T,i}^{(0)}(\widetilde{N}+k)) \\
        &\hspace{1em} - \ell_i(x_{i,u_i^{(0)}}(k, x_{i,u_i}(\widetilde{N})), u_i^{(0)}(k) , r_{T,i}^{(0)}(\widetilde{N}+k))
        \\
        &\stackrel{\mathclap{\eqref{eq:terminal_cost_upper_bound},\,\eqref{eq:stage_cost_lower_and_upper_bound}}}{\le} \hspace{0.5em} c_i^{\mathrm{f}} \alpha_{\mathrm{ub}}^{\ell_i} (\vert x_{i,u_i^{(0)}}(k, x_{i,u_i}(\widetilde{N}))\vert_{r_{T,i}^{(0)}(\widetilde{N}+k)})
        \\
        &\hspace{1em} - \alpha_{\mathrm{lb}}^{\ell_i} (\vert x_{i,u_i^{(0)}}(k, x_{i,u_i}(\widetilde{N})) \vert_{x_{T,i}^{(0)}(\widetilde{N}+k)}) 
        \\
        &\le 
        c_i^{\mathrm{f}} \alpha_{\mathrm{ub}}^{\ell_i} (\delta_{x_i}) - \alpha_{\mathrm{lb}}^{\ell_i} (\frac{c_i^{\mathrm{b}}}{2})       
    \end{align*}
    holds for all $k$ with $\vert x_{i,u_i^{(0)}}(k, x_{i,u_i}(\widetilde{N})) \vert_{x_{T,i}^{(0)}(\widetilde{N}+k)} \ge \frac{c_i^{\mathrm{b}}}{2}$, where the constant $\delta_{x_i}>0$ exists due to compactness of $Z_i$.
    Consequently, $N_1$ can be chosen independently of $r_{T,i}^{(0)}$ such that $\vert x_{i,u_i^{(0)}}(N_1, x_{i,u_i}(\widetilde{N})) \vert_{x_{T,i}^{(0)}(\widetilde{N}_1)} < \frac{c_i^{\mathrm{b}}}{2}$ with $\widetilde{N}_1 = \widetilde{N} + N_1$, due to~\eqref{eq:terminal_cost_decrease} and~\eqref{eq:terminal_cost_upper_bound}.

    First, consider the case $\vert x_{T,i}^{(0)}(\widetilde{N}_1) \vert_{\bar{x}_{T,i}^{(0)}(\widetilde{N}_1)} \le \frac{c_i^{\mathrm{b}}}{2}$ for all $i\in\mathbb{I}_{1:m}$.
    Hence, 
    $\vert x_{i,u_i^{(0)}}(N_1, x_{i,u_i}(\widetilde{N})) \vert_{\bar{x}_{T,i}^{(0)}(\widetilde{N}_1)} 
    \le \vert x_{i,u_i^{(0)}}(N_1, x_{i,u_i}(\widetilde{N})) \vert_{x_{T,i}^{(0)}(\widetilde{N}_1)} + \vert x_{T,i}^{(0)}(\widetilde{N}_1) \vert_{\bar{x}_{T,i}^{(0)}(\widetilde{N}_1)} 
    < c_i^{\mathrm{b}}$.
    Thus, from~\eqref{eq:terminal_non_empty_interior}, $x_{i,u_i^{(0)}}(N_1, x_{i,u_i}(\widetilde{N})) \in \mathcal{X}_i^{\mathrm{f}}(\bar{r}_{T,i}^{(0)}(\widetilde{N}_1))$.
    
    Second, consider the case $\vert x_{T,j}^{(k)}(\widetilde{N}_1) \vert_{\bar{x}_{T,j}^{(k)}(\widetilde{N}_1)} > \frac{c_{\mathrm{lb}}^{\mathrm{b}}}{2}$ for a $j\in\mathbb{I}_{1:m}$ with $c_{\mathrm{lb}}^{\mathrm{b}} = \min_{i} c_i^{\mathrm{b}}$,
    We start with $k=0$ but will consider this case multiple times.
    From Assumption~\ref{assm:better_cooperation_candidate}, there exists a cooperation output with a lower cooperative cost than $y_T^{(k)}$; denote it by $y_T^{(k+1)}$.
    This yields a sequence of cooperation outputs as indexed by $k$.
    Define $L = \max_i(L_{x,i}, L_{u,i})$. Then,
    \begin{align*}
        &\vert x_{i,u_i^{(k)}}(N_1, x_{i,u_i}(\widetilde{N})) \vert_{x_{T,i}^{(k+1)}(\widetilde{N}_1)} 
        \\
        &\le \vert x_{i,u_i^{(k)}}(N_1, x_{i,u_i}(\widetilde{N})) \vert_{x_{T,i}^{(k)}(\widetilde{N}_1)} + \vert x_{T,i}^{(k)}(\widetilde{N}_1) \vert_{{x}_{T,i}^{(k+1)}(\widetilde{N}_1)}
        \\
        &< \frac{c_i^{\mathrm{b}}}{2} + \vert r_{T,i}^{(k)} \vert_{{r}_{T,i}^{(k+1)}} 
        \stackrel{\mathclap{\text{Assm.~\ref{assm:unique_corresponding_equilibrium}}}}{\le} L\vert y_{T,i}^{(k)} \vert_{{y}_{T,i}^{(k+1)}} + \frac{c_i^{\mathrm{b}}}{2} 
        \\
        &\le \frac{c_i^{\mathrm{b}}}{2} + L\smash[b]{\sum_{i=1}^{m}}\vert y_{T,i}^{(k)} \vert_{{y}_{T,i}^{(k+1)}} \stackrel{\eqref{eq:better_cooperation_candidate_a}}{\le} 
        \frac{c_i^{\mathrm{b}}}{2} + L\theta c_{\psi}\gamma_{\psi}
    \end{align*}
    with $\gamma_{\psi} = \sup_{y_T\in \mathcal{Y}_T}\psi(y_T)$.
    Hence, if $\theta \le c_{\mathrm{lb}}^{\mathrm{b}}(2Lc_{\psi}\gamma_{\psi})^{-1}$, then from~\eqref{eq:terminal_non_empty_interior}, $ x_{i,u_i^{(k)}}(N_1, x_{i,u_i}(\widetilde{N})) \in \mathcal{X}_i^{\mathrm{f}}(x_{T,i}^{(k+1)}(\widetilde{N}_1))$.
    Furthermore, since $\psi$ is a continuous function on $\mathcal{Y}_T$ and positive definite with respect to $\mathcal{Y}_T^W$, there exists $\eta_{\psi}\in \mathcal{K}$ such that $\psi(y_T) \ge \eta_{\psi}(\vert y_T \vert_{\mathcal{Y}_T^W})$.
    With Assumption.~\ref{assm:unique_corresponding_equilibrium}, $\vert y_T^{(k)} \vert_{\mathcal{Y}_T^W} = \vert y_T^{(k)} \vert_{\bar{y}_T^{(k)}} \ge \sum_{i=1}^{m} \frac{\vert r_{T,i}^{(k)} \vert_{\bar{r}_{T,i}^{(k)}}}{L} \ge \frac{1}{L}\vert x_{T,j}^{(k)}(\widetilde{N}_1) \vert_{\bar{x}_{T,j}^{(k)}(\widetilde{N}_1)} > \frac{c_{\mathrm{lb}}^{\mathrm{b}}}{2L}$.
    From~\eqref{eq:better_cooperation_candidate_b}, this entails $W^{\mathrm{c}}(y_T^{(k+1)}) < W^{\mathrm{c}}(y_T^{(k)}) - \theta \eta_{\psi}(\frac{c_{\mathrm{lb}}^{\mathrm{b}}}{2L})^{\omega}$.
    Now, since $W^{\mathrm{c}} - W_0^{\mathrm{c}}$ is continuous and positive definite with respect to $\mathcal{Y}_T^W$, we also have from Definition~\ref{def:COF},
    $W^{\mathrm{c}}(y_T) \ge W^{\mathrm{c}}(y_T) - W_0^{\mathrm{c}} \ge \eta_{\mathrm{c}}(\vert y_T \vert_{\mathcal{Y}_T^W})$
    with some $\eta_{\mathrm{c}} \in \mathcal{K}$.
    In addition, $W^{\mathrm{c}}(y_T) \le \alpha_{\mathrm{ub}}^{\mathrm{c}}(\gamma_{\mathrm{c}})$ with $\gamma_{\mathrm{c}} = \sup_{y_T \in \mathcal{Y}_T} \vert y_T \vert_{\mathcal{Y}_T^{\mathrm{c}}}$.
    Since $\vert y_T^{(k+1)} \vert_{\mathcal{Y}_T^W} \ge \frac{1}{L}\vert x_{T,j}^{(k+1)}(\tau) \vert_{\bar{x}_{T,j}^{(k+1)}(\tau)}$ for all $\tau\in\mathbb{I}_{0:T-1}$, we have with $N_2\in\mathbb{N}$ that $ \eta_{\mathrm{c}}\Big(\frac{\vert x_{T,j}^{(k+1)}(\tau) \vert_{\bar{x}_{T,j}^{(k+1)}(\tau)}}{L}\Big) < \alpha_{\mathrm{ub}}^{\mathrm{c}}(\gamma_{\mathrm{c}}) - \sum_{k=0}^{N_2-1}\theta\eta_{\psi}\Big(\frac{c_{\mathrm{lb}}^{\mathrm{b}}}{2L}\Big)^{\omega}$.
    Thus, $N_2$ can be chosen such that $\vert x_{T,j}^{(N_2+1)}(\tau) \vert_{\bar{x}_{T,j}^{(N_2+1)}(\tau)} \le \frac{c_{\mathrm{lb}}^{\mathrm{b}}}{2}$ for some $\tau\in\mathbb{I}_{0:T-1}$.
    Hence, the second case holds at most $N_2+1$ times until the first case holds. 
    Thus, with $\hat{y}_T = \bar{y}_T^{(N_2+1)}$ and $\widehat{N} = \widetilde{N}_1 + N_2+1$, there exists $\hat{u} \in \mathbb{U}^N(x)$ by concatenating the previously described input trajectories together such that
    $x_{i,\hat{u}}(\widehat{N}, x_i) \in \mathcal{X}_i^{\mathrm{f}}(\hat{r}_{T,i}(\tau))$ for some $\tau\in\mathbb{I}_{0:T-1}$ and all $N\ge\widehat{N}$.
\end{proof}

\bibliographystyle{IEEEtran}
\bibliography{10_references}

\end{document}